\documentclass{article}
\usepackage[preprint]{neurips_2026}

\usepackage[american]{babel}

\usepackage{natbib} 
\usepackage{mathtools} 
\usepackage{booktabs} 
\usepackage{subfigure}
\usepackage{subcaption}
\usepackage{amsmath, empheq}
\usepackage{amssymb}
\usepackage{mathtools}
\usepackage{amsthm}
\usepackage{tabularx}
\usepackage{array}
\usepackage{longtable}
\usepackage{xcolor}
\usepackage{enumitem}
\usepackage{stackengine}
\usepackage{todonotes}
\usepackage{makecell}
\usepackage{titletoc}
\usepackage{array} 
\usepackage{makecell}
\usepackage[ruled, vlined, linesnumbered]{algorithm2e}
\usepackage{hyperref}

\newcommand{\ignore}[1]{}

\newtheorem{theorem}{Theorem}[section]
\newtheorem{proposition}[theorem]{Proposition}
\newtheorem{lemma}[theorem]{Lemma}
\newtheorem{corollary}[theorem]{Corollary}
\newtheorem{definition}[theorem]{Definition}
\newtheorem{assumption}[theorem]{Assumption}

\DeclareMathOperator*{\argmax}{arg\,max}

\newcommand{\scen}[1]{\texttt{#1}}
\newcommand{\indep}{\mathrel{\perp\!\!\!\perp}}

\title{Detecting Changes in Causal Dependence \\ with Kernels and Copulas}

\author{%
    Shakeel Gavioli-Akilagun\thanks{Corresponding author.} \\
    Department of Decision Analytics and Operations \\
    City University of Hong Kong \\
    Hong Kong, China \\
    \texttt{sgavioli@cityu.edu.hk} 
  \And
  Kieran Wood \\
    Oxford-Man Institute of Quantitative Finance \\
    University of Oxford \\
    Oxford, United Kingdom \\
  \texttt{kieran.wood@eng.ox.ac.uk}
  \And 
  Francesco Quinzan \\
The Department of Engineering Science \\
    University of Oxford \\
    Oxford, United Kingdom \\
    \texttt{francesco.quinzan@eng.ox.ac.uk}
}

\begin{document}

\maketitle

\begin{abstract}
We propose a framework for determining whether the causal dependence of an outcome $Y$ on a covariate $X$ changes at a given time point, given confounders $\boldsymbol{Z}$. For instance, in financial markets, the effect of a market indicator on asset returns may causally change over time. While many existing measures of association can be used to detect changes in joint and marginal distributions, in the absence of strong assumptions on the data generating process none are suitable for detecting changes in the causal mechanism or in the strength of causal relationship. In this work we approach the problem from a fully non-parametric perspective, and treat the causal mechanism as well as the distribution of the data as unknown. We introduce a quantity based on the integrated difference between kernel mean embeddings of certain conditionals copula, which is provably equal to zero if the causal dependence does not change and strictly positive else. A near-linear time estimator for the quantity is proposed, with rates of convergence explicitly spelled out. Extensive experiments demonstrate that the proposed statistic achieves high accuracy on multiple synthetic and real-world datasets. We additionally show how the proposed statistic can be used for change point detection when the goal is to detect changes in causal dependence occurring at an unknown times. 

\end{abstract}

\section{Introduction}
\label{section_introduction}
Modern machine learning systems are often deployed in environments whose underlying structure evolves over time. As a result, relationships that were once reliable for prediction or decision-making may cease to hold after a structural change. For example, in financial markets, the relationship between a market indicator and an asset return may change across regimes \citep{gu2020empirical}; in recommendation or control systems, signals that were previously informative may lose their relevance as user behavior or system dynamics shift \citep{DBLP:conf/aistats/CheungSZ19}. Detecting such changes is crucial, as failing to do so can lead to degraded performance or incorrect downstream decisions.

This challenge has motivated a large body of work on change point detection, which aims to identify time points at which the data-generating process changes \cite{Basseville1993}. Classical and modern approaches typically focus on detecting shifts in marginal distributions, moments, or measures of statistical association, and have been successfully applied in a wide range of settings (e.g., \citet{f6bf4042-8440-3ce2-87e3-78f31cedac7c,DBLP:conf/nips/HarchaouiBM08,DBLP:conf/icml/HockingRB15}). However, many real-world applications are not only concerned with \emph{whether} the data distribution changes, but with \emph{how} the relationship between variables changes. This distinction motivates the study of causal change points, where the object of interest is not a change in the joint distribution, but a change in the causal mechanism linking variables.

In this work we study the problem of determining whether the causal mechanism linking an outcome variable $Y$ and a covariate $X$ has changed at a given point in time, in the presence of confounders $\boldsymbol{Z}$. In keeping with the modern literature on non-parametric change point detection \citep{mcgonigle2025nonparametric, matteson2014nonparametric}, we allow both the joint and marginal distributions of the data to change. Since it is often unrealistic to assume knowledge of the causal mechanism \citep{azadkia2021fast}, we approach the problem from a fully non-parametric perspective, and treat both the causal mechanism and the distribution of the data as unknown. Consequently, we introduce a new kernel based statistic whose population counterpart is equal to zero if the causal dependence of $Y$ on $X$ given $Z$ does not change, and is strictly positive else.

\subsection{Related work}

There exists a rich literature on causal structure learning. Existing approaches mainly rely on learning structural equation models (\citealt{buhlmann2014cam, imoto2001estimation} -- SEMs), see Section~\ref{section: causal structure learning} for a precise definition, or computing measures of (conditional) statistical association \cite{spirtes1991algorithm, pearl1995theory}. However, both approaches are inadequate for the problem at hand, for several reasons: first, the causal effect of $X$ on $Y$ may change due to a change in the marginal distribution of the confounders even though the structural equations have not changed; second, many modern measures of statistical association \cite{szekely2014partial, zhenpartial2022}, though fully non-parametric, are not invariant to the marginal distributions of the data, and can therefore confuse a change in the marginal distributions with a change in the causal mechanism; finally, some measures \cite{chatterjee2021new, azadkia2021simple} are non-negative, and are therefore blind to changes where the magnitude of the causal effect is unaltered while the sign changes.

Previous works \cite{ghoshal2019direct, wang2018direct} have tackled the problem of detecting differences between causal graphs (see Section~\ref{section: causal structure learning}) given samples from each graph under the assumption that the data are generated by a linear SEM. \cite{chen2024identifying, saeed2020causal} consider the same problems for two or more related datasets. \cite{chen2023iscan, li2023kernel} extend the previous approaches to non-linear SEMs driven by additive noise. A handful of recent works study change point detection from a causal perspective, where the goal is to recover unknown points in time where the causal mechanism changes. \cite{huang2024} propose a method for recovering causal change point locations when the data are generated by linear SEMs. \cite{xu2025quickest} propose an algorithm for sequentially observed data, but require that the data be jointly and marginally Gaussian. \cite{DBLP:conf/aistats/GaoAYRK25} combine constraint-based causal discovery with conditional distribution testing, and therefore do not impose any structure on the causal mechanism. However, their methodology does not allow for changes in the marginal distribution of the data. While these works represent important first steps, they either rely strong assumptions on the SEMs or detect changes only after conditioning on estimated parent sets. Therefore, causal change point detection in general non-linear settings remains important an open problem.

\subsection{Notation}


Throughout the paper we make use of the following notations: for $n \in \mathbb{N}$ we write $[n] = \{ 1, \dots, n \}$; $\mathcal M_1^+\!\left( \mathbb{R}^d \right)$ denotes the set of all Borel probability measures on $\mathbb{R}^d$; for $\mu \in \mathcal{M}_+^1(\mathbb{R}^d)$  we write $\operatorname{supp} (\mu) = \overline{ \{ A \in \sigma (\mathbb{R}^d) \mid \mu (A) > 0\}}$ for its support, where $\overline{A}$ denotes the closure of a set $A$ and $\sigma (\mathbb{R}^d)$ is the Borel $\sigma$-algebra on $\mathbb{R}^d$; for two measures $\mu$ and $\nu$ we write $\mu \gg \nu$ if $\nu (A) = 0$ for every set $A$ for which $\mu(A) = 0$, and we write $\mu \sim \nu$ if $\mu \gg \nu$ and $\nu \gg \mu$; for an event $E$ we put $\boldsymbol{1}_{\{ E \}}$ for the indicator function taking value $1$ if $E$ occurs and $0$ else; for sequences $\{ a_n \}_{n \in \mathbb{N}}$ and $\{ b_n \}_{n \in \mathbb{N}}$ we write $a_n = \mathcal{O}(b_n)$ if $|a_n| \leq C |b_n|$ for some absolute $C>0$ and all $n$, and we write $a_n \asymp b_n$ if $a_n = \mathcal{O}(b_n)$ and $b_n = \mathcal{O} (a_n)$; for two random variables $W_1,W_1$ we denote their joint, marginal, and conditional probabilities by $\mathbb{P}_{W_1,W_2}$, $\mathbb{P}_{W_1}$, and $\mathbb{P}_{W_1 | W_2 = w}$ for $w \in \operatorname{supp} (\mathbb{P}_{W_2})$, respectively.

\section{Preliminaries} \label{section: preliminaries}

In this section we provide brief overviews of causal structure learning, copula theory, and kernel mean embeddings. Consequently, we introduce several definitions which will be of use throughout the remainder of the paper. 

\subsection{Causal Structure learning} \label{section: causal structure learning}

We briefly recall important concepts from causal structure learning. The following definition formalizes the notion of causal dependence. 

\begin{definition}[Causal model]
The triple $M = \left < \boldsymbol{U}, \boldsymbol{W}, \boldsymbol{f} \right >$, with $ \boldsymbol{W} = \{W_1, \dots, W_d \}$ being observable variables, $\boldsymbol{U} = \{ U_1, \dots, U_d \}$ begin unobserved exogenous noise, and $\boldsymbol{f} = \{ f_1, \dots, f_j \}$ being a set of functions is called a causal model if it holds that $W_i = f_i \left ( \text{PA}_i, U_i \right )$ for each $\in [d]$, where $\text{PA}_i \subseteq \boldsymbol{W} \setminus W_i$ is the smallest subset of $\boldsymbol{W}$ such that each $W_i$ is independent of $W$'s not in $\text{AP}_i$ given $W$'s in $\text{PA}_i$.  
\label{definition: causal model}
\end{definition}

The $f$'s in Definition~\ref{definition: causal model} are often referred to as structural equations \citep{pearlj}. Every causal model $M$ can be associated with a so called causal graph having vertices $\{ 1, \dots, d\}$, constructed by drawing 
an edge from each vertex $i \in [d]$ to the elements of $\text{PA}_i$. 

Structural causal models provide formal semantics for reasoning about how a system would respond to deliberate changes of the data generation process -- formally, interventions. In this work, we consider \emph{perfect} (or \emph{hard}) interventions \citep{pearlj}; a perfect intervention on a variable $W_j$ sets it to a fixed value $w_j$ and replaces the corresponding structural equation with the constant assignment $W_j := w_j$. This operation is denoted by $do(W_j = w_j)$, and the resulting distribution is written as $P(\boldsymbol{W} \mid do(W_j = w_j))$. Importantly, this post-interventional distribution generally differs from the conditional distribution $P(\boldsymbol{W} \mid W_j = w_j)$, since intervening removes all incoming causal influences into $W_j$, whereas conditioning does not. In general, interventional distributions cannot be expressed in terms of 
observational conditional distributions. The do-calculus introduced by 
Pearl provides graphical rules that characterize when such a 
reduction is possible. We refer the reader to \citet{pearlj} for a discussion of these graphical rules.

\subsection{Conditional copulas}

In this section we introduce the copula and conditional copula. Given a random vector $\boldsymbol{W} = (W_1, \dots, W_d)^\top \in \mathbb{R}^d$ with continuous marginals $F_{W_i} (x) = \mathbb{P} (W_i \leq x)$ for each $i \in [d]$ we recall that the associated copula $\mathcal{C} (\cdot)$ is defined as the distribution function of the random vector
\begin{equation*}
    \boldsymbol{U} = \left ( F_1 (W_1), \dots, F_d (W_d) \right )^\top.  
\end{equation*}
It can be readily seen that $\boldsymbol{U}$ has standard Uniform marginals, and the copula is therefore invariant to the marginal distribution of the $W$'s. By Sklar's theorem \citep{sklar1959fonctions, nelsen2006introduction} $\mathcal{C} (\cdot)$ is unique under the assumption of continuous marginals, and moreover
$\boldsymbol{W}$'s distribution function can be decomposed as a composition of the marginals and the copula via
\begin{align}
    & F_{\boldsymbol{W}} \left ( x_1, \dots, x_d \right ) = \mathcal{C} \left ( F_{W_1} (x_1), \dots, F_{W_d} (x_d) \right ), \quad \forall x_i \in \mathbb{R}, i \in [d]. 
    \label{equation: copula}
\end{align}
Indeed, the copula characterizes all dependencies between the components of $\boldsymbol{W}$ \citep{nelsen2006introduction}.  

The following definition introduces the conditional copula, which was first studied by \cite{gijbels2011conditional} and \cite{veraverbeke2011estimation}. 
\begin{definition}[Conditional copula]
Given a random vector $\boldsymbol{W} = (W_1, W_2, \boldsymbol{W}_3^\top)^\top$ with continuous marginals, $W_1, W_2 \in \mathbb{R}$, and $\boldsymbol{W}_3 \in \mathbb{R}^d$ for each $\boldsymbol{w} \in \operatorname{supp} (\mathbb{P}_{\boldsymbol{W}_3})$ there is a unique distribution function $\mathcal{C}_{\boldsymbol{w}} \left ( \cdot \right )$, with standard uniform marginals, for which the distribution function of $(W_1, W_2)^\top$ conditional on  $\boldsymbol{W}_3 = \boldsymbol{w}$ can be decomposed as 
\begin{align*}
& F_{W_1, W_2 \mid \boldsymbol{W}_3 = \boldsymbol{w}} (x_1, x_2) \hspace{0.5em} = \mathcal{C}_{\boldsymbol{w}} \left ( F_{W_1 \mid \boldsymbol{W}_3 = \boldsymbol{w}} (x_1), F_{W_2 \mid \boldsymbol{W}_3 = \boldsymbol{w}} (x_2) \right ), \hspace{0.5em} \forall x_1, x_2 \in \mathbb{R}. 
\end{align*}
$\mathcal{C}_{\boldsymbol{w}} (\cdot)$ is called the conditional copula of the random vector $(W_1, W_2)^\top$ given $\boldsymbol{W}_3 = \boldsymbol{w}$. 
\label{definition: conditional copula}
\end{definition}

Similar to \eqref{equation: copula}, the copula $\mathcal{C}_{\boldsymbol{w}} (\cdot)$ in Definition~\ref{definition: conditional copula} fully characterizes the conditional dependence structure between the random variables $W_1$ and $W_2$ conditional on $\boldsymbol{W}_3 = \boldsymbol{w}$. 

\subsection{Kernel Mean Embeddings}

Finally, we recall the definition of the kernel mean embedding (KME) of conditional and unconditional probability measures. The following definition describes the main concepts pertaining KMEs and refer the reader to \cite{Muandet17:KME} and \cite{klebanov2020rigorous} for details. 

\begin{definition}[Kernel Mean Embedding]
\label{def:KME}
Let $K : \mathbb{R}^d \times \mathbb{R}^d \to \mathbb{R}$ be a positive definite kernel\footnote{Kernels may of course be defined on more general spaces. Here we state the definition in terms of $\mathbb{R}^d$ as the present work only considers Euclidean data.}, and let $\mathcal{H}_K$ be the corresponding reproducing kernel Hilbert space (RKHS) of functions on. To any $\mathbb{P} \in \mathcal{M}_1^+$, one can associate the kernel mean embedding 
\begin{equation}
    \mu_K(\mathbb{P}) = \int_{\mathbb{R}^d} K(\cdot, \boldsymbol{w}) \mathrm{d} \mathbb{P} (\boldsymbol{w}) \in \mathcal{H}_K,
    \label{equation: KME}
\end{equation}
where the integral is meant in Bochner's sense \citep[Chapter~II.2]{diestel77vector}. Similarly given two random variables $W_1, W_2$ and some $w \in \operatorname{supp} (\mathbb{P}_{W_2})$, one can embed the conditional distribution $\mathbb{P}_{W_1 \mid W_2 = w}$ in into $\mathcal{H}_K$ in an identical fashion, and the resulting object is referred to as the conditional mean embedding. 
\end{definition}

Continuity and boundedness of the kernel are sufficient conditions for the existence of \eqref{equation: KME}. When the map $\mu_K : \mathcal{M}_1^+(\mathbb{R}^d) \mapsto \mathcal{H}_K$ is invective the kernel is said to be characteristic, and for two random variables $W_1, W_2$ one has that $\mathbb{P}_{W_1} \neq \mathbb{P}_{W_2} \iff \mu_K (\mathbb{P}_{W_1}) \neq \mu_K (\mathbb{P}_{W_2})$. Many commonly used kernels on $\mathbb{R}^d$ enjoy this property, including the Gaussian, Laplacian, $B_{2l+1}$-spline, and inverse multiquadratic kernels \citep{sriperumbudur2010hilbert}. 

\section{Methodology} \label{section: methodology}

In this section we formally state the problem of detecting a generic change in the causal relationship between two variables $X$ and $Y$ given confounders $\boldsymbol{Z}$. Consequently we introduce a statistic for resolving this problem, show how it may be estimated in a finite sample setting, and establish its consistency and rate of convergence of the estimator. 

\subsection{Problem statement} \label{section: problem statement}

We study the following problem. Given a sample $\{ (X_i, Y_i, \boldsymbol{Z}_i) : i \in [n] \}$ with $X,Y \in \mathbb{R}$ and $\boldsymbol{Z} \in \mathbb{R}^d$, and an integer $\eta \in [n-1]$, we seek to determine whether the causal dependence of $Y$ on $X$ given $\boldsymbol{Z}$ is the same for data indexed on or before $\eta$ as for data indexed after $\eta$. Denote by $\mathbb{P}_{X,Y,\boldsymbol{Z}}$ and $\mathbb{P}_{X',Y',\boldsymbol{Z}'}$ the distribution of a triple $\{ X_i, Y_i, \boldsymbol{Z}_i \}$ with $i \leq \eta$ and $i > \eta$ respectively. In the language of Section~\ref{section: causal structure learning}, our goal is to determine if it holds that 
\begin{align}
& \mathbb{P}_{X,Y,\boldsymbol{Z}} (Y = y \mid do(X = x, \boldsymbol{Z} = \boldsymbol{z})) = \mathbb{P}_{X',Y',\boldsymbol{Z}'} (Y' = y \mid do(X' = x, \boldsymbol{Z}' = \boldsymbol{z})).  
\label{equation: do equivalence}
\end{align}

As mentioned in Section \ref{section: causal structure learning}, reasoning about \eqref{equation: do equivalence} is non-trivial. This is because in general, it is not possible to estimate interventional distributions as in \eqref{equation: do equivalence} from observational samples. To circumvent this problem, 
we impose the following assumptions:

\begin{assumption}[Regression; $Y$ parent set; equal support]
It holds that (i) $Y \not \in \text{PA}_X \cup \text{PA}_{\boldsymbol{Z}}$, (ii) $ \mathsf{Pa}_Y \subseteq X \cup \boldsymbol{Z}$.\footnote{Note that if Y has no causal parents, then the parent set is the empt set and the assumption automatically holds.}, and (iii) $\mathbb{P}_{\boldsymbol{Z}} \sim \mathbb{P}_{\boldsymbol{Z}'} $. 
\label{assumption: regression; parent; support}
\end{assumption}

\begin{assumption}[Causal sufficiency]
There are no hidden confounders that are causing more than one variable in $\{ X_i, Y_i, \boldsymbol{Z}_i \}$, for all $i\in [n]$.
\label{assumption: causal sufficiency}
\end{assumption}

The above assumptions are quite mild, and aligned to prior work; see \citep{DBLP:journals/jmlr/Spirtes10,DBLP:conf/icml/QuinzanSJRB23,DBLP:journals/corr/abs-2311-06012}. Part (i) of Assumption~\ref{assumption: regression; parent; support} states that $Y$ does not cause $X$ or $\boldsymbol{Z}$, and part (ii) states that the causal parents of $Y$ must be included in the variables $X\cup \boldsymbol Z$. Part (iii) of Assumption~\ref{assumption: regression; parent; support} is satisfied for a wide range of scenarios; notably, the assumption is satisfied when both $\mathbb{P}_{\boldsymbol{Z}}$ and $\mathbb{P}_{\boldsymbol{Z}'}$ have full support $\mathbb{R}^d$. Note that if one could find a set $A \in \operatorname{supp}(P_{\boldsymbol{Z}})$ with $A \not\in \operatorname{supp}(\mathbb{P}_{\boldsymbol{Z'}})$ it would not be meaningful to ask whether the causal effect of $X$ on $Y$ has changed given that the confounders take values in $A$. Finally, Assumption~\ref{assumption: causal sufficiency} is also common in the causal inference literature; see for instance \citep{DBLP:journals/jmlr/Spirtes10}. 

We remark that the first two parts of Assumption~\ref{assumption: regression; parent; support} are necessary to ensure that the interventional quantity of interest can be expressed in terms of observational conditional distributions (see Lemma \ref{theorem: copula causality}). If these assumptions fail, then the interventional distribution is not uniquely determined by the observational distribution, making \eqref{equation: do equivalence} impossible to verify from observational data. Counterexamples are provided in Appendix \ref{appendix: assumptions neessity}.

\subsection{Causal change detection using kernels and copulas}

We now explain how the testing problem implied by \eqref{equation: do equivalence} may be resolved using conditional copulas and kernel mean embeddings introduced in Section~\ref{section: preliminaries}. The following lemma demonstrates that, under the assumptions stated in Section~\ref{section: problem statement}, when reasoning about cause-effect relationships it is sufficient to study the associated family of conditional copulas. 
\begin{theorem}[Causal reasoning through copulas]
Let $\mathcal{C}_{Z,Y \mid \boldsymbol{Z} = z}$ be the conditional copula of the random vector $(X,Y)^\top$ conditional on $\boldsymbol{Z} = \boldsymbol{z}$, and let $\mathcal{C}_{X'.Y' \mid \boldsymbol{Z'} = \boldsymbol{z}}$ be defined analogously. Grant Assumption \ref{assumption: regression; parent; support} and \ref{assumption: causal sufficiency} hold. Then, $\mathbb{P}_{X,Y,\boldsymbol{Z}} (Y = y \mid do(X = x, \boldsymbol{Z} = \boldsymbol{z}))  =  \mathbb{P}_{X',Y',\boldsymbol{Z}'} (Y' = y \mid do(X' = x, \boldsymbol{Z}' = \boldsymbol{z}))$ for all $y$ and $\boldsymbol{z}$ implies that for all $(u,v) \in [0, 1]^2$ and all $\boldsymbol{z} \in \operatorname{supp} (\mathbb{P}_Z)$ it holds that $\mathcal{C}_{X, Y \mid \boldsymbol{Z = z} } \left ( u, v \right ) = \mathcal{C}_{X', Y' \mid \boldsymbol{Z' = z}} \left ( u, v \right )$. 
\label{theorem: copula causality}
\end{theorem}
The proof, given in Appendix \ref{appeendix : copula causality}, uses the structural assumptions together with the do-calculus to establish that the interventional quantity can be reduced to an observational conditional expectation. In light of Theorem~\ref{theorem: copula causality}, determining whether \eqref{equation: do equivalence} holds is equivalent to testing apart the following hypotheses:  
\begin{align}
& H_0: \mathcal{C}_{X, Y \mid \boldsymbol{Z = z}} \left ( u, v \right ) = \mathcal{C}_{X', Y' \mid \boldsymbol{Z'} = \boldsymbol{z} } \left ( u, v \right ),  \quad  \forall (u,v) \in [0,1]^2, \boldsymbol{z} \in \operatorname{supp}(\mathbb{P}_{\boldsymbol{Z}}) \nonumber \\
& H_1: \exists (u,v) \in  [0,1]^2, \boldsymbol{z} \in \mathbb{R}^p \text{ s.t. } \mathcal{C}_{X, Y \mid \boldsymbol{Z} = \boldsymbol{z}} \left ( u, v\right ) \neq \mathcal{C}_{X', Y' \mid \boldsymbol{Z'} = \boldsymbol{z} } \left ( u, v \right ).
\label{equation: dependence testing}
\end{align}

Consequently, we propose to compare the kernel mean embeddings of the conditional copulas of the $X,Y$-s and $X',Y'$-s conditional on the events $\{ \boldsymbol{Z} =\boldsymbol{z} \}$ and $\{ \boldsymbol{Z'} = \boldsymbol{z} \}$ for all $\boldsymbol{z}$-s in the sample space of the $\boldsymbol{Z},\boldsymbol{Z'}$-s. For some non-negative measure $\Lambda$ such that $\Lambda \sim \mathbb{P}_{\boldsymbol{Z}} \mathbb{P}_{\boldsymbol{Z'}}$, and some positive definite kernel $K : [0,1]^2 \times [0,1]^2 \mapsto \mathbb{R}$, introduce the quantity
\begin{align}
& Q_{X,Y,Z}^{(\Lambda)} \coloneqq Q_{X,Y,Z} = \int_{\mathbb{R}^d} \left \| \mu_{K} \left ( \mathcal{C}_{X,Y \mid \boldsymbol{Z} = \boldsymbol{z}} \right ) - \mu_{K} \left ( \mathcal{C}_{X', Y' \mid \boldsymbol{Z'} = \boldsymbol{z}} \right ) \right \|_{\mathcal{H}_K}^2 \mathrm{d} \Lambda (\boldsymbol{z}), \label{equation: conditional copula MMD}
\end{align}
where $\| \cdot \|_{\mathcal{H}_K}$ is the Hilbert space norm induced by the kernel $K$. Observe that \eqref{equation: conditional copula MMD} can also be interpreted as an integral over maximum mean discrepancies\footnote{For two probability measures $\mathbb{P}$ and $\mathbb{Q}$ defined on some space $\mathcal{X}$, and a positive definite kernel $K : \mathcal{X} \times \mathcal{X} \mapsto \mathbb{R}$, their maximum mean discrepancy (induced by $K$) is given by $\operatorname{MMD}_K[\mathbb{P}, \mathbb{Q}] = \| \mu_K (\mathbb{P}) - \mu_K (\mathbb{Q}) \|_{\mathcal{H}_K}$.} \cite[MMD]{gretton2012mmd} between the conditional copulas of $(X,Y) | \boldsymbol{Z} = \boldsymbol{z}$ and $(X',Y') | \boldsymbol{Z'} = \boldsymbol{z}$. 

The following lemma establishes that $Q_{X,Y,Z}$ separates the null and alternative hypotheses in \eqref{equation: dependence testing} and is therefore appropriate for resolving \eqref{equation: do equivalence}.
\begin{proposition}[Consistency of $Q_{X,Y,Z}$]
Grant Assumption \ref{assumption: regression; parent; support} and \ref{assumption: causal sufficiency} hold. For any characteristic kernel $K$ and any positive measure $\Lambda$ for which $\mathbb{P}_{\boldsymbol{Z}} \ll \Lambda$ it holds that (i) if there exists a $\mathbb{P}_{\boldsymbol{Z}}$-measurable set $A \subseteq \mathbb{R}^d$ such that \eqref{equation: do equivalence} does not hold for some $x$ and all $\boldsymbol{z} \in A$ then $Q_{X,Y,Z} > 0$, and (ii) if \eqref{equation: do equivalence} holds for all $\boldsymbol{z} \in \operatorname{supp} (\mathbb{P}_{\boldsymbol{Z}})$ and all $x$ then $Q_{X,Y,Z} = 0$. 
\label{proposition: Q separation}
\end{proposition}
Proposition~\ref{proposition: Q separation} holds for any positive measure whose support contains the supports of  $\mathbb{P}_{\boldsymbol{Z}}, 
\mathbb{P}_{\boldsymbol{Z}'}$, which paves the way for choosing $\Lambda$ depending on the structure of the problem at hand. However, in the absence of additional information a natural candidate is $\Lambda = \frac{1}{2} \mathbb{P}_Z + \frac{1}{2} \mathbb{P}_{Z'}$, and in the sequel we stick to this choice. 

\subsection{Estimation via $k_n$-nearest neighbours} \label{section: estimation}

In this section we explain how \eqref{equation: conditional copula MMD} may be estimated from a finite sample. The main difficulty in obtaining an estimator arises from the need to estimate kernel mean embeddings of conditional copulas for all $\boldsymbol{z}$'s in $\operatorname{supp}(\mathbb{P}_{\boldsymbol{Z}})$. To that end, observe that $Q_{X,Y,Z}$ may be equivalently written as
\begin{align}
Q_{X,Y,Z} & = \int_{\mathbb{R}^d} \mathbb{E}_{U_1, U_2 \sim \mathcal{C}_{X,Y \mid \boldsymbol{Z} = \boldsymbol{z}} }\left [ K \left ( U_1, U_2 \right ) \right ] \mathrm{d} \Lambda (\boldsymbol{z})  + \int_{\mathbb{R}^d} \mathbb{E}_{U_1', U_2' \sim \mathcal{C}_{X',Y' \mid \boldsymbol{Z'} = \boldsymbol{z}}} \left [ K \left ( U_1' , U_2' \right ) \right ] \mathrm{d} \Lambda (\boldsymbol{z}) \nonumber \\
& \quad - 2 \int_{\mathbb{R}^d} \mathbb{E}_{\substack{U \sim \mathcal{C}_{X,Y \mid \boldsymbol{Z} = \boldsymbol{z}}, \hspace{0.5em} U' \sim \mathcal{C}_{X', Y' \mid \boldsymbol{Z'} = \boldsymbol{z}}}} \left [ K \left ( U, U' \right ) \right ] \mathrm{d} \Lambda (\boldsymbol{z}).
\label{equation: Q expansion}
\end{align}

At a high level, we approximated the inner expectations by U-statistics computed over $k_n$-nearest neighbour--based pseudo-samples, which act as local draws from the conditional copula at a fixed $\boldsymbol{z}$. Denote the samples indexed respectively before or after $\eta$ by $\{ ( X_i, Y_i, \boldsymbol{Z}_i ) :i \in [\eta] \}$ and $\{ ( X_i', Y_i', \boldsymbol{Z}_i' ) :i \in [n - \eta] \}$. For each $\boldsymbol{z} \in \mathbb{R}^d$ let $\mathcal{N} ( {\boldsymbol{z}} )$ and $\mathcal{N}' ( {\boldsymbol{z}} )$ denote the indices of the $k_n$ nearest elements of $\{ \boldsymbol{Z}_1, \dots, \boldsymbol{Z}_\eta \}$ and $\{ \boldsymbol{Z}_1', \dots, \boldsymbol{Z}_{n-\eta}' \}$ to $\boldsymbol{z}$. Therefore, introduce the quantities
\begin{equation*}
\hat{U}_{X_j}^{\boldsymbol{Z}_i} = \frac{1}{k_n} \sum_{l \in \mathcal{N}(\boldsymbol{Z}_i)} \mathbf{1}_{\left \{ X_l \leq X_j \right \}}, \quad \hat{U}_{Y_j}^{\boldsymbol{Z}_i} =  \frac{1}{k_n} \sum_{l \in \mathcal{N} (\boldsymbol{Z}_i)} \mathbf{1}_{\left \{ Y_l \leq Y_j \right \}}, \quad \hat{U}_{X_j, Y_j}^{\boldsymbol{Z}_i} = \left ( \hat{U}_{X_j}^{\boldsymbol{Z}_i}, \hat{U}_{Y_j}^{\boldsymbol{Z}_i} \right )^\top, 
\end{equation*}
and let $\hat{U}_{X_j', Y_j'}^{\boldsymbol{Z}_i'}$, $\hat{U}_{X_j, Y_j}^{\boldsymbol{Z}_i'}$, and $\hat{U}_{X_j', Y_j'}^{\boldsymbol{Z}_i}$ be defined analogously. In the sequel we will also write 
\begin{equation*}
    \hat{\mathbb{K}}^{\boldsymbol{Z}_i}_{X_j, Y_j, X_k, Y_k} = K ( \hat{U}_{X_j, Y_j}^{\boldsymbol{Z}_i}, \hat{U}_{X_k, Y_k}^{\boldsymbol{Z}_i} ). 
\end{equation*}
With this notation in place, introduce the following quantities

\begin{align*}
& \hat{T}_1 =  \frac{1}{k_n (k_n-1)} \Big ( \frac{1}{\eta} \sum_{i_1 =1}^{\eta} \sum_{\substack{i_2, i_3 \in \mathcal{N} (\boldsymbol{Z}_{i_1})\\ i_2 < i_3}} \hat{\mathbb{K}}_{X_{i_2}, Y_{i_2}, X_{i_3}, Y_{i_3}}^{\boldsymbol{Z}_{i_1}} + \frac{1}{n - \eta} \sum_{i_1 =1}^{n - \eta} \sum_{\substack{i_2, i_3 \in \mathcal{N} (\boldsymbol{Z}_{i_1}') \\ i_2 < i_3}} \hat{\mathbb{K}}_{X_{i_2}, Y_{i_2}, X_{i_3}, Y_{i_3}}^{\boldsymbol{Z}_{i_1}'} \Big ), \\
& \hat{T}_2 = \frac{1}{k_n(k_n-1)} \Big ( \frac{1}{\eta} \sum_{i_1 =1}^{\eta} \sum_{\substack{i_2, i_3 \in \mathcal{N}'(\boldsymbol{Z}_{i_1}) \\ i_2 < i_3}} \hat{\mathbb{K}}_{X_{i_2}', Y_{i_2}', X_{i_3}', Y_{i_3}'}^{\boldsymbol{Z}_{i_1}} + \frac{1}{n - \eta} \sum_{i_1 =1}^{n - \eta} \sum_{\substack{i_2, i_3 \in \mathcal{N}'(\boldsymbol{Z}_{i_1}') \\ i_2 < i_3}} \hat{\mathbb{K}}_{X_{i_2}', Y_{i_2}', X_{i_3}', Y_{i_3}'}^{\boldsymbol{Z}_{i_1}'} \Big ) , \\
& \hat{T}_3 = \frac{1}{k_n^2} \Big ( \frac{1}{\eta} \sum_{i_1=1}^{\eta} \sum_{\substack{i_2 \in \mathcal{N} (\boldsymbol{Z}_{i_1}) \\ i_3 \in \mathcal{N}' (\boldsymbol{Z}_{i_1})}} \hat{\mathbb{K}}_{X_{i_2}, Y_{i_2}, X_{i_3}', Y_{i_3}'}^{\boldsymbol{Z}_{i_1}} + \frac{1}{n - \eta} \sum_{i_1=1}^{n - \eta} \sum_{\substack{i_2 \in \mathcal{N} (\boldsymbol{Z}_{i_1}') \\ i_3 \in \mathcal{N}'(\boldsymbol{Z}_{i_1}')}} \hat{\mathbb{K}}_{X_{i_2}, Y_{i_2}, X_{i_3}', Y_{i_3}'}^{\boldsymbol{Z}_{i_1}'} \Big ), 
\end{align*}
$\hat{T}_1$, $\hat{T}_2$, and $\hat{T}_3$ are natural plug-in estimators for each of the three terms appearing in \eqref{equation: Q expansion}. With these quantities in place \eqref{equation: conditional copula MMD} can be approximated via $\hat{Q}_{X,Y,Z} = \hat{T}_1 + \hat{T}_2 - \hat{T}_3. $
Regarding the time complexity of the proposed estimator, the setting where $\eta \asymp n$, the cost of computing $\hat{Q}$ is of the order $\mathcal{O} ( n \times \max \{ k_n \log (n),  k_n^2 \})$. This follows from the fact that each required $k_n$-nearest neighbour graph can be computed in $\mathcal{O} (n k_n \log (n))$ time \cite{friedman1977algorithm}. 

\subsection{Consistency and rates}

In this section we establish the consistency and rate of convergence of the nearest neighbour estimator proposed in the previous section. The following assumptions will be required. 

\begin{assumption}[Proportional measures]
$\forall A \in \operatorname{supp} (\mathbb{P}_{\boldsymbol{Z}})$ it holds that $\mathbb{P}(\boldsymbol{Z} \in A) \asymp \mathbb{P} (\boldsymbol{Z}' \in A)$. 
\end{assumption}

\begin{assumption}[Sub-Weibull]
There are constants $\boldsymbol{z}_*,\boldsymbol{z}_*' \in \mathbb{R}^d$ and $\alpha, d_0, T,C >0$ such (i) $\mathbb{P}\left( W \geq t\right) \lesssim e^{-C t^\alpha}$ for $W \in \{ \| \boldsymbol{Z} - \boldsymbol{z}_* \|_2, \| \boldsymbol{Z}' - \boldsymbol{z}_*' \|_2 \} $, and (ii) $\mathbb{P}\left( W \leq r\right) \gtrsim r^{d_0}$ for all $W \in \{ \| \boldsymbol{Z} - \boldsymbol{z} \|_2 : \| \boldsymbol{z} - \boldsymbol{z}_* \|_2 \leq T \} \cup \{ \| \boldsymbol{Z}' - \boldsymbol{z}' \|_2 : \| \boldsymbol{z}' - \boldsymbol{z}_*' \|_2 \leq T \}$.  
\label{assumption: Weibull tails}
\end{assumption}

\begin{assumption}[Lipschitz marginals]
There is a $\beta > 0$ such for any $\boldsymbol{z}_1, \boldsymbol{z}_2 \in \mathbb{R}^d$ it holds that (i) $\| F_{W | \boldsymbol{Z = z_1}} - F_{W | \boldsymbol{Z = z_2}} \|_\infty \lesssim \| \boldsymbol{z_1} - \boldsymbol{z_2}
\|_2^\beta$ for $W \in \{ X, Y \}$ and (ii) $\| F_{W' | \boldsymbol{Z' = z_1}} - F_{W' | \boldsymbol{Z' = z_2}} \|_\infty \lesssim \| \boldsymbol{z_1} - \boldsymbol{z_2}
\|_2^\beta$ for $W' \in \{ X', Y' \}$. 
\label{assumption: Lipschitz marginals}
\end{assumption}

Assumption \ref{assumption: Weibull tails} is needed to control the bias of the nearest neighbour estimator (see Lemma~\ref{lemma: knn expectation bound}), and Assumption \ref{assumption: Lipschitz marginals} is standard when estimating conditional associations in a non-parametric fashion (confer \cite{huang2022kernel} and \cite{azadkia2021simple}). With these assumptions in place, we have the following result. 

\begin{theorem}[rate of convergence]
Grant Assumptions \ref{assumption: regression; parent; support} and \ref{assumption: causal sufficiency}, as well as Assumptions \ref{assumption: Weibull tails} and \ref{assumption: Lipschitz marginals} hold. Assume moreover that the kernel $K$ is continuous with bounded first derivative, $\min(\eta, n - \eta) \asymp n$, and $1 \leq k_n \leq n / 2$. It holds that
\begin{equation*}
\left | \hat{Q}_{Z,Y,Z} - Q_{X,Y,Z} \right | = \mathcal{O}_{\mathbb{P}} \left ( \frac{k_n}{\sqrt{n}} +  \frac{1}{k_n^4} + \left(\frac{k_n}{n}\right)^{\beta / d_0} + \frac{(\log n)^{\beta / \alpha}}{n^2} \right ). 
\end{equation*}
\label{theorem: rate of convergence}
\end{theorem}

Note that $d_0$ in Assumption \ref{assumption: Weibull tails} captures the intrinsic dimensionality of the $\boldsymbol{Z}, \boldsymbol{Z}'$'s, and in practice may be significantly smaller than $d$, the ambient dimension. Theorem \ref{theorem: rate of convergence} therefore reveals that the rate of convergence adapts to the intrinsic dimensionality of the data. In the setting where $d_0 \leq 2 \beta$, choosing $k_n \asymp n^{d_0/ (2\beta + d_0)}$ yields a rate of convergence of the order $n^{-\beta / (2 \beta + d_0)}$ which matches Stone's optimal rate of convergence \citep{stone1980optimal} for conditional density estimation for distributions with bounded $\beta$-th derivative. When $d_0 > 2 \beta$ choosing $k_n \asymp n^{1/10}$ yields a rate of convergence of the order $n^{-2/5}$, which this time matches Stone's optimal rate for conditional density estimation when the first two partial derivatives are bounded. 

\subsection{Inference} \label{section: inference}

While we were unable to establish the limiting distribution of $\hat{Q}_{X,Y,Z}$ under the null of no change in the causal mechanism, the estimator proposed in Section~\ref{section: estimation} can nevertheless be used for inference via a simple permutation test. For some $B \in \mathbb{N}$ let $\hat{Q}_{X,Y,Z}^{(1)}, \dots, \hat{Q}_{X,Y,Z}^{(B)}$ be computed using independent permutations of the observed sample $\{ (X_i, Y_i, \boldsymbol{Z}_i) : i \in [n] \}$. Then, a $p$-value may be computed via:  
\begin{equation}
    p_{\text{perm}} = \frac{1}{B+1} \sum_{b=1}^B \boldsymbol{1}_{\left \{ \hat{Q}_{X,Y,Z} > \hat{Q}_{X,Y,Z}^{(b)} \right \}}. 
\label{equation: perm p-val}
\end{equation}

Since $\hat{Q}_{X,Y,Z}$ exhibits near-linear time complexity, it is generally not too expensive to compute \eqref{equation: perm p-val} in practice. Note that $\hat{Q}_{X,Y,Z}$ can be interpreted as a weighted average of $n$ integral probability metrics, and there is a rich literature on permutation test for integral probability metrics; see \cite{praestgaard1995permutation}. The additional simulation study in Section \ref{section: behaviour under the null} of the appendix demonstrates that \eqref{equation: perm p-val} generates valid p-values for a wide range of data generating mechanisms.

\section{Simulation studies} \label{section: simulation studies}

In this section we evaluate the performance of our proposed statistic on synthetic data. In the sequel all kernel based methods discussed, including our proposed statistic, make use of the Gaussian kernel $K(\boldsymbol{w}_1, \boldsymbol{w}_2) = e^{- \gamma \| \boldsymbol{w}_1 - \boldsymbol{w}_2^2 \|_2}$ with bandwidth $\gamma$ chosen using the median heuristic \citep{garreau2017large}, and when computing our statistic we set the number of nearest neighbours to $30$. Our experiments can be reproduced using code available via \emph{Anonymous GitHub} \href{https://anonymous.4open.science/r/ccd-kern-copulas-1C5E}{here}.

\subsection{Competing methods} \label{section: competing methods}

To our knowledge, no two sample statistic has so far been put forward in the literature for the generic problem of detecting changes in causal dependence in a fully non-parametric fashion. We therefore compare our proposed methodology against modern (and classical) non-parametric one sample measures of conditional dependence. Letting $T_\bullet$ represent a generic one sample statistic, we construct a two sample statistic via $Q \coloneqq |T_{X, Y, \boldsymbol{Z}} - T_{X', Y', \boldsymbol{Z}'}|$, so that large values of $Q$ are indicative of a change across segments. We compare with the rank and distance-based measures of conditional dependence of \citeauthor{azadkia2021simple} (\scen{XI\_AC}) and \citeauthor{sze2007dcor,sze2014pdcor} (\scen{PDCORR}), the Hilbert-Schmidt Conditional Independence Criterion (\scen{CHSIC}) of \citep{fukumizu2007kernel}, the graph-based kernel partial correlation of coefficient (\scen{KPCGRAPH}) of \citeauthor{huang2022kernel}, and classical partial rank correlations \citep[\scen{PSPEAR} and \scen{PKENDALL}]{kendall1942partial}.

\subsection{Simulation results}

We assess the detection power of our proposed statistic and its competitors by computing the area under the curve (AUC) based on $500$ Monte Carlo replications, for sample paths of different data generating mechanisms with $400$ observations from both the pre and post-change distribution. We consider a range of data generating mechanisms, with including linear and non-linear dependence of the conditional mean of $Y$ as well as dependence in higher order moments; the data generating mechanisms are described in detail in Section~\ref{section: data generating mechanisms} of the appendix. The results are presented in Table~\ref{section: competing methods}, which shows that our proposed statistic is highly competitive against the state of the art, in particular when the change affects higher order moments of the response $Y$.

\begin{table*}[h]
\centering
\setlength{\tabcolsep}{4pt}
\renewcommand{\arraystretch}{1.15}
\caption{Synthetic positive (change) scenarios grouped by change type. AUC shown for multiple methods; p-values shown as median p-value with rejection count (k/R) for $\hat{Q}_{X,Y,Z}^{(\Lambda)}$.}
\label{tab:scenarios_positive:POOL_PERM}
\resizebox{\textwidth}{!}{%
\footnotesize
\begin{tabular}{@{} l ccccccc c @{}}
\toprule
\textbf{Test name} & \multicolumn{7}{c}{\textbf{AUC (mean (sd))}} & \makecell[c]{\textbf{p-value}\\\textbf{(k/R)}} \\
\cmidrule(lr){2-8}
 & $\hat{Q}_{X,Y,Z}$ & \scen{XI\_AC} & \scen{CHSIC} & \scen{PDCORR} & \scen{PSPEAR} & \scen{PKENDALL} & \scen{KPCGRAPH} & $\hat{Q}_{X,Y,Z}$ \\
\midrule
\multicolumn{9}{@{}l@{}}{\textbf{Causal graph edge structure: linear edge on (01), nonlinear edge on (02), new driver (03), edge off (04), collinear $X$ (05)}} \\
\scen{PMB01} & 0.947 (0.006) & 0.776 (0.017) & 1.000 (0.000) & 1.000 (0.000) & 1.000 (0.000) & 1.000 (0.000) & 0.547 (0.011) & 0.042 (26/50) \\
\scen{PMB02} & 0.845 (0.010) & 0.510 (0.012) & 0.977 (0.004) & 0.912 (0.011) & 0.961 (0.008) & 0.961 (0.007) & 0.508 (0.002) & 0.027 (30/50) \\
\scen{PMB03} & 0.825 (0.011) & 0.644 (0.015) & 1.000 (0.000) & 1.000 (0.000) & 0.995 (0.002) & 0.998 (0.001) & 0.546 (0.010) & 0.189 (9/50) \\
\scen{PMB04} & 1.000 (0.000) & 0.633 (0.016) & 1.000 (0.000) & 1.000 (0.000) & 1.000 (0.000) & 1.000 (0.000) & 0.508 (0.005) & 0.003 (42/50) \\
\scen{PMB05} & 1.000 (0.000) & 0.838 (0.013) & 1.000 (0.000) & 1.000 (0.000) & 1.000 (0.000) & 1.000 (0.000) & 0.651 (0.018) & 0.002 (50/50) \\
\midrule
\multicolumn{9}{@{}l@{}}{\textbf{Effect: sign flip (01; Gauss 06), strength shift (02; +$Z$ 03; +distract 04; +$Z$ unscaled 05), sparse $X$ (07), Poisson (08)}} \\
\scen{PEF01} & 1.000 (0.000) & 0.488 (0.014) & 1.000 (0.000) & 0.957 (0.007) & 1.000 (0.000) & 1.000 (0.000) & 0.509 (0.008) & 0.002 (50/50) \\
\scen{PEF02} & 0.855 (0.010) & 0.670 (0.018) & 1.000 (0.000) & 0.998 (0.001) & 0.986 (0.003) & 0.997 (0.001) & 0.508 (0.007) & 0.150 (10/50) \\
\scen{PEF03} & 0.910 (0.008) & 0.709 (0.017) & 1.000 (0.000) & 0.998 (0.001) & 0.992 (0.003) & 0.997 (0.001) & 0.518 (0.005) & 0.020 (36/50) \\
\scen{PEF04} & 0.910 (0.008) & 0.709 (0.017) & 1.000 (0.000) & 0.998 (0.001) & 0.992 (0.003) & 0.997 (0.001) & 0.518 (0.005) & 0.020 (36/50) \\
\scen{PEF05} & 0.910 (0.008) & 0.709 (0.017) & 1.000 (0.000) & 0.998 (0.001) & 0.992 (0.003) & 0.997 (0.001) & 0.518 (0.005) & 0.020 (36/50) \\
\scen{PEF06} & 1.000 (0.000) & 0.481 (0.016) & 1.000 (0.000) & 1.000 (0.000) & 1.000 (0.000) & 1.000 (0.000) & 0.508 (0.012) & 0.002 (50/50) \\
\scen{PEF07} & 1.000 (0.000) & 0.810 (0.016) & 1.000 (0.000) & 1.000 (0.000) & 0.995 (0.001) & 1.000 (0.000) & 0.778 (0.016) & 0.002 (50/50) \\
\scen{PEF08} & 0.946 (0.006) & 0.536 (0.017) & 1.000 (0.000) & 0.967 (0.006) & 0.927 (0.009) & 0.996 (0.001) & 0.531 (0.016) & 0.013 (39/50) \\
\midrule
\multicolumn{9}{@{}l@{}}{\textbf{Nonlinear mechanism: shape change (01), interaction edge $X$-$Z$ (02), tail-gated edge on (03), quadratic flip (04), high-freq sine (05)}} \\
\scen{PNL01} & 1.000 (0.000) & 0.647 (0.017) & 0.994 (0.002) & 0.993 (0.002) & 1.000 (0.000) & 1.000 (0.000) & 0.539 (0.013) & 0.002 (50/50) \\
\scen{PNL02} & 1.000 (0.000) & 0.532 (0.012) & 0.551 (0.013) & 0.994 (0.001) & 0.921 (0.010) & 0.933 (0.009) & 0.506 (0.007) & 0.003 (47/50) \\
\scen{PNL03} & 1.000 (0.000) & 0.595 (0.013) & 0.561 (0.012) & 1.000 (0.000) & 0.996 (0.001) & 0.992 (0.003) & 0.542 (0.012) & 0.004 (40/50) \\
\scen{PNL04} & 1.000 (0.000) & 0.490 (0.015) & 0.503 (0.004) & 0.510 (0.009) & 0.486 (0.017) & 0.511 (0.019) & 0.491 (0.008) & 0.002 (50/50) \\
\scen{PNL05} & 0.998 (0.001) & 0.996 (0.002) & 0.480 (0.013) & 0.627 (0.016) & 0.491 (0.016) & 1.000 (0.000) & 0.994 (0.002) & 0.002 (50/50) \\
\midrule
\multicolumn{9}{@{}l@{}}{\textbf{Conditional noise-shape changes: cond.\ skew (01), cond.\ tails (02), cond.\ mixture (03)}} \\
\scen{PNM01} & 1.000 (0.000) & 0.538 (0.012) & 0.682 (0.016) & 0.865 (0.012) & 0.532 (0.009) & 0.496 (0.013) & 0.556 (0.012) & 0.002 (50/50) \\
\scen{PNM02} & 0.894 (0.009) & 0.542 (0.018) & 0.544 (0.017) & 0.878 (0.011) & 0.542 (0.017) & 0.470 (0.017) & 0.591 (0.016) & 0.119 (12/50) \\
\scen{PNM03} & 0.678 (0.015) & 0.479 (0.016) & 0.607 (0.017) & 0.981 (0.005) & 0.662 (0.017) & 0.996 (0.001) & 0.706 (0.016) & 0.257 (9/50) \\
\midrule
\multicolumn{9}{@{}l@{}}{\textbf{Volatility / variance regimes: cond.\ heteroskedasticity (01), volatility clustering (02), global noise scale (03)}} \\
\scen{PVR01} & 0.759 (0.012) & 0.624 (0.017) & 0.568 (0.015) & 0.837 (0.015) & 0.832 (0.015) & 0.833 (0.015) & 0.527 (0.012) & 0.022 (31/50) \\
\scen{PVR02} & 1.000 (0.000) & 0.621 (0.015) & 0.950 (0.008) & 1.000 (0.000) & 0.997 (0.001) & 1.000 (0.000) & 0.628 (0.017) & 0.002 (50/50) \\
\scen{PVR03} & 1.000 (0.000) & 0.615 (0.015) & 0.959 (0.008) & 1.000 (0.000) & 0.995 (0.002) & 1.000 (0.000) & 0.623 (0.016) & 0.002 (50/50) \\
\bottomrule
\end{tabular}%
}
\end{table*}

\section{Real data example: application to causal change point detection} \label{section: real data examples}

In this section we show the practical value of our statistic by analysing real world financial data. We also illustrate how our statistic may be used in change point scenarios, where the location(s) at which the causal mechanism changes are unknown. 

\subsection{Causal change point detection} \label{section: causal change point detection}

We briefly introduce the notion of a causal change point. Given a data sequence $\{ (X_i, Y_i, \boldsymbol{Z}_i) : i \in [n] \}$, and putting putting $\mathbb{P}_i \coloneq \mathbb{P}_{X_i, Y_i, \boldsymbol{Z}_i}$ for short, we call an integer $\eta \in [n-1]$ a causal change point if it holds that
\begin{align*}
& \mathbb{P}_{\eta} (Y_\eta = y \mid do(X_\eta = x, \boldsymbol{Z}_\eta = \boldsymbol{z})) \neq \mathbb{P}_{X\eta+1} (Y_{\eta+1} = y \mid do(X_{\eta+1} = x, \boldsymbol{Z}_{\eta+1} = \boldsymbol{z})),
\end{align*}
meaning that the causal mechanism changes between times $\eta$ and $\eta+1$. Note that the causal change point differs from the usual definition of a change point, where the joint distribution changes \citep{brodsky2013nonparametric}, since the joint distribution of the data may change without affecting the causal mechanism. The number of such change points, as well as their locations, is unknown. 

Next we explain how the statistic introduced in Section~\ref{section: estimation} can be used to recover the unknown locations of the causal change points. The main idea is to compute the statistic over a sliding window of width $W < n$. Having thus obtained a sequence of statistics, we assign the first candidate change point to the arg-max of the sequence, then eliminate all statistics computed at a distance $W$ or less from this point. We then recursively search for the next candidate change point in this fashion until fewer than $W$ points are left; this step is inspired by \cite[Section 3]{eichinger2018mosum}. Then, for each candidate change point location we compute a p-value using \eqref{equation: perm p-val}, and retain only those locations whose p-value is smaller than some threshold $\bar{p} \in (0,1)$. This produce is summarized in Algorithm~\ref{algorithm: multiple change point detection} in the appendix, where we write $\hat{Q}_{X,Y,Z}^{(i)}$ for the statistic $\hat{Q}_{X, Y, \boldsymbol{Z}}$ computed using the samples $\{ (X_j, Y_j, \boldsymbol{Z}_j) : i-W < j \leq i \}$ and $\{ (X_j, Y_j, \boldsymbol{Z}_j) : i < j \leq i + W \}$. In practice one may choose $\bar{p}$ is a data driven way to account for multiple testing, for example using the Benjamini–Yekutieli procedure \citep{benjamini2001control}. 

\subsection{Analysis of crude oil prices}

We use Algorithm~\ref{algorithm: multiple change point detection} to detect changes in causal dependence between 10-Year U.S. Treasury Yields ($Y$) and Crude Oil prices ($X$), while accounting for inflation and real activity proxies ($\boldsymbol{Z}$). The data ranges from 1960 to 1990 and have monthly frequency. We pre-process the data to ensure stationarity (see Section~\ref{section: pre-processing} of the appendix for precise steps) then run Algorithm~\ref{algorithm: multiple change point detection} with a window size of $W = 25$ and thresholds $\bar{p} = 0.1$ and $\bar{p} = 0.05$. All tuning parameters for the statistic $\hat{Q}_{X,Y,Z}$ are set as in Section~\ref{section: simulation studies}.

Our data example is motivated by the oil--macro literature \cite{hamilton1983} and the monetary reinterpretation of stagflation dynamics \cite{barsky2001}. Algorithm~\ref{algorithm: multiple change point detection} detects a prominent change aligning with the 1973–74 Organization of the Petroleum Exporting Countries (OPEC) embargo, which is consistent with a major restructuring of macro–financial transmission during the first oil shock. The method also flags a change in 1970 and another around 1979. The second change is consistent with the second oil shock window Volcker policy announcements. In contrast, earlier episodes often associated with macro stress (e.g., the 1966 credit crunch and the 1971 Nixon shock) do not register as causal breaks, reinforcing that the detected changepoints reflect genuine changes in the underlying linkage rather than broad market turbulence. A plot showing values of the test statistic along with estimated change point locations and associated p-values is given in Figure~\ref{fig:S1_inflation}.

\begin{figure*}[t]
    \centering
    \includegraphics[width=\textwidth]{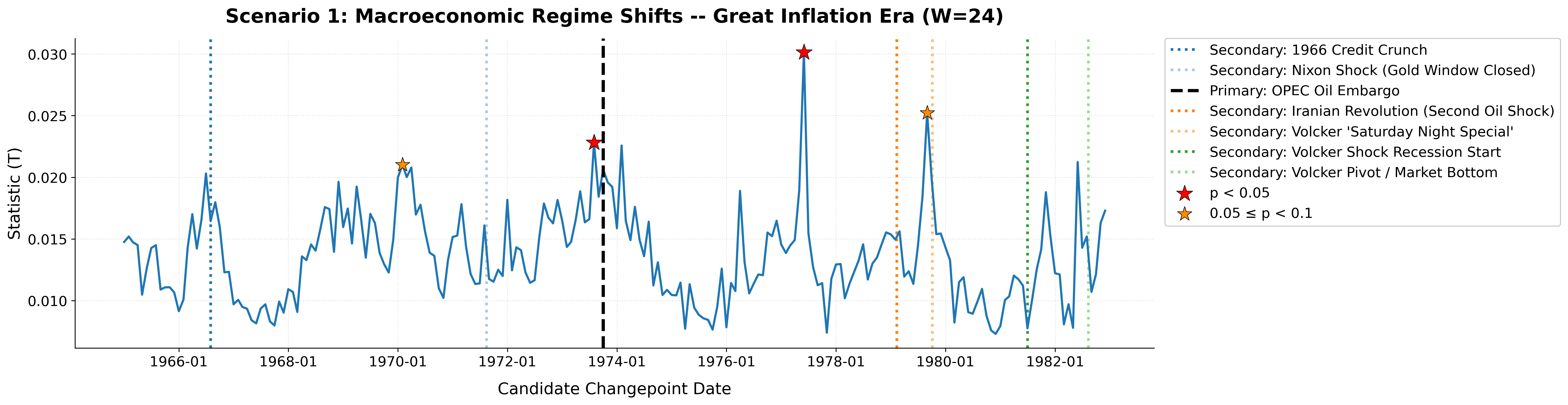}
    \caption{Stars represent estimated p-values, and dashed lines represent historically important events.}
    \label{fig:S1_inflation}
\end{figure*}

\section{Discussion}
\label{discussion_section}

We proposed a fully non-parametric framework for detecting changes in causal dependence using conditional copulas and kernel mean embeddings (Theorem~\ref{theorem: copula causality} and Proposition~\ref{proposition: Q separation}). Empirically, the method performs strongly across a wide range of settings (Table~\ref{tab:scenarios_positive:POOL_PERM} and Fig. \ref{fig:S1_inflation}), including challenging non-linear and noise-driven scenarios. The real data analysis further demonstrates that the method identifies economically meaningful changepoints aligned with known events. At the same time, several limitations remain. First, the method relies on standard causal assumptions such as causal sufficiency and correct specification of $\boldsymbol{Z}$; violations may lead to spurious detections or missed changes. Second, the $k_n$-nearest neighbour estimation can degrade in high-dimensional settings, which is reflected in reduced power in more complex scenarios (e.g.\ \scen{PNM03}, \scen{PVR01}). Third, inference relies on permutation testing, since the null distribution of the statistic is unknown. Finally, while the method detects whether a change occurs, it does not directly characterize the nature of the change. Overall, our results show that detecting changes in causal dependence is feasible in a fully non-parametric setting. This capability is particularly relevant in high-stakes domains such as finance or policy analysis, where shifts in underlying mechanisms can materially affect decisions. 


\bibliography{bib/kernal-causal-references}

@article{Muandet17:KME,
  year = {2017},
  volume = {10},
  journal = {Foundations and Trends in Machine Learning},
  title = {Kernel Mean Embedding of Distributions: A Review and Beyond},
  number = {1-2},
  pages = {1--141},
  author = {Krikamol Muandet and Kenji Fukumizu and Bharath Sriperumbudur and Bernhard Sch\"olkopf}
}

@book{pearlj,
	author = {Judea Pearl},
	title = {Causality: Models, Reasoning and Inference},
	publisher = {Cambridge University Press},
	year = {2000}
}

@inproceedings{DBLP:conf/aistats/GaoAYRK25,
  author       = {Shanyun Gao and
                  Raghavendra Addanki and
                  Tong Yu and
                  Ryan A. Rossi and
                  Murat Kocaoglu},
  title        = {Causal Discovery-Driven Change Point Detection in Time Series},
  booktitle    = {Proc. of {AISTATS}},
  volume       = {258},
  pages        = {3754--3762},
  year         = {2025},
}

@article{gu2020empirical,
  title={Empirical asset pricing via machine learning},
  author={Gu, Shihao and Kelly, Bryan and Xiu, Dacheng},
  journal={The Review of Financial Studies},
  volume={33},
  number={5},
  pages={2223--2273},
  year={2020},
}

@inproceedings{DBLP:conf/aistats/CheungSZ19,
  author       = {Wang Chi Cheung and
                  David Simchi{-}Levi and
                  Ruihao Zhu},
  title        = {Learning to Optimize under Non-Stationarity},
  booktitle    = {Proc. of {AISTATS}},
  volume       = {89},
  pages        = {1079--1087},
  year         = {2019},
}

@book{Basseville1993,
author = {Basseville, M. and Nikiforov, I.},
publisher = {Prentice Hall Englewood Cliffs},
title = {{Detection of abrupt changes: theory and application}},
volume = {104},
year = {1993}
}

@inproceedings{DBLP:conf/nips/HarchaouiBM08,
  author       = {Za{\"{\i}}d Harchaoui and
                  Francis R. Bach and
                  Eric Moulines},
  title        = {Kernel Change-point Analysis},
  booktitle    = {Proc. Of NeurIPS},
  pages        = {609--616},
  year         = {2008},
}

@inproceedings{DBLP:conf/icml/HockingRB15,
  author       = {Toby Hocking and
                  Guillem Rigaill and
                  Guillaume Bourque},
  title        = {PeakSeg: constrained optimal segmentation and supervised penalty learning
                  for peak detection in count data},
  booktitle    = {Proc. of ICML},
  volume       = {37},
  pages        = {324--332},
  year         = {2015},
}

@article{f6bf4042-8440-3ce2-87e3-78f31cedac7c,
 author = {R. Killick and P. Fearnhead and I. A. Eckley},
 journal = {Journal of the American Statistical Association},
 number = {500},
 pages = {1590--1598},
 title = {Optimal Detection of Changepoints With a Linear Computational Cost},
 volume = {107},
 year = {2012}
}

@article{DBLP:journals/corr/abs-2311-06012,
  author       = {Emmanouil Angelis and
                  Francesco Quinzan and
                  Ashkan Soleymani and
                  Patrick Jaillet and
                  Stefan Bauer},
  title        = {Double Machine Learning Based Structure Identification from Temporal
                  Data},
  journal      = {Trans. Mach. Learn. Res.},
  volume       = {2025},
  year         = {2025},
}

@article{huang2024,
      title={Causal Change Point Detection and Localization}, 
      author={Shimeng Huang and Jonas Peters and Niklas Pfister},
      year={2024},
    journal = {CoRR},
      volume={abs/2403.12677},
}

@article{mcgonigle2025nonparametric,
  title={Nonparametric data segmentation in multivariate time series via joint characteristic functions},
  author={McGonigle, Euan T and Cho, Haeran},
  journal={Biometrika},
  pages={asaf024},
  year={2025},
  publisher={Oxford University Press}
}

@article{fukumizu2007kernel,
  title={Kernel measures of conditional dependence},
  author={Fukumizu, Kenji and Gretton, Arthur and Sun, Xiaohai and Sch{\"o}lkopf, Bernhard},
  journal={Advances in neural information processing systems},
  volume={20},
  year={2007}
}

@article{azadkia2021fast,
  title={A fast non-parametric approach for local causal structure learning},
  author={Azadkia, Mona and Taeb, Armeen and B{\"u}hlmann, Peter},
  journal={arXiv preprint arXiv:2111.14969},
  year={2021}
}

@article{azadkia2021simple,
  title={A simple measure of conditional dependence},
  author={Azadkia, Mona and Chatterjee, Sourav},
  journal={The Annals of Statistics},
  volume={49},
  number={6},
  pages={3070--3102},
  year={2021},
  publisher={JSTOR}
}

@article{gijbels2011conditional,
  title={Conditional copulas, association measures and their applications},
  author={Gijbels, Ir{\`e}ne and Veraverbeke, No{\"e}l and Omelka, Marel},
  journal={Computational Statistics \& Data Analysis},
  volume={55},
  number={5},
  pages={1919--1932},
  year={2011},
  publisher={Elsevier}
}

@article{veraverbeke2011estimation,
  title={Estimation of a conditional copula and association measures},
  author={Veraverbeke, No{\"e}l and Omelka, Marek and Gijbels, Irene},
  journal={Scandinavian Journal of Statistics},
  volume={38},
  number={4},
  pages={766--780},
  year={2011},
  publisher={Wiley Online Library}
}

@inproceedings{sklar1959fonctions,
  title={Fonctions de r{\'e}partition {\`a} n dimensions et leurs marges},
  author={Sklar, M},
  booktitle={Annales de l'ISUP},
  volume={8},
  number={3},
  pages={229--231},
  year={1959}
}

@article {zhenpartial2022,
    AUTHOR = {Huang, Zhen and Deb, Nabarun and Sen, Bodhisattva},
     TITLE = {Kernel partial correlation coefficient---a measure of
              conditional dependence},
   JOURNAL = {J. Mach. Learn. Res.},
  FJOURNAL = {Journal of Machine Learning Research (JMLR)},
    VOLUME = {23},
      YEAR = {2022},
     PAGES = {Paper No. [216], 58},
      ISSN = {1532-4435},
   MRCLASS = {62H20 (62G05 62G20)},
  MRNUMBER = {4577169},
MRREVIEWER = {Martynas Manstavi\v{c}ius},
}

@article{xu2025quickest,
  title={Quickest Causal Change Point Detection by Adaptive Intervention},
  author={Xu, Haijie and Zhang, Chen},
  journal={arXiv preprint arXiv:2506.07760},
  year={2025}
}

@article{chatterjee2021new,
  title   = {A New Coefficient of Correlation},
  author  = {Chatterjee, Sourav},
  journal = {Journal of the American Statistical Association},
  volume  = {116},
  number  = {536},
  pages   = {2009--2022},
  year    = {2021}
}

@article{sze2007dcor,
  title   = {Measuring and Testing Dependence by Correlation of Distances},
  author  = {Sz{\'e}kely, G{\'a}bor J. and Rizzo, Maria L. and Bakirov, Nail K.},
  journal = {The Annals of Statistics},
  volume  = {35},
  number  = {6},
  pages   = {2769--2794},
  year    = {2007}
}

@article{sze2014pdcor,
  title   = {Partial Distance Correlation with Methods for Dissimilarities},
  author  = {Sz{\'e}kely, G{\'a}bor J. and Rizzo, Maria L.},
  journal = {The Annals of Statistics},
  volume  = {42},
  number  = {6},
  pages   = {2382--2412},
  year    = {2014},
  doi     = {10.1214/14-AOS1255}
}

@article{gretton2012mmd,
  title   = {A Kernel Two-Sample Test},
  author  = {Gretton, Arthur and Borgwardt, Karsten M. and Rasch, Malte J. and Sch{\"o}lkopf, Bernhard and Smola, Alex J.},
  journal = {Journal of Machine Learning Research},
  volume  = {13},
  pages   = {723--773},
  year    = {2012}
}

@article{kendall1942partial,
  title={Partial rank correlation},
  author={Kendall, Maurice G},
  journal={Biometrika},
  volume={32},
  number={3/4},
  pages={277--283},
  year={1942},
  publisher={JSTOR}
}

@book{nelsen2006introduction,
  title={An introduction to copulas},
  author={Nelsen, Roger B},
  year={2006},
  publisher={Springer}
}

@article{buhlmann2014cam,
  title={CAM: Causal additive models, high-dimensional order search and penalized regression},
  journal={Ann. Statist.},
  volume={42},
  pages={2526-2556},
  author={B{\"u}hlmann, Peter and Peters, Jonas and Ernest, Jan},
  year={2014}
}

@incollection{imoto2001estimation,
  title={Estimation of genetic networks and functional structures between genes by using Bayesian networks and nonparametric regression},
  author={Imoto, Seiya and Goto, Takao and Miyano, Satoru},
  booktitle={Biocomputing 2002},
  pages={175--186},
  year={2001},
  publisher={World Scientific}
}

@article{spirtes1991algorithm,
  title={An algorithm for fast recovery of sparse causal graphs},
  author={Spirtes, Peter and Glymour, Clark},
  journal={Social science computer review},
  volume={9},
  number={1},
  pages={62--72},
  year={1991},
  publisher={Sage Publications Sage CA: Thousand Oaks, CA}
}

@incollection{pearl1995theory,
  title={A theory of inferred causation},
  author={Pearl, Judea and Verma, Thomas S},
  booktitle={Studies in Logic and the Foundations of Mathematics},
  volume={134},
  pages={789--811},
  year={1995},
  publisher={Elsevier}
}

@article{szekely2014partial,
  title={Partial distance correlation with methods for dissimilarities},
  journal={Ann. Statist.},
  volume={42},
  pages={2382-2412},
  author={Sz{\'e}kely, G{\'a}bor J and Rizzo, Maria L},
  year={2014}
}

@article{chen2023iscan,
  title={iscan: Identifying causal mechanism shifts among nonlinear additive noise models},
  author={Chen, Tianyu and Bello, Kevin and Aragam, Bryon and Ravikumar, Pradeep},
  journal={Advances in Neural Information Processing Systems},
  volume={36},
  pages={44671--44706},
  year={2023}
}

@article{matteson2014nonparametric,
  title={A nonparametric approach for multiple change point analysis of multivariate data},
  author={Matteson, David S and James, Nicholas A},
  journal={Journal of the American Statistical Association},
  volume={109},
  number={505},
  pages={334--345},
  year={2014},
  publisher={Taylor \& Francis}
}

@inproceedings{DBLP:conf/icml/QuinzanSJRB23,
  author       = {Francesco Quinzan and
                  Ashkan Soleymani and
                  Patrick Jaillet and
                  Cristian R. Rojas and
                  Stefan Bauer},
title        = {{DRCFS:} Doubly Robust Causal Feature Selection},
  booktitle    = {Proc. of {ICML}},
  series       = {Proceedings of Machine Learning Research},
  volume       = {202},
  pages        = {28468--28491},
  year         = {2023},
}

@article{chen2024identifying,
  title={Identifying general mechanism shifts in linear causal representations},
  author={Chen, Tianyu and Bello, Kevin and Locatello, Francesco and Aragam, Bryon and Ravikumar, Pradeep},
  journal={Advances in Neural Information Processing Systems},
  volume={37},
  pages={42405--42429},
  year={2024}
}

@article{wang2018direct,
  title={Direct estimation of differences in causal graphs},
  author={Wang, Yuhao and Squires, Chandler and Belyaeva, Anastasiya and Uhler, Caroline},
  journal={Advances in neural information processing systems},
  volume={31},
  year={2018}
}

@article{ghoshal2019direct,
  title={Direct learning with guarantees of the difference dag between structural equation models},
  author={Ghoshal, Asish and Bello, Kevin and Honorio, Jean},
  journal={arXiv preprint arXiv:1906.12024},
  year={2019}
}

@article{li2023kernel,
  title={Kernel-based partial permutation test for detecting heterogeneous functional relationship},
  author={Li, Xinran and Jiang, Bo and Liu, Jun S},
  journal={Journal of the American Statistical Association},
  volume={118},
  number={542},
  pages={1429--1447},
  year={2023},
  publisher={Taylor \& Francis}
}

@inproceedings{saeed2020causal,
  title={Causal structure discovery from distributions arising from mixtures of dags},
  author={Saeed, Basil and Panigrahi, Snigdha and Uhler, Caroline},
  booktitle={International conference on machine learning},
  pages={8336--8345},
  year={2020},
  organization={PMLR}
}

@article{klebanov2020rigorous,
  title={A rigorous theory of conditional mean embeddings},
  author={Klebanov, Ilja and Schuster, Ingmar and Sullivan, Timothy John},
  journal={SIAM Journal on Mathematics of Data Science},
  volume={2},
  number={3},
  pages={583--606},
  year={2020},
  publisher={SIAM}
}

@article{friedman1977algorithm,
  title={An algorithm for finding best matches in logarithmic expected time},
  author={Friedman, Jerome H and Bentley, Jon Louis and Finkel, Raphael Ari},
  journal={ACM Transactions on Mathematical Software (TOMS)},
  volume={3},
  number={3},
  pages={209--226},
  year={1977},
  publisher={ACM New York, NY, USA}
}

@article{praestgaard1995permutation,
  title={Permutation and bootstrap Kolmogorov-Smirnov tests for the equality of two distributions},
  author={Praestgaard, Jens Thomas},
  journal={Scandinavian Journal of Statistics},
  pages={305--322},
  year={1995},
  publisher={JSTOR}
}

@article{DBLP:journals/jmlr/Spirtes10,
  author    = {Peter Spirtes},
  title     = {Introduction to Causal Inference},
  journal   = {Journal of Machine Learning Research},
  volume    = {11},
  pages     = {1643--1662},
  year      = {2010}
}

@article{sriperumbudur2010hilbert,
  title={Hilbert space embeddings and metrics on probability measures},
  author={Sriperumbudur, Bharath K and Gretton, Arthur and Fukumizu, Kenji and Sch{\"o}lkopf, Bernhard and Lanckriet, Gert RG},
  journal={The Journal of Machine Learning Research},
  volume={11},
  pages={1517--1561},
  year={2010},
  publisher={JMLR. org}
}

@article{hamilton1983,
  title = {Oil and the Macroeconomy since World War II},
  author = {Hamilton, James D.},
  journal = {Journal of Political Economy},
  volume = {91},
  number = {2},
  pages = {228--248},
  year = {1983},
  publisher = {The University of Chicago Press}
}

@article{barsky2001,
  title = {Do We Really Know that Oil Caused the Great Stagflation? A Monetary Alternative},
  author = {Barsky, Robert B. and Kilian, Lutz},
  journal = {NBER Macroeconomics Annual},
  volume = {16},
  pages = {137--183},
  year = {2001},
  publisher = {MIT Press}
}

@article{brunnermeier2009,
  title = {Deciphering the Liquidity and Credit Crunch 2007--2008},
  author = {Brunnermeier, Markus K.},
  journal = {Journal of Economic Perspectives},
  volume = {23},
  number = {1},
  pages = {77--100},
  year = {2009}
}

@article{cornett2011,
  title = {Liquidity Risk Management and Credit Supply in the Financial Crisis},
  author = {Cornett, Marcia Millon and McNutt, Jamie John and Strahan, Philip E. and Tehranian, Hassan},
  journal = {Journal of Financial Economics},
  volume = {101},
  number = {2},
  pages = {297--312},
  year = {2011},
  publisher = {Elsevier}
}

@article{longstaff2010,
  title = {The Subprime Credit Crisis and Contagion in Financial Markets},
  author = {Longstaff, Francis A.},
  journal = {Journal of Financial Economics},
  volume = {97},
  number = {3},
  pages = {436--450},
  year = {2010},
  publisher = {Elsevier}
}

@article{baur2010,
  title = {Is Gold a Hedge or a Safe Haven? An Analysis of Stocks, Bonds and Gold},
  author = {Baur, Dirk G. and Lucey, Brian M.},
  journal = {The Financial Review},
  volume = {45},
  number = {2},
  pages = {217--229},
  year = {2010},
  publisher = {Wiley Online Library}
}

@article{cheema2022,
  title = {The 2008 Global Financial Crisis and COVID-19 Pandemic: How Safe Are the Safe Haven Assets?},
  author = {Cheema, Muhammad A. and Faff, Robert W. and Szulczyk, Kenneth R.},
  journal = {International Review of Financial Analysis},
  volume = {83},
  pages = {102316},
  year = {2022},
  publisher = {Elsevier}
}

@article{davies2018,
  title = {The Heterogeneous Impact of Brexit: Early Indications from the FTSE},
  author = {Davies, Ronald B. and Studnicka, Zuzanna},
  journal = {European Economic Review},
  volume = {110},
  pages = {1--17},
  year = {2018},
  publisher = {Elsevier}
}

@article{garreau2017large,
  title={Large sample analysis of the median heuristic},
  author={Garreau, Damien and Jitkrittum, Wittawat and Kanagawa, Motonobu},
  journal={arXiv preprint arXiv:1707.07269},
  year={2017}
}

@article{vizing1964estimate,
  title={On an estimate of the chromatic class of a p-graph},
  author={Vizing, Vadim G},
  journal={Diskret analiz},
  volume={3},
  pages={25--30},
  year={1964}
}

@article{huang2022kernel,
  title={Kernel partial correlation coefficient--a measure of conditional dependence},
  author={Huang, Zhen and Deb, Nabarun and Sen, Bodhisattva},
  journal={Journal of Machine Learning Research},
  volume={23},
  number={216},
  pages={1--58},
  year={2022}
}

@article{roudaki2026kernel,
  title={Kernel Integrated $R^{2}$: A Measure of Dependence},
  author={Roudaki, Pouya and Gavioli-Akilagun, Shakeel and Kalinke, Florian and Azadkia, Mona and Szab{\'o}, Zolt{\'a}n},
  journal={arXiv preprint arXiv:2602.22985},
  year={2026}
}

@article{boucheron2003concentration,
  title={Concentration inequalities using the entropy method},
  author={Boucheron, St{\'e}phane and Lugosi, G{\'a}bor and Massart, Pascal},
  journal={The Annals of Probability},
  volume={31},
  number={3},
  pages={1583--1614},
  year={2003},
  publisher={Institute of Mathematical Statistics}
}

@article{benjamini2001control,
  title={The control of the false discovery rate in multiple testing under dependency},
  author={Benjamini, Yoav and Yekutieli, Daniel},
  journal={Annals of statistics},
  pages={1165--1188},
  year={2001},
  publisher={JSTOR}
}

@article{stone1980optimal,
  title={Optimal rates of convergence for nonparametric estimators},
  author={Stone, Charles J},
  journal={The Annals of Statistics},
  pages={1348--1360},
  year={1980},
  publisher={JSTOR}
}

@article{eichinger2018mosum,
  title={A MOSUM procedure for the estimation of multiple random change points},
  author={Eichinger, Birte and Kirch, Claudia},
  journal={Bernoulli },
  pages={526--564},
  year={2018}
}

@BOOK{diestel77vector,
  AUTHOR =       {Joseph Diestel and John J. Uhl},
  TITLE =        {Vector Measures},
  PUBLISHER =    {American Mathematical Society},
  YEAR =         {1977},
}

@book{brodsky2013nonparametric,
  title={Nonparametric methods in change point problems},
  author={Brodsky, Emily and Darkhovsky, Boris S},
  year={2013},
  publisher={Springer Science \& Business Media}
}

\newpage

\appendix

\begin{center}

{\Large\bf SUPPLEMENTARY MATERIALS}

\end{center}

\setcounter{equation}{0}

\section*{Contents}

\startcontents[sections]
\printcontents[sections]{l}{1}{\setcounter{tocdepth}{2}}

\section{Algorithm for causal change point detection} \label{section: causal }

Pseudo code for the procedure for recovering the unknown locations of causal change points using our proposed statistic, which was described in Section~\ref{section: causal change point detection}, is given in Algorithm~\ref{algorithm: multiple change point detection} below. 

\begin{algorithm}[!htbp]
\caption{causal change point detection using the proposed two sample statistic.}\label{algorithm: multiple change point detection}
\SetKwData{Left}{left}\SetKwData{This}{this}\SetKwData{Up}{up}
\SetKwFunction{Union}{Union}\SetKwFunction{FindCompress}{FindCompress}
\SetKwInOut{Input}{input}
\SetKwInOut{Output}{output}
\Input{data $\{ X_i, Y_i, \boldsymbol{Z}_i \}$ $i \in [n]$; window $W < n$; threshold $\bar{p} \in (0,1)$.}
\Output{estimated change point locations $\hat{\Theta} \subseteq [n]$.}
\BlankLine
$J \leftarrow [n]$; $\tilde{\Theta} \leftarrow \{ \emptyset \}$; $\hat{\Theta} \leftarrow \{ \emptyset \}$ \\
\BlankLine
\While{$|J| > 0$}{
$\tilde{\eta} \leftarrow \argmax_{i \in J} \hat{Q}_{X,Y,Z}^{(i)}$ \\
$\tilde{\Theta} \leftarrow \tilde{\Theta} \cup \{ \tilde{\eta} \}$ \\
$J \leftarrow J \setminus \{ j : \exists \eta' \in \tilde{\Theta} \text{ s.t. } |\eta' - j| < W \} $
}
\BlankLine
\For{$\eta' \in \tilde{\Theta}$}{
\text{compute p-value for $\hat{Q}_{X,Y,Z}^{(\eta')}$ using \eqref{equation: perm p-val}} \\
\If{\upshape{p-value} $\leq \bar{p}$}{
$\hat{\Theta} \leftarrow \hat{\Theta} \cup \{ \eta' \}$
}
\BlankLine
}
\Return{$\hat{\Theta}$}
\BlankLine
\end{algorithm}

\section{Additional numerical illustrations}

In this section we provide additional numerical illustrations on the performance of the proposed statistic. Section~\ref{section: p-vals instead} reproduces the simulation study in the main paper but reports p-values rather than AUC values, Section~\ref{section: behaviour under the null} investigates the behaviour of the proposed statistic, as well as competitors listed in Section~\ref{section: competing methods} in settings where the causal mechanism does not change, Section~\ref{section: further real data examples} presents the analysis of three additional financial time series using Algorithm~\ref{algorithm: multiple change point detection}, and finally Section~\ref{section: pre-processing} gives details of the pre-processing steps applied in the real data examples. 
\subsection{Additional studies on detection power} \label{section: p-vals instead}

We repeat the simulation study presented in Section~\ref{section: simulation studies}, but rather than presenting the AUC we report the median p-value for each method computed according to Section~\ref{section: inference}, as well as the proportion of times the p-value was smaller than $0.05$. The results are presented in Table~\ref{tab:scenarios_positive_pvals:POOL_PERM}, and generally align with the results of the simulation study presented in the main paper. 

\begin{table*}[!htbp]
\centering
\scriptsize
\setlength{\tabcolsep}{4pt}
\renewcommand{\arraystretch}{1.15}
\caption{Synthetic positive scenarios: p-values by method. Entries are median p-values with rejection counts (k/R) per method.}
\label{tab:scenarios_positive_pvals:POOL_PERM}
\begin{tabularx}{\textwidth}{@{} >{\raggedright\arraybackslash}X ccccccc @{}}
\toprule
\textbf{Test name} & \multicolumn{7}{c}{\textbf{p-value (k/R)}} \\
\cmidrule(lr){2-8}
 & $\hat{Q}_{X,Y,Z}$ & \scen{XI\_AC} & \scen{CHSIC} & \scen{PDCORR} & \scen{PSPEAR} & \scen{PKENDALL} & \scen{KPCGRAPH} \\
\midrule
\multicolumn{8}{@{}l@{}}{\textbf{Markov blanket / edge structure: linear edge on (01), nonlinear edge on (02), new driver (03), edge off (04), collinear $X$ (05)}} \\
\scen{PMB01} & 0.042 (26/50) & 0.073 (17/50) & 0.002 (50/50) & 0.002 (50/50) & 0.002 (50/50) & 0.002 (50/50) & 0.534 (2/50) \\
\scen{PMB02} & 0.027 (30/50) & 0.533 (4/50) & 0.002 (44/50) & 0.025 (33/50) & 0.009 (39/50) & 0.002 (43/50) & 0.472 (4/50) \\
\scen{PMB03} & 0.189 (9/50) & 0.294 (11/50) & 0.002 (50/50) & 0.002 (50/50) & 0.002 (49/50) & 0.002 (50/50) & 0.645 (0/50) \\
\scen{PMB04} & 0.003 (42/50) & 0.320 (7/50) & 0.002 (50/50) & 0.002 (50/50) & 0.002 (49/50) & 0.002 (50/50) & 0.351 (4/50) \\
\scen{PMB05} & 0.002 (50/50) & 0.062 (20/50) & 0.002 (50/50) & 0.002 (50/50) & 0.002 (50/50) & 0.002 (50/50) & 0.709 (0/50) \\
\midrule
\multicolumn{8}{@{}l@{}}{\textbf{Effect: sign flip (01; Gauss 06), strength shift (02; +$Z$ 03; +distract 04; +$Z$ unscaled 05), sparse $X$ (07), Poisson (08)}} \\
\scen{PEF01} & 0.002 (50/50) & 0.550 (2/50) & 0.002 (50/50) & 0.002 (48/50) & 0.002 (50/50) & 0.002 (50/50) & 0.421 (5/50) \\
\scen{PEF02} & 0.150 (10/50) & 0.179 (10/50) & 0.002 (50/50) & 0.002 (50/50) & 0.002 (50/50) & 0.002 (50/50) & 0.444 (5/50) \\
\scen{PEF03} & 0.020 (36/50) & 0.117 (14/50) & 0.002 (50/50) & 0.002 (50/50) & 0.002 (49/50) & 0.002 (50/50) & 0.460 (1/50) \\
\scen{PEF04} & 0.020 (36/50) & 0.117 (14/50) & 0.002 (50/50) & 0.002 (50/50) & 0.002 (49/50) & 0.002 (50/50) & 0.460 (1/50) \\
\scen{PEF05} & 0.020 (36/50) & 0.117 (14/50) & 0.002 (50/50) & 0.002 (50/50) & 0.002 (49/50) & 0.002 (50/50) & 0.460 (1/50) \\
\scen{PEF06} & 0.002 (50/50) & 0.659 (2/50) & 0.002 (50/50) & 0.002 (50/50) & 0.002 (50/50) & 0.002 (50/50) & 0.533 (1/50) \\
\scen{PEF07} & 0.002 (50/50) & 0.035 (32/50) & 0.002 (50/50) & 0.002 (50/50) & 0.002 (50/50) & 0.002 (50/50) & 0.634 (0/50) \\
\scen{PEF08} & 0.013 (39/50) & 0.382 (5/50) & 0.002 (50/50) & 0.002 (45/50) & 0.010 (37/50) & 0.002 (49/50) & 0.402 (4/50) \\
\midrule
\multicolumn{8}{@{}l@{}}{\textbf{Nonlinear mechanism: shape change (01), interaction edge $X$-$Z$ (02), tail-gated edge on (03), quadratic flip (04), high-freq sine (05)}} \\
\scen{PNL01} & 0.002 (50/50) & 0.423 (1/50) & 0.002 (48/50) & 0.002 (48/50) & 0.002 (50/50) & 0.002 (50/50) & 0.430 (5/50) \\
\scen{PNL02} & 0.003 (47/50) & 0.460 (1/50) & 0.527 (3/50) & 0.002 (49/50) & 0.010 (36/50) & 0.011 (36/50) & 0.449 (5/50) \\
\scen{PNL03} & 0.004 (40/50) & 0.418 (2/50) & 0.480 (4/50) & 0.002 (50/50) & 0.002 (47/50) & 0.002 (48/50) & 0.375 (3/50) \\
\scen{PNL04} & 0.002 (50/50) & 0.633 (0/50) & 0.402 (4/50) & 0.106 (21/50) & 0.865 (0/50) & 0.497 (3/50) & 0.343 (13/50) \\
\scen{PNL05} & 0.002 (50/50) & 0.002 (49/50) & 0.592 (2/50) & 0.366 (8/50) & 0.418 (5/50) & 0.002 (50/50) & 0.002 (50/50) \\
\midrule
\multicolumn{8}{@{}l@{}}{\textbf{Conditional noise-shape changes: cond.\ skew (01), cond.\ tails (02), cond.\ mixture (03)}} \\
\scen{PNM01} & 0.002 (50/50) & 0.491 (3/50) & 0.197 (9/50) & 0.036 (27/50) & 0.469 (4/50) & 0.416 (5/50) & 0.606 (2/50) \\
\scen{PNM02} & 0.119 (12/50) & 0.525 (5/50) & 0.408 (5/50) & 0.056 (24/50) & 0.419 (7/50) & 0.454 (5/50) & 0.494 (4/50) \\
\scen{PNM03} & 0.257 (9/50) & 0.396 (1/50) & 0.300 (4/50) & 0.008 (38/50) & 0.194 (9/50) & 0.002 (49/50) & 0.345 (11/50) \\
\midrule
\multicolumn{8}{@{}l@{}}{\textbf{Volatility / variance regimes: cond.\ heteroskedasticity (01), volatility clustering (02), global noise scale (03)}} \\
\scen{PVR01} & 0.022 (31/50) & 0.225 (11/50) & 0.382 (6/50) & 0.032 (28/50) & 0.023 (28/50) & 0.043 (27/50) & 0.516 (5/50) \\
\scen{PVR02} & 0.002 (50/50) & 0.431 (3/50) & 0.002 (40/50) & 0.002 (50/50) & 0.002 (48/50) & 0.002 (50/50) & 0.265 (11/50) \\
\scen{PVR03} & 0.002 (50/50) & 0.399 (3/50) & 0.002 (40/50) & 0.002 (50/50) & 0.002 (46/50) & 0.002 (50/50) & 0.279 (10/50) \\
\bottomrule
\end{tabularx}
\end{table*}

\subsection{Behaviour under the null of no change in causal dependence} \label{section: behaviour under the null}

We repeat the experiments in Section~\ref{section: simulation studies} and~\ref{section: p-vals instead} using synthetic data generated such that the causal dependence does not change. The data generating processes are described in detail in Section~\ref{section: negative controls}, and the results are reported in the tables below. 

\begin{table*}[!htbp]
\centering
\setlength{\tabcolsep}{4pt}
\renewcommand{\arraystretch}{1.15}
\caption{Synthetic control (null) scenarios grouped by what they stress-test. AUC shown for multiple methods; p-values shown as median p-value with rejection count (k/R) for $\hat{Q}_{X,Y,Z}$.}
\label{tab:scenarios_null:POOL_PERM}
\resizebox{\textwidth}{!}{%
\footnotesize
\begin{tabular}{@{} l ccccccc c @{}}
\toprule
\textbf{Test name} & \multicolumn{7}{c}{\textbf{AUC (mean (sd))}} & \makecell[c]{\textbf{p-value}\\\textbf{(k/R)}} \\
\cmidrule(lr){2-8}
 & $\hat{Q}_{X,Y,Z}$ & \scen{XI\_AC} & \scen{CHSIC} & \scen{PDCORR} & \scen{PSPEAR} & \scen{PKENDALL} & \scen{KPCGRAPH} & $\hat{Q}_{X,Y,Z}$ \\
\midrule
\multicolumn{9}{@{}l@{}}{\textbf{Baseline nulls: stationary null (01), multi-$X$ aligned with PMB03 (02)}} \\
\scen{NCL01} & 0.499 (0.003) & 0.500 (0.000) & 0.500 (0.000) & 0.500 (0.000) & 0.500 (0.000) & 0.500 (0.000) & 0.500 (0.000) & 0.523 (1/50) \\
\scen{NCL02} & 0.475 (0.008) & 0.500 (0.008) & 0.926 (0.010) & 0.662 (0.016) & 0.515 (0.008) & 0.713 (0.017) & 0.911 (0.010) & 0.994 (0/50) \\
\midrule
\multicolumn{9}{@{}l@{}}{\textbf{Invariance checks: global rescale (01), $Y$ monotone (02), $X$ monotone given $Z$ (03)}} \\
\scen{NIV01} & 0.829 (0.010) & 0.500 (0.001) & 1.000 (0.000) & 0.500 (0.000) & 0.500 (0.000) & 0.500 (0.000) & 0.498 (0.003) & 0.272 (14/50) \\
\scen{NIV02} & 0.499 (0.003) & 0.500 (0.000) & 0.545 (0.015) & 0.936 (0.008) & 0.500 (0.000) & 0.500 (0.000) & 0.548 (0.013) & 0.976 (0/50) \\
\scen{NIV03} & 0.491 (0.004) & 0.487 (0.006) & 0.597 (0.017) & 0.804 (0.015) & 0.500 (0.003) & 0.506 (0.004) & 0.519 (0.012) & 0.724 (0/50) \\
\midrule
\multicolumn{9}{@{}l@{}}{\textbf{Marginal drift only: $Z$ loc/scale (01), $X$ mean shift (02), $Y$ trend (03), $Y$ seasonal (04)}} \\
\scen{NMD01} & 0.708 (0.014) & 0.484 (0.010) & 0.990 (0.002) & 0.507 (0.002) & 0.509 (0.007) & 0.525 (0.009) & 0.488 (0.013) & 0.112 (18/50) \\
\scen{NMD02} & 0.506 (0.018) & 0.494 (0.019) & 0.498 (0.018) & 0.525 (0.017) & 0.500 (0.018) & 0.507 (0.018) & 0.473 (0.019) & 0.559 (1/50) \\
\scen{NMD03} & 0.499 (0.003) & 0.500 (0.000) & 0.500 (0.000) & 0.500 (0.000) & 0.500 (0.000) & 0.500 (0.000) & 0.500 (0.000) & 0.450 (3/50) \\
\scen{NMD04} & 0.504 (0.003) & 0.500 (0.000) & 0.500 (0.000) & 0.500 (0.000) & 0.500 (0.000) & 0.500 (0.000) & 0.500 (0.000) & 0.519 (3/50) \\
\midrule
\multicolumn{9}{@{}l@{}}{\textbf{Noise-law / tail-shape: noise-law shift (01), tails shift (02)}} \\
\scen{NNS01} & 0.513 (0.013) & 0.479 (0.016) & 0.513 (0.011) & 0.507 (0.016) & 0.486 (0.016) & 0.493 (0.017) & 0.491 (0.016) & 0.390 (4/50) \\
\scen{NNS02} & 0.567 (0.015) & 0.491 (0.018) & 0.536 (0.016) & 0.483 (0.017) & 0.473 (0.017) & 0.587 (0.018) & 0.491 (0.017) & 0.286 (10/50) \\
\midrule
\multicolumn{9}{@{}l@{}}{\textbf{Covariate / confounding drift: $X|Z$ drift (01), $Z$ cov shift (02), $X|Z$ drift, copula (03), $F_Z$ drift only (04), discrete $Z$ (05)}} \\
\scen{NCF01} & 0.499 (0.002) & 0.478 (0.014) & 0.973 (0.005) & 0.500 (0.000) & 0.500 (0.000) & 0.500 (0.000) & 0.489 (0.010) & 0.585 (3/50) \\
\scen{NCF02} & 0.603 (0.013) & 0.498 (0.017) & 0.441 (0.014) & 0.997 (0.001) & 0.893 (0.012) & 0.899 (0.012) & 0.612 (0.016) & 0.295 (3/50) \\
\scen{NCF03} & 0.672 (0.010) & 0.490 (0.010) & 0.505 (0.011) & 0.781 (0.012) & 0.512 (0.005) & 0.987 (0.003) & 0.814 (0.015) & 0.293 (8/50) \\
\scen{NCF04} & 0.518 (0.004) & 0.512 (0.004) & 0.506 (0.005) & 0.489 (0.008) & 0.509 (0.007) & 0.521 (0.010) & 0.517 (0.007) & 0.478 (1/50) \\
\scen{NCF05} & 0.682 (0.015) & 0.466 (0.018) & 0.490 (0.010) & 0.674 (0.017) & 0.482 (0.010) & 0.558 (0.015) & - & 0.327 (5/50) \\
\midrule
\multicolumn{9}{@{}l@{}}{\textbf{Partial observability: latent $Z$, stable regime (01)}} \\
\scen{NPO01} & 0.503 (0.004) & 0.500 (0.000) & 0.500 (0.000) & 0.500 (0.000) & 0.500 (0.000) & 0.500 (0.000) & 0.500 (0.000) & 0.511 (4/50) \\
\bottomrule
\end{tabular}%
}
\end{table*}

\begin{table*}[!htbp]
\centering
\scriptsize
\setlength{\tabcolsep}{4pt}
\renewcommand{\arraystretch}{1.15}
\caption{Synthetic control scenarios: p-values by method. Entries are median p-values with rejection counts (k/R) per method.}
\label{tab:scenarios_null_pvals:POOL_PERM}
\begin{tabularx}{\textwidth}{@{} >{\raggedright\arraybackslash}X ccccccc @{}}
\toprule
\textbf{Test name} & \multicolumn{7}{c}{\textbf{p-value (k/R)}} \\
\cmidrule(lr){2-8}
 & $\hat{Q}_{X,Y,Z}$ & \scen{XI\_AC} & \scen{CHSIC} & \scen{PDCORR} & \scen{PSPEAR} & \scen{PKENDALL} & \scen{KPCGRAPH} \\
\midrule
\multicolumn{8}{@{}l@{}}{\textbf{Baseline nulls: stationary null (01), multi-$X$ aligned with PMB03 (02)}} \\
\scen{NCL01} & 0.523 (1/50) & 0.492 (2/50) & 0.541 (1/50) & 0.415 (4/50) & 0.371 (5/50) & 0.406 (3/50) & 0.415 (5/50) \\
\scen{NCL02} & 0.994 (0/50) & 0.583 (2/50) & 0.019 (35/50) & 0.432 (4/50) & 0.588 (3/50) & 0.215 (12/50) & 0.055 (24/50) \\
\midrule
\multicolumn{8}{@{}l@{}}{\textbf{Invariance checks: global rescale (01), $Y$ monotone (02), $X$ monotone given $Z$ (03)}} \\
\scen{NIV01} & 0.272 (14/50) & 0.673 (0/50) & 0.002 (50/50) & 0.644 (1/50) & 0.798 (0/50) & 0.705 (0/50) & 0.471 (3/50) \\
\scen{NIV02} & 0.976 (0/50) & 0.601 (0/50) & 0.504 (2/50) & 0.016 (34/50) & 0.531 (2/50) & 0.480 (7/50) & 0.372 (6/50) \\
\scen{NIV03} & 0.724 (0/50) & 0.656 (1/50) & 0.433 (3/50) & 0.114 (16/50) & 0.506 (3/50) & 0.460 (3/50) & 0.363 (7/50) \\
\midrule
\multicolumn{8}{@{}l@{}}{\textbf{Marginal drift only: $Z$ loc/scale (01), $X$ mean shift (02), $Y$ trend (03), $Y$ seasonal (04)}} \\
\scen{NMD01} & 0.112 (18/50) & 0.498 (2/50) & 0.002 (50/50) & 0.397 (4/50) & 0.353 (4/50) & 0.433 (5/50) & 0.442 (5/50) \\
\scen{NMD02} & 0.559 (1/50) & 0.500 (2/50) & 0.360 (4/50) & 0.488 (3/50) & 0.514 (1/50) & 0.481 (2/50) & 0.526 (3/50) \\
\scen{NMD03} & 0.450 (3/50) & 0.519 (3/50) & 0.669 (2/50) & 0.233 (15/50) & 0.539 (2/50) & 0.504 (2/50) & 0.416 (4/50) \\
\scen{NMD04} & 0.519 (3/50) & 0.555 (0/50) & 0.410 (3/50) & 0.476 (5/50) & 0.420 (4/50) & 0.444 (4/50) & 0.457 (3/50) \\
\midrule
\multicolumn{8}{@{}l@{}}{\textbf{Noise-law / tail-shape: noise-law shift (01), tails shift (02)}} \\
\scen{NNS01} & 0.390 (4/50) & 0.499 (2/50) & 0.492 (5/50) & 0.520 (2/50) & 0.461 (0/50) & 0.590 (2/50) & 0.416 (3/50) \\
\scen{NNS02} & 0.286 (10/50) & 0.427 (2/50) & 0.410 (6/50) & 0.557 (3/50) & 0.514 (3/50) & 0.335 (6/50) & 0.505 (3/50) \\
\midrule
\multicolumn{8}{@{}l@{}}{\textbf{Covariate / confounding drift: $X|Z$ drift (01), $Z$ cov shift (02), $X|Z$ drift, copula (03), $F_Z$ drift only (04), discrete $Z$ (05)}} \\
\scen{NCF01} & 0.585 (3/50) & 0.572 (0/50) & 0.026 (34/50) & 0.102 (23/50) & 0.567 (0/50) & 0.419 (4/50) & 0.663 (0/50) \\
\scen{NCF02} & 0.295 (3/50) & 0.650 (3/50) & 0.524 (3/50) & 0.002 (48/50) & 0.033 (34/50) & 0.069 (23/50) & 0.438 (1/50) \\
\scen{NCF03} & 0.293 (8/50) & 0.517 (3/50) & 0.554 (3/50) & 0.358 (8/50) & 0.417 (4/50) & 0.002 (47/50) & 0.122 (16/50) \\
\scen{NCF04} & 0.478 (1/50) & 0.571 (1/50) & 0.627 (3/50) & 0.424 (4/50) & 0.426 (4/50) & 0.461 (5/50) & 0.457 (6/50) \\
\scen{NCF05} & 0.327 (5/50) & 0.369 (1/50) & 0.489 (1/50) & 0.006 (32/50) & 0.510 (2/50) & 0.322 (6/50) & 1.000 (0/50) \\
\midrule
\multicolumn{8}{@{}l@{}}{\textbf{Partial observability: latent $Z$, stable regime (01)}} \\
\scen{NPO01} & 0.511 (4/50) & 0.544 (2/50) & 0.495 (3/50) & 0.629 (4/50) & 0.556 (3/50) & 0.403 (5/50) & 0.467 (4/50) \\
\bottomrule
\end{tabularx}
\end{table*}

\subsection{Additional real data examples} \label{section: further real data examples}

In this section we present the analysis of three further financial time series using Algorithm~\ref{algorithm: multiple change point detection}. All tuning parameters are set identically to Section~\ref{section: real data examples}, with the exception of the window sizes winch are chosen to be $252$, $63$, and $126$ for each of the data examples respectively. We also provide figures showing the values taken by the test statistic, as well as estimated change point locations and associated p-values, in Figures~\ref{fig:S2_gfc_sustained},~\ref{fig:S3_covid_initial}, and~\ref{fig:S4_brexit}. The data examples, along with our findings, are described below. 

\paragraph{Scenario 2: Systemic Credit Crisis -- GFC.} 

We test a ``flight-to-quality / funding-stress'' mechanism between the S\&P 500 ($X$) and Investment Grade Corporate Bonds ($Y$), conditioning on implied volatility and interbank liquidity risk proxies ($Z$) \cite{cornett2011,brunnermeier2009,longstaff2010}. Analysed over a 2003--2016 runway, the detector flags highly significant ruptures prior to the Lehman Brothers collapse during the 2004 rate hike cycle and early 2007 subprime tremors. Conversely, it effectively filters out pure volatility events like the 2006 housing peak and the Bear Stearns crisis. Post-Lehman, the algorithm strongly identifies the 2008 collapse itself, alongside zero-rate era shocks such as the 2010 Eurozone crisis and early 2016 turbulence caused by the first rate hike. The 2011 U.S. Downgrade registers only as a marginally significant shift while transient panics like the 2013 ``Taper Tantrum'', sparked by fears of the Federal Reserve reducing its bond purchases, do not trigger causal breaks at all, further underscoring the metric's robustness.

\paragraph{Scenario 3: Market Liquidity Stress -- COVID-19.} 

We examine the conditional relationship between the S\&P 500 ($X$) and Gold ($Y$), conditioning on real yields and the U.S. Dollar Index ($Z$). Gold is a canonical hedge/safe-haven benchmark in the literature \cite{baur2010}, while COVID-era market stress generated a broad ``dash for cash'' episode that disrupted usual cross-asset behavior \cite{cheema2022}. The detector flags highly significant structural instabilities preceding the pandemic, including the 2015 China devaluation, the 2018 ``Volmageddon'' shock, and the 2019 repo crisis, culminating in the severe March 2020 COVID-19 rupture. Post-pandemic, the detector continues to capture major causal decoupling events tied to the 2022 Fed rate hikes and the 2023 Silicon Valley Bank collapse, culminating in an unprecedented structural shift in late 2023 as gold's safe-haven pricing completely detached from surging real yields.

\paragraph{Scenario 4: Sovereign Political Risk -- Brexit.} 

We analyse the relationship between the British Pound to U.S. Dollar exchange rate (GBP/USD) ($X$) and the UK FTSE 100 Index ($Y$), conditioning on global equity and risk-aversion baselines ($Z$). The economic intuition follows the Brexit/FTSE heterogeneity literature, in which sterling moves can mechanically reprice internationally exposed UK equities \cite{davies2018}. The detector flags marginally significant early-sample stress around the Eurozone crisis and the start of the referendum campaign, before registering its absolute maximum causal break directly aligning with the 2016 Brexit referendum period. Post-referendum, the algorithm continues to identify highly significant structural ruptures corresponding to the turbulent withdrawal timeline, culminating in a massive secondary spike in early 2022, coinciding with the geopolitical onset of the European energy crisis and the Bank of England's aggressive tightening cycle.

\begin{figure*}[!htbp]
    \centering
    \includegraphics[width=\textwidth]{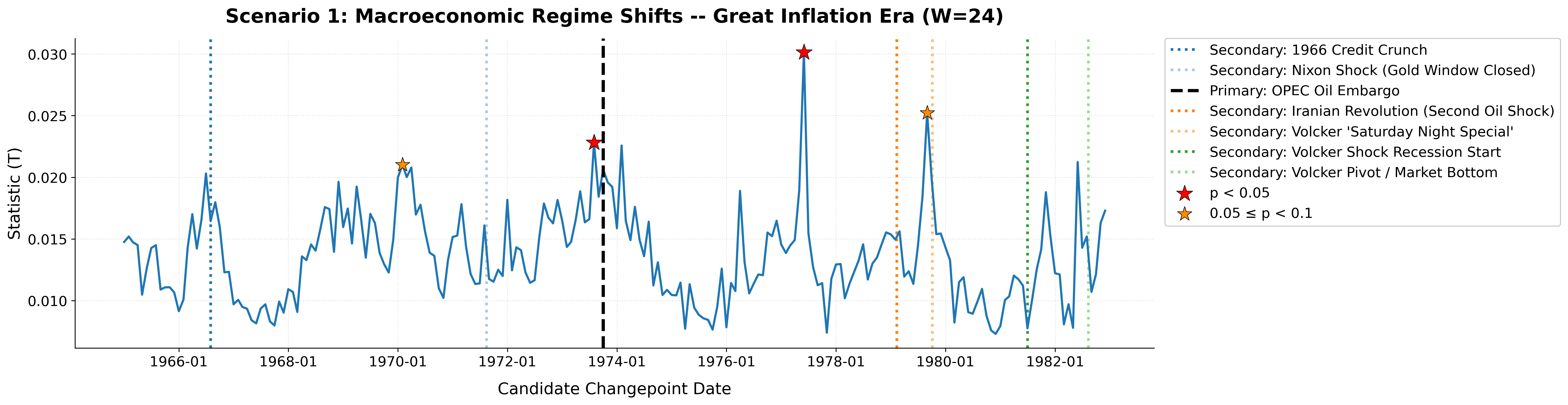}
    \caption{Causal Changepoint Detection: Scenario 1 (Great Inflation Era) -- red stars ($p < 0.05$) and orange stars ($0.05 \leq p < 0.1$) indicate estimated change p. Primary target events are marked with thick dashed black lines, while historical secondary events are marked with colored dotted lines.}
    \label{fig:S1_inflation}
\end{figure*}

\begin{figure*}[!htbp]
    \centering
    \includegraphics[width=\textwidth]{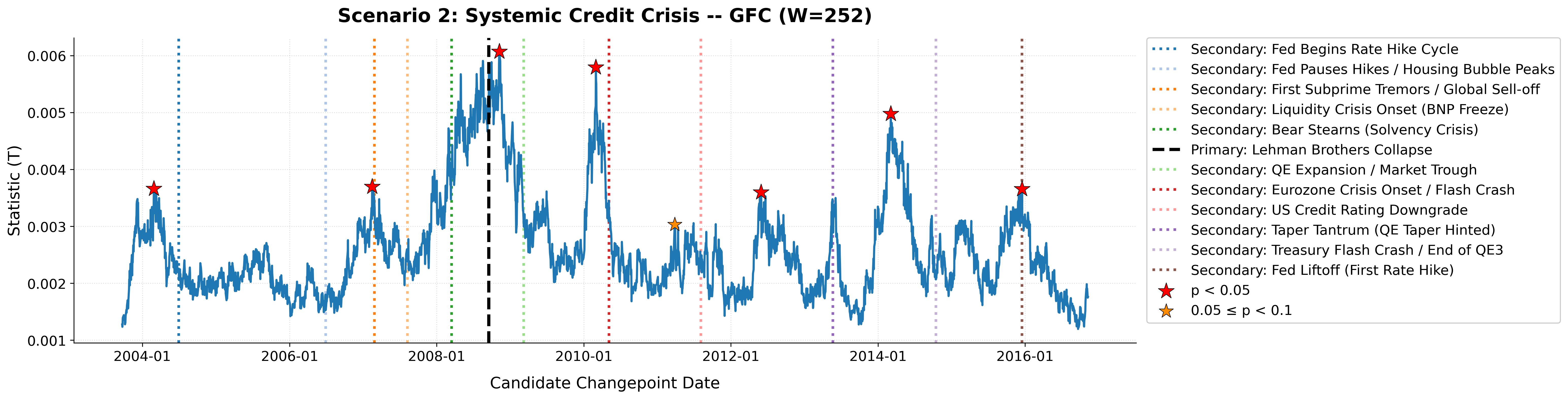}
    \caption{\textbf{Causal Changepoint Detection: Scenario 2 (GFC Sustained Break Variant).} Test statistic $T$ over time for the Global Financial Crisis using a longer window size ($W=252$) to capture sustained regime shifts. Red stars indicate $p < 0.05$ and orange stars indicate $0.05 \leq p < 0.1$.}
    \label{fig:S2_gfc_sustained}
\end{figure*}


\begin{figure*}[!htbp]
    \centering
    \includegraphics[width=\textwidth]{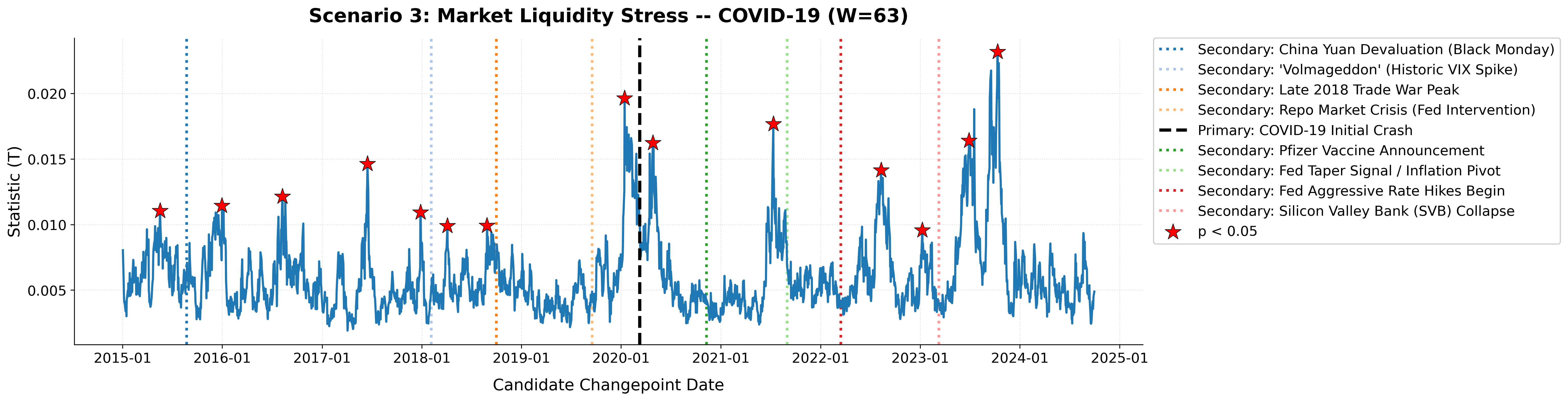}
    \caption{\textbf{Causal Changepoint Detection: Scenario 3 (COVID-19 Initial Break Variant).} Test statistic $T$ over time during the COVID-19 market stress using a shorter window size ($W=63$). Red stars indicate $p < 0.05$ and orange stars indicate $0.05 \leq p < 0.1$.}
    \label{fig:S3_covid_initial}
\end{figure*}

\begin{figure*}[!htbp]
    \centering
    \includegraphics[width=\textwidth]{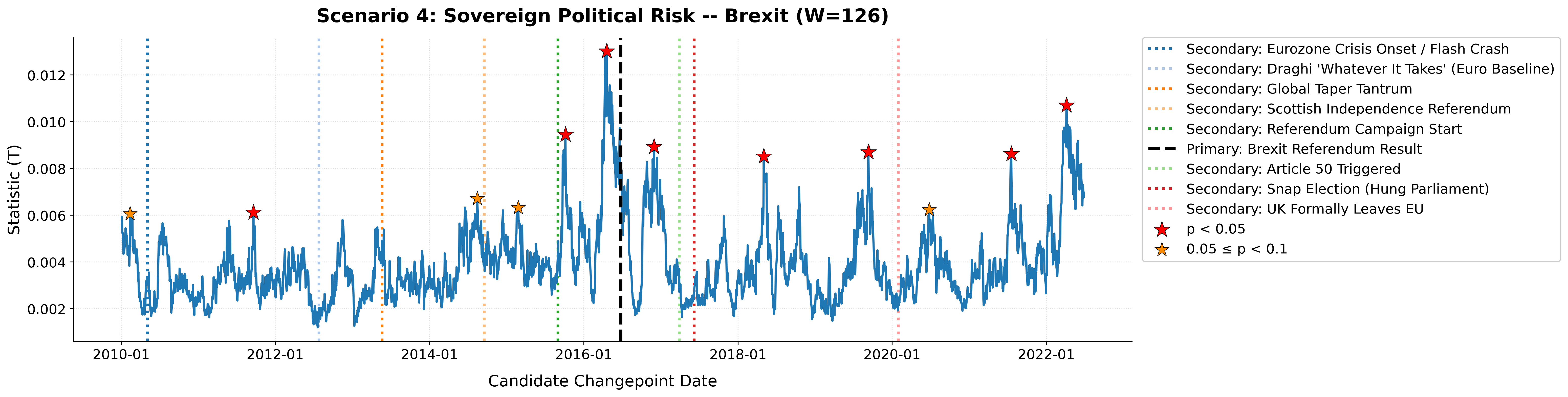}
    \caption{\textbf{Causal Changepoint Detection: Scenario 4 (Sovereign Political Risk -- Brexit).} Test statistic $T$ over time surrounding the Brexit referendum and subsequent policy shifts. Red stars indicate $p < 0.05$ and orange stars indicate $0.05 \leq p < 0.1$.}
    \label{fig:S4_brexit}
\end{figure*}

\subsection{Details on pre-processing in real data examples} \label{section: pre-processing}

In this section we provide details on the pre-processing steps carried out on data analysed in the real world examples presented in Sections~\ref{section: real data examples} and \ref{section: further real data examples}. The following transformations were applied to each dataset, in the order stated below. 

\paragraph{Stationarity Transformations} We transform each variable into a return series $X_i$ based on its asset class in the following way. For asset prices ($P_i$; stocks, commodities, indices, ETFs) we transformed prices to log returns via $X_i \coloneq \log (P_i) - \log (P_{i-1})$. For Rates, spreads, and volatility indices ($R_i$; yields, VIX, TED spread, etc.) we worked with the first difference of the series: $X_i \coloneq R_{i} - R_{i-1}$. 

\paragraph{Heteroskedasticity Adjustment}
We normalize each transformed series using a recursive Exponentially Weighted Moving Average volatility estimator. For a span $S$ (daily data: $S=63$ trading days, approximately one quarter; monthly data: $S=12$ months), define $\alpha = \frac{2}{S+1}$. the variance was then estimated locally at each time step $i$ via $\hat{\sigma}_i^2 = (1-\alpha)\hat{\sigma}_{t-1}^2 + \alpha (X_i - \bar{\mu}_i)^2$, and we worke with the series $\tilde{X}_i = X_i / \hat{\sigma}_t + \epsilon$ where $\epsilon=10^{-6}$ was included for numerical stability.

\section{Necessity of Assumption \ref{assumption: regression; parent; support}}
\label{appendix: assumptions neessity}

To see why Assumption \ref{assumption: regression; parent; support} is necessary, we construct examples of distinct causal models that induce the same observational distribution but imply different interventional effects. In such cases, the observational law does not uniquely determine the interventional distribution, and causal reasoning from observational data alone becomes impossible.

\subsubsection{Necessity of part (i) of Assumption \ref{assumption: regression; parent; support}}
We construct two distinct causal models that induce the same joint Gaussian law. One of these models violates part (i) of Assumption \ref{assumption: regression; parent; support}. First, we consider a model $\{X, Y, \boldsymbol Z \}$ which we refer to as Model A. This model is such that $X \sim \mathcal{N}(0,1)$, $Y = X + \varepsilon$ with $ \varepsilon\sim\mathcal{N}(0,1)$ independent noise, and $\boldsymbol Z \sim \mathcal{N}(0,1)$. In this case, the observational covariance matrix between $X$, $Y$, and $\boldsymbol Z$ is 
\[\Sigma = \begin{bmatrix}1 & 1 & 0\\ 1 & 2 & 0 \\ 0 & 0 & 1\end{bmatrix}.
\]

Next we consider a second model $\{X', Y', \boldsymbol Z' \}$ which we refer to as Model B. This model is such that
$Y' \sim \mathcal{N}(0,2)$, $ X' = 0.5\,Y + \varepsilon'$, with $\varepsilon ' \sim\mathcal{N}(0,0.5)$ independent posterior noise, and $\boldsymbol Z'  \sim \mathcal{N}(0,1)$. Note that this model violates part (i) of Assumption \ref{assumption: regression; parent; support}. Then the induced covariance matrix is 
\[\Sigma' = \begin{bmatrix}1 & 1 & 0\\ 1 & 2 & 0 \\ 0 & 0 & 1\end{bmatrix}.
\]

Hence Models A and B are \emph{observationally indistinguishable}, in the sense that they induce the same joint Gaussian law. However, they imply different interventional behaviour:
in Model A, intervening on \(X\) changes \(Y\) (because \(Y\) depends on \(X\)), whereas in Model B, intervening on \(X\) does not affect \(Y\), since \(Y\) is exogenous.

\subsubsection{Necessity of part (ii) of Assumption \ref{assumption: regression; parent; support}}
We construct two distinct causal models that induce the same joint Gaussian law on the observed random variables. One of these models violates part (ii) of Assumption \ref{assumption: regression; parent; support}. First, we consider a model $\{X, Y, \boldsymbol Z \}$ which we refer to as Model A. This model violates part (ii) of Assumption \ref{assumption: regression; parent; support}, and it is defined as
\[
X = W + \varepsilon_X,\ \varepsilon_X\sim\mathcal{N}(0,1),\quad Y = W + \varepsilon_Y,\ \varepsilon_Y\sim\mathcal{N}(0,1), \quad \boldsymbol Z  \sim \mathcal{N}(0,1),
\]
with $W \sim \mathcal{N}(0,1)$ a fourth unobserved random variable. Then the observational covariance between $X$, $Y$, $\boldsymbol Z$ is 
\[
\Sigma = \begin{bmatrix}2 & 1 & 0\\ 1 & 2 & 0 \\ 0 & 0 & 1\end{bmatrix}.
\]

We consider a second model $\{X', Y', \boldsymbol Z' \}$ which we refer to as Model B. This model is such that 
$X'\sim\mathcal{N}(0,2)$, $Y' = 0.5 X + \varepsilon$ with $ \varepsilon\sim\mathcal{N}(0,1.5)$, and $\boldsymbol Z' \sim \mathcal{N}(0,1)$. Then
the observational covariance between $X'$, $Y'$, $\boldsymbol Z'$ is 
\[
\Sigma' = \begin{bmatrix}2 & 1 & 0\\ 1 & 2 & 0 \\ 0 & 0 & 1\end{bmatrix}.
\]

The observed random variables in Models A and B are observationally identical, but they differ under intervention. In fact, in Model C, \(Y\) does not respond to an intervention on \(X\). However, in Model D, intervening on \(X\) changes the distribution of \(Y\).

\section{Proofs}

In this section we present proof of the results stated in the main paper. External results and definitions needed for the proofs are given in Section~\ref{section: external results}, and the proofs are given in the subsequent section. 

\subsection{External results and definitions} \label{section: external results}

\begin{definition}[Graph powers]
For a graph $\mathcal{G}$ with vertex set $\mathcal{V}$ and edge set $\mathcal{E}$ the $k$-th power, written $\mathcal{G}^k$, is defined as the graph obtained by adding an edge between any two vertices in $\mathcal{E}$ which are connected by a walk of length at most $k$.
\label{definition: graph powers}
\end{definition}

\begin{theorem}[Vizing's theorem; \cite{vizing1964estimate}]
Every simple undirected graph may be edge coloured using a number of colours that is at most one larger than the maximum degree of the graph, where the maximum degree denotes the maximum number of neighbours of any vertex is the graph.
\end{theorem}

\begin{lemma}[Nearest-neighbour distance; \cite{roudaki2026kernel}]
Let $X_1, \ldots, X_n$ be distributed i.i.d. on a metric space $(\mathcal{X}, d_\mathcal{X})$, and let $d_j$ be the distance from $X_j$ to its $k$-th nearest neighbour among $\left\{X_{\ell}: \ell \neq j\right\}$, i.e. $d_j \coloneq \max _{k \in \mathcal{N}_j} d_\mathcal{X}(X_k, X_j)$ where $\mathcal{N}_j$ is the set of $k$-nearest neighbours of $X_j$ in $\{X_k\}_{k\neq j}$. Assume that 
\begin{enumerate}[label=(\roman*)]
    \item there exist $x^* \in \mathcal{X}$ and $\alpha, C_1, C_2>0$ such that for all $t \geq 0$ it holds that $\mathbb{P}\left(d_\mathcal{X}(X_1, x^*) \geq t\right) \leq C_1 e^{-C_2 t^\alpha}$, and 
    \item there exist constants $d>0$ and $c_0>0$ such that for every radius $T>0$, for all $x\in \mathcal X$ with $d_\mathcal{X}(x^*, x) \leq T$ and all $r>0$ it holds that $\mathbb{P}\left(d_\mathcal{X}(X_1, x) \leq r\right) \geq c_0 r^d$. 
\end{enumerate}
Let $\beta>0$. Then for all $n \geq 3$ and $1 \leq k \leq n / 2$ it holds that
\begin{equation*}
    \mathbb{E}[d_j^\beta] \lesssim \left(\frac{k}{n}\right)^{\beta / d} + \frac{(\log n)^{\beta / \alpha}}{n^2} . 
\end{equation*}
\label{lemma: knn expectation bound}
\end{lemma}

\begin{corollary}[Nearest-neighbour distance]
Let $X_1, \ldots, X_n$ be as in Lemma~\ref{lemma: knn expectation bound}. Let the $X$'s be distributed according to a measure $\mu$ on $\mathcal{X}$, and let $Y$ independent of the $X$'s and distributed according to a different measure $\nu$ such that $\int_A \mathrm{d} \mu \asymp \int_A \mathrm{d} \mu$ for any $A \in \sigma (\mathcal{X})$. Let $d_Y$ denote the $k$ nearest neighbours of $Y$ among the $X$'s. For $n \geq 3$ and $1 \leq k \leq n / 2$ it holds that 
\begin{equation*}
    \mathbb{E}[d_Y^\beta] \lesssim \left(\frac{k}{n}\right)^{\beta / d} + \frac{(\log n)^{\beta / \alpha}}{n^2} . 
\end{equation*}
for any $\beta > 0$. 
\label{corollary: Y knn expectation bound}
\end{corollary}

\begin{theorem}[Bounded differences inequality; \citealt{boucheron2003concentration}] \label{theorem: bounded difference}
Let $\mathcal{X}$ be a measurable space. A function $f: \mathcal{X}^n \to \mathbb{R}$ has the bounded difference property for some constants $c_1, \dots, c_n$ if, for each $i = 1, \dots, n$,
\begin{equation}
\sup_{\substack{x_1, \dots, x_n \\ x_i' \in \mathcal{X}}} \left | f \left ( x_1, \dots, x_{i-1}, x_i, x_{i+1}, x_n \right ) - f \left ( x_1, \dots, x_{i-1}, x_i', x_{i+1}, \dots x_n \right ) \right | \leq c_i. 
\label{equation: bounded difference property}
\end{equation}
Then, if $X_1, \dots, X_n$ is a sequence of independently distributed random variables and (\ref{equation: bounded difference property}) holds, putting $Z = f \left ( X_1, \dots, X_n \right )$ and $\nu = \frac{1}{4} \sum_{i=1}^n c_i^2$, for any $t > 0$, it holds that
\begin{equation*}
\mathbb{P} \left ( Z - \mathbb{E} \left ( Z \right ) > t \right ) \leq e^{-t^2 / \left ( 2 \nu \right )}.  
\end{equation*}
\end{theorem}

\subsection{Proof of Theorem \ref{theorem: copula causality}}
\label{appeendix : copula causality}
\begin{proof}
We first show that the following equations hold.
\begin{equation}
\label{eq:2_new}
    \mathbb{P}(Y \mid do(X = x, \boldsymbol{Z} = \boldsymbol{z})) = \mathbb{P}(Y \mid X = x, \boldsymbol{Z}= \boldsymbol{z}) 
\end{equation}
\begin{equation}
\label{eq:3_new}
    \mathbb{P}(Y' \mid do(X' = x, \boldsymbol{Z}' = \boldsymbol{z})) = \mathbb{P}(Y' \mid X '= x, \boldsymbol{Z}'= \boldsymbol{z}) 
\end{equation}
We only show the proof for Eq. \eqref{eq:2_new} , since the proof for Eq. \eqref{eq:3_new} is analogous.
%
%
%
First, note that $ \mathsf{Pa}_Y \subseteq X \cup \boldsymbol{Z}$.
By Rule 2 of the do-calculus~\citep[page~85]{pearlj}, Eq. \eqref{eq:2_new} hold, because $Y$ becomes independent of $\{X\} \cup \boldsymbol{Z}$ once all arrows from $ \mathsf{Pa}_Y$ to $Y$ are removed from the graph of the DGP. Hence, the claim follows by the definition of conditional copulas.
%
\end{proof}

\subsection{Proof of proposition \ref{proposition: Q separation}}

\begin{proof}[Proof of Part (i)]
Given such an $A_{\boldsymbol{z}}$   any $\boldsymbol{z} \in A_{\boldsymbol{z}}$ it holds that 
\begin{equation}
    \left \| \mu_{K} \left ( \mathcal{C}_{X,Y \mid \boldsymbol{Z = z }} \right ) - \mu_{K} \left ( \mathcal{C}_{X', Y' \mid \boldsymbol{Z' = z}} \right ) \right \|_{\mathcal{H}_K}^2 = \operatorname{MMD}_K^2 \left [ \mathcal{C}_{X,Y \mid \boldsymbol{Z = z }}, \mathcal{C}_{X,Y \mid \boldsymbol{Z' = z }} \right ] > 0. 
    \label{equation: copula MMD is positive}
\end{equation}
The equality follows from Lemma 4 in \cite{gretton2012mmd}, and the inequality follows from the facts that (i) $\mathcal{C}_{X,Y \mid \boldsymbol{Z = z }}$ and $\mathcal{C}_{X,Y \mid \boldsymbol{Z = z' }}$ are probability measures on $[0,1]^2$, (ii) by Theorem~\ref{theorem: copula causality} the two measures are not identical, and (iii) when $K$ is a characteristic kernel $\operatorname{MMD}_K$  metrizes
the space of Borel probability measures on $[0,1]^2$ \citep{sriperumbudur2010hilbert}. Therefore we have that 
\begin{align*}
Q_{X,Y,Z} & = \int_{\mathbb{R}^d}  \operatorname{MMD}_K^2 \left [ \mathcal{C}_{X,Y \mid \boldsymbol{Z = z }}, \mathcal{C}_{X,Y \mid \boldsymbol{Z' = z }} \right ] \mathrm{d} \Lambda (\boldsymbol{z}) \\
& \geq \int_{A_{\boldsymbol{z}}}  \operatorname{MMD}_K^2 \left [ \mathcal{C}_{X,Y \mid \boldsymbol{Z = z }}, \mathcal{C}_{X,Y \mid \boldsymbol{Z' = z }} \right ] \mathrm{d} \Lambda (\boldsymbol{z}) > 0. 
\end{align*}
Where the weak inequality is due to the facts that (i) $\operatorname{MMD}_K^2 \geq 0$ and (ii) $A_{\boldsymbol{z}} \subseteq \mathbb{R}^d$, and the strict inequality is due to the fact that (i) by \eqref{equation: copula MMD is positive} the integrand is positive and (ii) $\Lambda (A_{\boldsymbol{z}}) > 0$ since $A_{\boldsymbol{z}}$ is $\mathbb{P}_{\boldsymbol{Z}}, \mathbb{P}_{\boldsymbol{Z}'}$-measurable and $\Lambda \gg \mathbb{P}_{\boldsymbol{Z}}, \mathbb{P}_{\boldsymbol{Z}'}$. 
\end{proof}

\begin{proof}[Proof of Part (ii)]
Since \eqref{equation: do equivalence} holds for all $\boldsymbol{z} \in \operatorname{supp} (\mathbb{P}_{\boldsymbol{Z}}), \operatorname{supp} (\mathbb{P}_{\boldsymbol{Z}'})$ by Theorem~\ref{theorem: copula causality} we must have that $\mathcal{C}_{X,Y \mid \boldsymbol{Z = z }}$ and $\mathcal{C}_{X,Y \mid \boldsymbol{Z' = z }}$ are identical for every such $\boldsymbol{z}$. Consequently 
\begin{equation}
 \operatorname{MMD}_K^2 \left [ \mathcal{C}_{X,Y \mid \boldsymbol{Z = z }}, \mathcal{C}_{X,Y \mid \boldsymbol{Z' = z }} \right ] = 0
\end{equation}
since $\operatorname{MMD}_K$ is a metric when $K$ is characteristic (see above). Then we immediately have that
\begin{equation*}
\int_{\mathbb{R}^d}  \operatorname{MMD}_K^2 \left [ \mathcal{C}_{X,Y \mid \boldsymbol{Z = z }}, \mathcal{C}_{X,Y \mid \boldsymbol{Z' = z }} \right ] \mathrm{d} \Lambda (\boldsymbol{z}) = 0, 
\end{equation*}
which establishes the desired result. 
\end{proof}

\subsection{Proof of Theorem \ref{theorem: rate of convergence}}

\begin{proof}
We begin by establishing the rate of convergence of the first double sum in $\hat{T}_1$, namely
\begin{equation*}
\hat{V}_{1,1} \coloneq \frac{1}{2\eta} \sum_{i_1 =1}^{\eta} \frac{2}{k_n(k_n-1)} \sum_{\substack{i_2, i_3 \in \mathcal{N} (\boldsymbol{Z}_{i_1}) \\ i_2 < i_3}} \hat{\mathbb{K}}_{X_{i_2}, Y_{i_2}, X_{i_3}, Y_{i_3}}^{\boldsymbol{Z}_{i_1}},
\end{equation*}
to the term 
\begin{equation*}
V^*_{1,1} \coloneq \frac{1}{2} \int \mathbb{E}_{U_1, U_2 \sim \mathcal{C}_{X,Y \mid \boldsymbol{Z} = \boldsymbol{z}}} \left [ K \left ( U_1, U_2 \right ) \right ] \mathrm{d} \mathbb{P}_{\boldsymbol{Z}} (\boldsymbol{z}) 
\end{equation*}
which corresponds to one of the two integrals comprising first term in the expansion of $\hat{Q}_{X,Y,Z}$ in equation \eqref{equation: Q expansion} when integrating against $\Lambda = \frac{1}{2} \mathbb{P}_Z + \frac{1}{2} \mathbb{P}_{Z'}$. Towards this end, for each $i \in [n]$ introduce the (random) distribution functions
\begin{align*}
& \tilde{F}_{X \mid \boldsymbol{Z = Z_i}} (\cdot) = \frac{1}{k_n} \sum_{l \in \mathcal{N}_{Z_i}^{Z_{1:\eta}}} F_{X \mid \boldsymbol{Z = Z_l}} (\cdot) \\
& \tilde{F}_{Y \mid \boldsymbol{Z = Z_i}} (\cdot) = \frac{1}{k_n} \sum_{l \in \mathcal{N}_{Z_i}^{Z_{1:\eta}}} F_{Y \mid \boldsymbol{Z = Z_l}} (\cdot).
\end{align*}
Consequently for each $i,j \in [n]$ define
\begin{align*}
& \tilde{U}_{X_j, Y_j}^{\boldsymbol{Z}_i} \coloneqq \left ( \tilde{F}_{X \mid \boldsymbol{Z = Z_i}} (X_j), \tilde{F}_{Y \mid \boldsymbol{Z = Z_i}} (Y_j) \right )^\top \\
& \check{U}_{X_i, Y_i}^{\boldsymbol{Z}_i} \coloneqq \left ( F_{X \mid \boldsymbol{Z = Z_i}} (X_j), F_{Y \mid \boldsymbol{Z = Z_i}} (Y_j) \right )^\top,  
\end{align*}
and moreover introduce the quantities
\begin{align*}
& \tilde{V}_{1,1} = \frac{1}{2\eta} \sum_{i_1 =1}^{\eta} \frac{2}{k_n(k_n-1)} \sum_{\substack{i_2, i_3 \in \mathcal{N} (\boldsymbol{Z}_{i_1}) \\ i_2 < i_3}} K \left ( \tilde{U}_{X_{i_2}, Y_{i_2}}^{\boldsymbol{Z}_{i_1}} , \tilde{U}_{X_{i_3}, Y_{i_3}}^{\boldsymbol{Z}_{i_1}} \right ) \\
& \hspace{5em} \coloneq \frac{1}{2\eta} \sum_{i_1 =1}^{\eta} \frac{2}{k_n(k_n-1)} \sum_{\substack{i_2, i_3 \in \mathcal{N}_{Z_{i_1}}^{Z_{1:\eta}} \\ i_2 < i_3}} \tilde{\mathbb{K}}_{X_{i_2}, Y_{i_2}, X_{i_3}, Y_{i_3}}^{\boldsymbol{Z}_{i_1}}, \\
& \check{V}_{1,1} = \frac{1}{2\eta} \sum_{i_1 =1}^{\eta} \frac{2}{k_n(k_n-1)} \sum_{\substack{i_2, i_3 \in \mathcal{N} (\boldsymbol{Z}_{i_1}) \\ i_2 < i_3}} K \left ( \check{U}_{X_{i_2}, Y_{i_2}}^{\boldsymbol{Z}_{i_1}} , \check{U}_{X_{i_3}, Y_{i_3}}^{\boldsymbol{Z}_{i_1}} \right ) \\
& \hspace{5em} \coloneq \frac{1}{2\eta} \sum_{i_1 =1}^{\eta} \frac{2}{k_n(k_n-1)} \sum_{\substack{i_2, i_3 \in \mathcal{N}(\boldsymbol{Z}_{i_1}) \\ i_2 < i_3}} \check{\mathbb{K}}_{X_{i_2}, Y_{i_2}, X_{i_3}, Y_{i_3}}^{\boldsymbol{Z}_{i_1}}. 
\end{align*}
By the triangle inequality we have that
\begin{align}
\left | \hat{V}_{1,1} - V^*_{1,1} \right | \leq \left | \hat{V}_{1,1} - \tilde{V}_{1,1} \right | + \left | \tilde{V}_{1,1} - \check{V}_{1,1} \right | + \left | \check{V}_{1,1} - V^*_{1,1}\right |, 
\label{equation: Q decomposition}
\end{align}
and we stochastically bound each of the terms in \eqref{equation: Q decomposition} in turn. For the sample $\{ \boldsymbol{Z}_1, \dots, \boldsymbol{Z}_\eta \}$ let $\mathcal{F}_Z$ be the corresponding sigma field, let $\mathcal{G}_n$ be the $k_n$-nearest neighbour graph based on the sample, and let $\mathcal{G}_n^2$ be its second power -- see Definition~\ref{definition: graph powers}. Denote it's edge and vertex set by $E(\mathcal{G}_n^2)$ and $V(\mathcal{G}^2)$ respectively. Let $\{ c_1, \dots, c_{\chi (\mathcal{G}_n^2)} \}$ be a valid colouring of the edges, where $\chi (\mathcal{G}_n^2)$ denotes the edge chromatic number, and for each $l \in [\chi(\mathcal{G}_n^2)]$ let $\mathcal{C}_l$ denote the set of edges coloured with $c_l$. In the sequel we will write $a_n \lesssim b_n$ for two sequences $\{ a_n \}_{n > 0}$ and $\{ b_n \}_{n > 0}$ if there exist an absolute $C > 0$ such that $|a_n| \leq C |b_n|$ for all $n > 0$. For the first term in \eqref{equation: Q decomposition} therefore have that
\begin{align}
& \mathbb{E} \left [ \left | \hat{V}_{1,1} - \tilde{V}_{1,1} \right | \right ] \nonumber \\
&  \lesssim \sum_{l \leq 1 \leq \chi (\mathcal{G}_n^2)} \mathbb{E} \left [ \frac{1}{\eta k_n^2} \left | \sum_{\{i_1, i_2, i_3 \} \in \mathcal{C}_l }  \hat{\mathbb{K}}_{X_{i_2}, Y_{i_2}, X_{i_3}, Y_{i_3}}^{\boldsymbol{Z}_{i_1}} -  \tilde{\mathbb{K}}_{X_{i_2}, Y_{i_2}, X_{i_3}, Y_{i_3}}^{\boldsymbol{Z}_{i_1}} \right | \right ] \nonumber \\
& = \sum_{l \leq 1 \leq \chi (\mathcal{G}_n^2)} \mathbb{E} \left [ \int_0^\infty \mathbb{P} \left ( \frac{1}{\eta k_n^2} \left | \sum_{\{i_1, i_2, i_3 \} \in \mathcal{C}_l } \hat{\mathbb{K}}_{X_{i_2}, Y_{i_2}, X_{i_3}, Y_{i_3}}^{\boldsymbol{Z}_{i_1}} -  \tilde{\mathbb{K}}_{X_{i_2}, Y_{i_2}, X_{i_3}, Y_{i_3}}^{\boldsymbol{Z}_{i_1}} \right | > t \mid \mathcal{F}_Z \right ) \mathrm{d}t  \right ].
\label{equation: Q conditioning bound}
\end{align}
For each $i_1, i_2, i_3 \in [n]$ let $\bar{U}_{X_{i_2} Y_{i_2}}^{\boldsymbol{Z}_{i_1}}$ and $\bar{U}_{X_{i_3}, Y_{i_3}}^{\boldsymbol{Z}_{i_1}}$ lie respectively between $\hat{U}_{X_{i_2}, Y_{i_2}}^{\boldsymbol{Z}_{i_1}}$ and $\tilde{U}_{X_{i_2}, Y_{i_2}}^{\boldsymbol{Z}_{i_1}}$ and between $\hat{U}_{X_{i_3}, Y_{i_3}}^{\boldsymbol{Z}_{i_1}}$ and $\tilde{U}_{X_{i_3}, Y_{i_3}}^{\boldsymbol{Z}_{i_1}}$. By a first order expansion around $ ( \tilde{U}_{X_{i_2}, Y_{i_2}}^{\boldsymbol{Z}_{i_1}} , \tilde{U}_{X_{i_3}, Y_{i_3}}^{\boldsymbol{Z}_{i_1}} )^\top$ we therefore have that 
\begin{align}
& \hat{\mathbb{K}}_{X_{i_2}, Y_{i_2}, X_{i_3}, Y_{i_3}}^{\boldsymbol{Z}_{i_1}} -  \tilde{\mathbb{K}}_{X_{i_2}, Y_{i_2}, X_{i_3}, Y_{i_3}}^{\boldsymbol{Z}_{i_1}} \nonumber \\
& \quad = \nabla K \left ( \bar{U}_{X_{i_2}, Y_{i_2}}^{\boldsymbol{Z}_{i_1}} , \bar{U}_{X_{i_3}, Y_{i_3}}^{\boldsymbol{Z}_{i_1}} \right )^\top \left ( \left ( \hat{U}_{X_{i_2}, Y_{i_2}}^{\boldsymbol{Z}_{i_1}} - \tilde{U}_{X_{i_2}, Y_{i_2}}^{\boldsymbol{Z}_{i_1}} \right ) ^ \top, \left ( \hat{U}_{X_{i_3}, Y_{i_3}}^{\boldsymbol{Z}_{i_1}} - \tilde{U}_{X_{i_3}, Y_{i_3}}^{\boldsymbol{Z}_{i_1}} \right )^\top \right )^\top \nonumber \\
& \quad \coloneqq \xi^X_{{i_1},{i_2}} + \xi_{i_1, i_2}^Y + \xi^X_{{i_1},{i_3}} + \xi_{i_1, i_23}^Y \label{equation: xi decomposition}
\end{align}
where, with $[\boldsymbol{v}]_k$ denoting $k$-th entry of a vector $\boldsymbol{v}$ we have put
\begin{equation*}
    \xi^X_{{i_1},{i_2}} = \left [ \nabla K \left ( \bar{U}_{X_{i_2}, Y_{i_2}}^{\boldsymbol{Z}_{i_1}} , \bar{U}_{X_{i_3}, Y_{i_3}}^{\boldsymbol{Z}_{i_1}} \right ) \right ]_1 \times \left (  \frac{1}{k_n} \sum_{l \in \mathcal{N}(\boldsymbol{Z}_{i_1})} \mathbf{1}_{\left \{ X_l \leq X_{i_2} \right \}} - \frac{1}{k_n} \sum_{l \in \mathcal{N}(\boldsymbol{Z}_{i_1})} F_{X \mid \boldsymbol{Z = Z_l}} (X_{i_2})\right ), 
\end{equation*}
with the remaining $\xi$'s being defined analogously. Introduce the (random) quantities
\begin{equation*}
    L_{X,2}^{(l)} = \frac{1}{\eta k_n^2} \left | \sum_{\{i_1, i_2, i_3 \} \in \mathcal{C}_l } \mathbb{E} \left [ \xi_{i_1, i_2}^{X} \mid \mathcal{F}_Z \right ] \right | \quad \text{for} \quad l \in [\chi (\mathcal{G}_n^2)],
\end{equation*}
and let $L_{Y,2}^{(l)}, L_{X,3}^{(l)}, L_{Y,3}^{(l)}$ be defined analogously. Making use of \eqref{equation: xi decomposition} we obtain that
\begin{align}
& \int_0^\infty \mathbb{P} \left ( \frac{1}{\eta k_n^2} \left | \sum_{\{i_1, i_2, i_3 \} \in \mathcal{C}_l }  \hat{\mathbb{K}}_{X_{i_2}, Y_{i_2}, X_{i_3}, Y_{i_3}}^{\boldsymbol{Z}_{i_1}} -  \tilde{\mathbb{K}}_{X_{i_2}, Y_{i_2}, X_{i_3}, Y_{i_3}}^{\boldsymbol{Z}_{i_1}} \right | > t \mid \mathcal{F}_Z \right ) \mathrm{d}t \label{equation: Q integral bound} \\
& \quad \leq  4 L_{X,2}^{(l)} +  \int_{4 L_{X,2}^{(l)}}^\infty \mathbb{P} \left ( \frac{1}{\eta k_n^2} \left | \sum_{\{i_1, i_2, i_3 \} \in \mathcal{C}_l } \xi_{i_1, i_2}^{X} - \mathbb{E} \left [ \xi_{i_1, i_2}^{X} \mid \mathcal{F}_Z \right ] \right | > \frac{t}{4} - L_{X,2}^{(l)} \mid \mathcal{F}_Z \right ) \mathrm{d}t \nonumber \\
& \quad \quad + 4 L_{Y,2}^{(l)} +  \int_{4 L_{Y,2}^{(l)}}^\infty \mathbb{P} \left ( \frac{1}{\eta k_n^2} \left | \sum_{\{i_1, i_2, i_3 \} \in \mathcal{C}_l } \xi_{i_1, i_2}^{Y} - \mathbb{E} \left [ \xi_{i_1, i_2}^{Y} \mid \mathcal{F}_Z \right ] \right | > \frac{t}{4} - L_{Y,2}^{(l)} \mid \mathcal{F}_Z \right ) \mathrm{d}t \nonumber \\
& \quad \quad +  4 L_{X,3}^{(l)} +  \int_{4 L_{X,3}^{(l)}}^\infty \mathbb{P} \left ( \frac{1}{\eta k_n^2} \left | \sum_{\{i_1, i_2, i_3 \} \in \mathcal{C}_l } \xi_{i_1, i_3}^{X} - \mathbb{E} \left [ \xi_{i_1, i_3}^{X} \mid \mathcal{F}_Z \right ] \right | > \frac{t}{4} - L_{X,3}^{(l)} \mid \mathcal{F}_Z \right ) \mathrm{d}t \nonumber \\
& \quad \quad + 4 L_{Y,3}^{(l)} +  \int_{4 L_{Y,3}^{(l)}}^\infty \mathbb{P} \left ( \frac{1}{\eta k_n^2} \left | \sum_{\{i_1, i_2, i_3 \} \in \mathcal{C}_l } \xi_{i_1, i_3}^{Y} - \mathbb{E} \left [ \xi_{i_1, i_3}^{Y} \mid \mathcal{F}_Z \right ] \right | > \frac{t}{4} - L_{Y,3}^{(l)} \mid \mathcal{F}_Z \right ) \mathrm{d}t. \nonumber
\end{align}
For any $l \in [\chi (\mathcal{G}_n^2)]$ the $\hat{U}$'s and $\tilde{U}$'s with $\{i_1, i_2, i_3\}$'s indexed on the same $\mathcal{C}_l$ are conditionally independent given $\mathcal{F}_Z$. Consequently, the same holds for the $\xi$'s. Since the kernel $K$ is assumed to be bounded the $\xi$'s are likewise bounded, and consequently by Hoeffding's inequality we obtain that 
\begin{align*}
& \int_{4 L_{X,2}^{(l)}}^\infty \mathbb{P} \left ( \frac{1}{\eta k_n^2} \left | \sum_{\{i_1, i_2, i_3 \} \in \mathcal{C}_l } \xi_{i_1, i_2}^{X} - \mathbb{E} \left [ \xi_{i_1, i_2}^{X} \mid \mathcal{F}_Z \right ] \right | > \frac{t}{4} - L_{X,2}^{(l)} \mid \mathcal{F}_Z \right ) \mathrm{d}t \\
& \quad \leq \int_0^\infty \exp \left ( - C_1 \frac{t^2 \eta^2 k_n^4}{\left | \mathcal{C}_l \right |}\right ) \mathrm{d} t \lesssim \frac{\sqrt{\left | \mathcal{C}_l \right |}}{\eta k_n^2}
\end{align*}
where $C_1$ is an absolute constant depending on $\sup_{u,u'} |K(u,u')|$. Due to the $\xi$'s being bounded we further have that
\begin{equation}
    L_{X,2}^{(l)} \lesssim \frac{\left | \mathcal{C}_l \right |}{\eta k_n^2}
    \label{equation: L bound}
\end{equation}
independently of $\mathcal{F}_Z$. The same bounds hold for the remaining terms in \eqref{equation: Q integral bound}. By Vizing's theorem we have that $\chi(\mathcal{G}_n^2)$ is bounded from above by one plus the maximum degree of $\mathcal{G}_n^2$, which must be no larger than $k_n^2$. Moreover each $\mathcal{C}_l$ must have cardinality no larger than $\mathcal{O}(\eta / k_n^2)$. Combining these facts, we obtain that
\begin{equation}
    \mathbb{E} \left [ \left | \hat{V}_{1,1} - \tilde{V}_{1,1} \right | \right ] \lesssim \sum_{1 \leq l \leq [\chi(\mathcal{G}_n^2)]} \left \{ \frac{\sqrt{\left | \mathcal{C}_l \right |}}{\eta k_n^2} + \frac{\left | \mathcal{C}_l \right |}{\eta k_n^2} \right \} \leq \frac{k_n}{\sqrt{\eta}} + \frac{1}{k_n^4} \lesssim \frac{k_n}{\sqrt{n}} + \frac{1}{k_n^4},
    \label{equation: V hat V tilde bound}
\end{equation}
where the final inequality holds because it is assumed that $\eta \asymp n$ in the statement of Theorem~\ref{theorem: rate of convergence}. For the second term in \eqref{equation: Q conditioning bound} by arguments analogous to \eqref{equation: xi decomposition} we have that for any $\{ i_1, i_2, i_3 \}$ up to some absolute constant 
\begin{align}
& \left | \tilde{\mathbb{K}}_{X_{i_2}, Y_{i_2}, X_{i_3}, Y_{i_3}}^{\boldsymbol{Z}_{i_1}} - \check{\mathbb{K}}_{X_{i_2}, Y_{i_2}, X_{i_3}, Y_{i_3}}^{\boldsymbol{Z}_{i_1}}  \right | \label{equation: K triangle bound} \\
& \quad \lesssim \left | \frac{1}{k_n} \sum_{l \in \mathcal{N}(\boldsymbol{Z}_{i_1})} F_{X \mid \boldsymbol{Z = Z_l}} \left ( X_{i_2} \right ) - F_{X \mid \boldsymbol{Z = Z_{i_1}}} \left ( X_{i_2} \right )  \right | + \left | \frac{1}{k_n} \sum_{l \in \mathcal{N}(\boldsymbol{Z}_{i_1})} F_{Y \mid \boldsymbol{Z = Z_l}} \left ( Y_{i_2} \right ) - F_{Y \mid \boldsymbol{Z = Z_{i_1}}} \left ( Y_{i_2} \right )  \right | \nonumber \\
& \quad + \left | \frac{1}{k_n} \sum_{l \in \mathcal{N}(\boldsymbol{Z}_{i_1})} F_{X \mid \boldsymbol{Z = Z_l}} \left ( X_{i_3} \right ) - F_{X \mid \boldsymbol{Z = Z_{i_1}}} \left ( X_{i_3} \right )  \right | + \left | \frac{1}{k_n} \sum_{l \in \mathcal{N}(\boldsymbol{Z}_{i_1})} F_{Y \mid \boldsymbol{Z = Z_l}} \left ( Y_{i_3} \right ) - F_{Y \mid \boldsymbol{Z = Z_{i_1}}} \left ( Y_{i_3} \right )  \right |. \nonumber
\end{align}
Then, by Assumption~\ref{assumption: Lipschitz marginals} we have that
\begin{align}
& \mathbb{E} \left [ \left | \frac{1}{k_n} \sum_{l \in \mathcal{N} (\boldsymbol{Z}_{i_1})} F_{X \mid \boldsymbol{Z = Z_l}} \left ( X_{i_2} \right ) - F_{X \mid \boldsymbol{Z = Z_{i_1}}} \left ( X_{i_2} \right )  \right | \right ] \nonumber \\
& \quad \leq \frac{1}{k_n} \sum_{l \in \mathcal{N}(\boldsymbol{Z}_{i_1})} \mathbb{E} \left [ \sup_{x \in \mathbb{R}} \left | F_{X \mid \boldsymbol{Z = Z_l}} \left ( x \right ) - F_{X \mid \boldsymbol{Z = Z_{i_1}}} \left ( x \right ) \right | \right ] \lesssim \frac{1}{k_n} \sum_{l \in \mathcal{N} (\boldsymbol{Z}_{i_1})} \mathbb{E} \left [ \left \| \boldsymbol{Z}_{i_1} - \boldsymbol{Z}_l \right \|_2^\beta \right ]. 
\label{equation: Z distance  bound}
\end{align}
Since Assumption~\ref{assumption: Weibull tails} matches the conditions needed to invoke Lemma~\ref{lemma: knn expectation bound} we obtain 
\begin{align*}
\frac{1}{k_n} \sum_{l \in \mathcal{N}} \mathbb{E} \left [ \left \| \boldsymbol{Z}_{i_1} - \boldsymbol{Z}_l \right \|_2^\beta \right ]  \lesssim \left(\frac{k_n}{\eta}\right)^{\beta / d_0} + \frac{(\log \eta)^{\beta / \alpha}}{\eta^2} \lesssim \left(\frac{k_n}{n}\right)^{\beta / d_0} + \frac{(\log n)^{\beta / \alpha}}{n^2}, 
\end{align*}
and by the triangle inequality we obtain that that 
\begin{equation}
    \mathbb{E}[|\tilde{V}_{1,1} - \check{V}_{1,1}|] \lesssim \left(\frac{k_n}{n}\right)^{\beta / d_0} + \frac{(\log n)^{\beta / \alpha}}{n^2}. 
    \label{equation: V tilde V check bound}
\end{equation}
For the final term in \eqref{equation: K triangle bound} putting 
\begin{equation*}
    L_K^{(l)} = \left | \frac{1}{\eta k_n^2} \sum_{\{i_1, i_2, i_3 \} \in \mathcal{C}_l} \mathbb{E} \left [ \check{\mathbb{K}}_{X_{i_2}, Y_{i_2}, X_{i_3}, Y_{i_3}}^{\boldsymbol{Z}_{i_1}} \mid \mathcal{F}_Z \right ] - V_{1,1}^* \right |
\end{equation*}
we have that
\begin{align*}
& \mathbb{E} \left [ \left | \check{V}_{1,1} - V^*_{1,1}\right | \right ] \\
&\quad \lesssim \sum_{l \leq 1 \leq \chi (\mathcal{G}_n^2)} \mathbb{E} \left [ \left | \frac{1}{\eta k_n^2} \sum_{\{i_1, i_2, i_3 \} \in \mathcal{C}_l} \check{\mathbb{K}}_{X_{i_2}, Y_{i_2}, X_{i_3}, Y_{i_3}}^{\boldsymbol{Z}_{i_1}} - V_{1,1}^* \right | \right ] \\ 
& \quad = \sum_{l \leq 1 \leq \chi (\mathcal{G}_n^2)} \mathbb{E} \left [ \int_0^\infty \mathbb{P} \left ( \left | \frac{1}{\eta k_n^2} \sum_{\{i_1, i_2, i_3 \}} \check{\mathbb{K}}_{X_{i_2}, Y_{i_2}, X_{i_3}, Y_{i_3}}^{\boldsymbol{Z}_{i_1}} - \mathbb{E} \left [ \check{\mathbb{K}}_{X_{i_2}, Y_{i_2}, X_{i_3}, Y_{i_3}}^{\boldsymbol{Z}_{i_1}} \mid \mathcal{F}_Z \right ] \right | > t - L_K^{(l)} \mid \mathcal{F}_Z \right ) \mathrm{d}t \right ] \\
& \quad \leq \sum_{l \leq 1 \leq \chi (\mathcal{G}_n^2)} \Big \{  \mathbb{E} \left [ L_K^{(l)} \right ] + \leq \int_0^\infty \exp \left ( - C_2 \frac{t^2 \eta^2 k_n^4}{\left | \mathcal{C}_l \right |}\right ) \mathrm{d} t \Big \},
\end{align*}
where the final line holds, for some absolute $C_2 > 0$ since $\check{\mathbb{K}}$'s with $\{i_1, i_2, i_3\}$'s indexed on the same $\mathcal{C}_l$ are conditionally independent given $\mathcal{F}_Z$ and since $K$ is bounded the $\check{\mathbb{K}}$'s are likewise bounded hence one may apply Hoeffding's inequality. A bound of the same order as \eqref{equation: L bound} holds for the $L_K$. Hence, arguing as in \eqref{equation: V hat V tilde bound} we finally obtain that 
\begin{equation}
    \mathbb{E} \left [ \left | \check{V}_{1,1} - V^*_{1,1}\right | \right ] \lesssim \frac{k_n}{\sqrt{n}} + \frac{1}{k_n^4}.
    \label{equation: V check V* bound}
\end{equation}
Therefore combining \eqref{equation: V hat V tilde bound}, \eqref{equation: V tilde V check bound}, and \eqref{equation: V check V* bound} Markov's inequality readily yields that $|\hat{V}_{1,1} - V^*_{1,1}| = \mathcal{O}_\mathbb{P} ( k_n n^{-1/2} +  k_n^{-4} + \left(k_n / n\right)^{\beta / d_0} + (\log n)^{\beta / \alpha} n^{-2} ) $. Next we establish the rate of convergence of the second double sum in $\hat{T}_1$, namely
\begin{equation*}
\hat{V}_{1,2} \coloneq \frac{1}{2(n-\eta)} \sum_{i_1 =1}^{n-\eta} \frac{2}{k_n(k_n-1)} \sum_{\substack{i_2, i_3 \in \mathcal{N} (\boldsymbol{Z}_{i_1}') \\ i_2 < i_3}} \hat{\mathbb{K}}_{X_{i_2}, Y_{i_2}, X_{i_3}, Y_{i_3}}^{\boldsymbol{Z}_{i_1}'},
\end{equation*}
to the term 
\begin{equation*}
V^*_{1,2} \coloneq \frac{1}{2} \int \mathbb{E}_{U_1, U_2 \sim \mathcal{C}_{X,Y \mid \boldsymbol{Z} = \boldsymbol{z}}} \left [ K \left ( U_1, U_2 \right ) \right ] \mathrm{d} \mathbb{P}_{\boldsymbol{Z}'} (\boldsymbol{z}). 
\end{equation*}
Let $\tilde{V}_{1,2}$ and $\check{V}_{1,2}$ be defined analogously to $\check{V}_{1,1}$ and $\tilde{V}_{1,1}$. We once again have that $|\hat{V}_{1,2} - V^*_{1,2}| \leq |\hat{V}_{1,2} - \tilde{V}_{1,2}| + |\tilde{V} - \check{V}_{1,2}| + |\check{V}_{1,2} - V^*_{1,2}|$. The first and third terms may be handled identically to \eqref{equation: V hat V tilde bound} and \eqref{equation: V check V* bound}. The second term can be stochastically bounded using Corollary~\ref{corollary: Y knn expectation bound} in place of Lemma~\ref{lemma: knn expectation bound}. Therefore $|\hat{V}_{1,2} - V^*_{1,2}| = \mathcal{O}_\mathbb{P} ( k_n n^{-1/2} +  k_n^{-4} + \left(k_n / n\right)^{\beta / d_0} + (\log n)^{\beta / \alpha} n^{-2} ) $, and the same stochastic bound holds for the difference between $\hat{T}_1$ and the first term in \eqref{equation: Q expansion}. Identical arguments yield the same rate of convergence for $\hat{T}_2$ and $\hat{T}_3$ to the second and third terms in \eqref{equation: Q expansion}. This completes the proof. 
\end{proof}

\section{Data generating mechanisms in numerical experiments} \label{section: data generating mechanisms}

\paragraph{Common setup and notation (defaults).}
We consider a single changepoint at location $\tau$ with total length $n=\eta+n - \eta$.
Segment~1 corresponds to $t\le \tau$ and segment~2 corresponds to $t>\tau$.
Throughout, we denote the pre- and post-change distributions by $P_{\mathrm{pre}}$ and $P_{\mathrm{post}}$, and we label a scenario as a
\emph{mechanism change} when $P_{\mathrm{pre}}(Y\mid X,Z)\neq P_{\mathrm{post}}(Y\mid X,Z)$ (and as a \emph{null} otherwise). Unless otherwise stated with mutually independent noise terms and independence of noise from covariates:
\[
\varepsilon_t^x \perp (Z_t,\varepsilon_t^y),\qquad \varepsilon_t^y \perp (X_t,Z_t),\qquad
\{(\varepsilon_t^x,\varepsilon_t^y)\}_{t=1}^n\ \text{i.i.d.}
\]
Base marginals are standard Gaussian unless specified:
\[
Z_t\sim\mathcal N(0,1),\qquad \varepsilon_t^x\sim\mathcal N(0,\sigma_x^2),\qquad \varepsilon_t^y\sim\mathcal N(0,\sigma_y^2).
\]
When a scenario uses a driver generated from $Z$, we take the default linear confounding structure
\[
X_t = Z_t + \varepsilon_t^x,
\]
and when not otherwise specified we use the default linear response mean
\[
Y_t = 0.5\,Z_t + 0.6\,X_t + \varepsilon_t^y.
\]
Unless stated otherwise, coefficients and nonlinearities appear exactly as written in the scenario definitions. For mixture/tail scenarios that require moment control , we follow the stated procedure of
\emph{global} standardization on the post segment to enforce mean $0$ and standard deviation $\sigma_y$.
All expectations used for normalization are approximated by the corresponding empirical mean over the segment on which the
normalization is defined (matching the implementation).
In the i.i.d.\ synthetic suite, we set the driver noise and observation noise to
$\sigma_x \coloneqq 0.10$ and $\sigma_y \coloneqq 0.10$.

\subsection{Positive Controls (True Changes)} \label{section: positive controls}
\subsubsection{PMB — Positive Markov blanket and edge structure}
\paragraph{\texttt{PMB01\_EDGE\_ON\_LINEAR}} 
\textbf{Summary:} Edge appears: pre $Y\leftarrow Z$, post $Y\leftarrow X+Z$.
\textbf{Expected:} Change.
\textbf{Why:} This is the canonical Markov-blanket expansion scenario: $Y \indep X \mid Z$ holds in segment~1,
but the appearance of a direct $X\to Y$ effect in segment~2 breaks this conditional independence.
It serves as the simplest sanity-check that a method can detect a clear, linear mechanism change.
\begin{equation}
Y_t=
\begin{cases}
Z_t + \varepsilon^y_t, & t\le \tau,\\
X_t + Z_t + \varepsilon^y_t, & t>\tau,
\end{cases}
\qquad \varepsilon^y_t\sim\mathcal{N}(0,\sigma_y^2).
\end{equation}

\paragraph{\texttt{PMB02\_EDGE\_ON\_NONLIN}}
\textbf{Summary: } Edge appears with nonlinear baseline: pre $Y=f_Z(Z)$, post adds $\alpha f_X(X)$.
\textbf{Expected:} Change.
\textbf{Why:} This tests whether methods can detect a mechanism change in the conditional mean under nonlinear confounding
(from $Z$) and a nonlinear causal effect (from $X$). It is designed to be challenging for purely linear residualization,
since the post-change dependence enters through a nonlinear link.
Define
\begin{equation}
g(Z_t)=\sin(1.0\cdot Z_t),\qquad f(X_t)=0.6\,\tanh(1.0\cdot X_t).
\end{equation}
Then
\begin{equation}
Y_t=
\begin{cases}
\sin(Z_t)+\varepsilon^y_t, & t\le \tau,\\
\sin(Z_t)+0.6\,\tanh(X_t)+\varepsilon^y_t, & t>\tau,
\end{cases}
\qquad \varepsilon^y_t\sim\mathcal{N}(0,\sigma_y^2).
\end{equation}

\paragraph{\texttt{PMB03\_EDGE\_ON\_NEWDRV}}
\textbf{Summary: } Multi-$X$: new driver becomes causal post-cp (new parent in $X$-set).
\textbf{Expected:} Change.
\textbf{Why:} A new causal parent emerges from a set of correlated candidates, so detection requires distinguishing
genuine mechanism changes from correlation structure in the covariates.
This is especially informative for localization: ideally the method should attribute the change specifically to $X^{(2)}$.
Let $m\ge 2$, $Z_t\sim\mathcal{N}(0,1)$ and for $j=1,\dots,m$,
\begin{equation}
X^{(j)}_t = Z_t + \varepsilon^{x,j}_t,\quad \varepsilon^{x,j}_t\sim\mathcal{N}(0,\sigma_x^2).
\end{equation}
With $(\alpha,\beta)=(0.6,0.5)$:
\begin{equation}
Y_t=
\begin{cases}
0.6\,X^{(1)}_t + 0.5\,Z_t + \varepsilon^y_t, & t\le \tau,\\
0.6\,X^{(1)}_t + 0.6\,X^{(2)}_t + 0.5\,Z_t + \varepsilon^y_t, & t>\tau,
\end{cases}
\quad \varepsilon^y_t\sim\mathcal{N}(0,\sigma_y^2).
\end{equation}

\paragraph{\texttt{PMB04\_EDGE\_OFF}}
\textbf{Summary: } Edge removed: pre depends on $X$, post drops $X$.
\textbf{Expected:} Change.
\textbf{Why:} This is the symmetric counterpart to the canonical “edge turns on” scenario: $Y \not\!\perp X \mid Z$ holds in segment~1,
but the direct $X\to Y$ link \emph{disappears} in segment~2, restoring conditional independence.
It verifies that a detector is not implicitly one-sided and can identify the loss of a causal driver, not only its emergence.
Define
\begin{equation}
g(Z_t)=\sin(1.0\cdot Z_t),\qquad f(X_t)=\tanh(1.0\cdot X_t).
\end{equation}
Then
\begin{equation}
Y_t=
\begin{cases}
\sin(Z_t)+0.6\,\tanh(X_t)+\varepsilon^y_t, & t\le \tau,\\
\sin(Z_t)+\varepsilon^y_t, & t>\tau,
\end{cases}
\qquad \varepsilon^y_t\sim\mathcal{N}(0,\sigma_y^2).
\end{equation}

\paragraph{\texttt{PMB05\_EDGE\_ON\_COLLINEAR\_X}}
\textbf{Summary: } Edge-on under severe collinearity among $X$’s (tests driver identifiability).
\textbf{Expected:} Change.
\textbf{Why:} This tests Markov-blanket expansion under \emph{severe multicollinearity}: many covariates are strongly correlated through shared latent factors,
so a method must isolate the designated driver’s conditional effect rather than reacting to correlated “competitor” features.
It is especially informative for localization and robustness, since feature correlation can dilute conditional dependence signals and cause unstable attribution.
Let $F_t,Z_t\sim\mathcal N(0,1)$ independent, and for $j=1,\dots,m$ define correlated covariates
\begin{equation}
X^{(j)}_t = \rho_f F_t + \rho_z Z_t + \varepsilon^{x,j}_t,\qquad \varepsilon^{x,j}_t\sim\mathcal N(0,\sigma_x^2),
\end{equation}
with $\rho_f,\rho_z>0$ chosen so that $\mathrm{Corr}(X^{(j)},X^{(k)})$ is high for $j\neq k$.
Using $X_t \equiv X^{(1)}_t$ as the designated driver and a linear baseline in $Z$, we generate
\begin{equation}
Y_t=
\begin{cases}
0.5\,Z_t+\varepsilon^y_t, & t\le \tau,\\
0.5\,Z_t+0.6\,X_t+\varepsilon^y_t, & t>\tau,
\end{cases}
\qquad \varepsilon^y_t\sim\mathcal{N}(0,\sigma_y^2).
\end{equation}

\subsubsection{PEF — Effect changes (sign/strength/sparsity/modulation)}

\paragraph{\texttt{PEF01\_SIGN\_FLIP}} 
\textbf{Summary:} Sign flip: pre $+\alpha X$, post $-\alpha X$ (conditional on $Z$).
\textbf{Expected:} Change.
\textbf{Why:} This is a canonical structural break where the direction of the driver effect reverses while the Markov blanket stays fixed.
It checks whether a method detects qualitative changes in conditional dependence (including sign) rather than only changes in marginal scale.
\begin{equation}
Y_t=
\begin{cases}
0.5\,Z_t + 0.6\,X_t + \varepsilon^y_t, & t\le \tau,\\
0.5\,Z_t - 0.6\,X_t + \varepsilon^y_t, & t> \tau,
\end{cases}
\qquad \varepsilon^y_t\sim\mathcal{N}(0,\sigma_y^2).
\end{equation}

\paragraph{\texttt{PEF02\_STRENGTH\_SHIFT}} 
\textbf{Summary:} Magnitude change: $\alpha_1\rightarrow\alpha_2$ (same sign).
\textbf{Expected:} Change.
\textbf{Why:} The driver effect magnitude changes while preserving the functional form and conditioning structure, yielding a more subtle break than edge on/off.
This evaluates sensitivity to changes in effect size under otherwise stable confounding and noise.
\begin{equation}
Y_t=
\begin{cases}
0.5\,Z_t + 0.3\,X_t + \varepsilon^y_t, & t\le \tau,\\
0.5\,Z_t + 0.9\,X_t + \varepsilon^y_t, & t> \tau,
\end{cases}
\qquad \varepsilon^y_t\sim\mathcal{N}(0,\sigma_y^2).
\end{equation}

\paragraph{\texttt{PEF03\_STRENGTH\_SHIFT\_MULTIZ}} 
\textbf{Summary:} Strength change with multi-dim $Z$ (coefficients scaled to keep total confounding roughly constant).
\textbf{Expected:} Change.
\textbf{Why:} This examines effect-size breaks under high-dimensional confounding where conditioning becomes intrinsically harder as $d_z$ increases.
The $1/\sqrt{d_z}$ scaling keeps the overall confounding signal approximately constant, isolating dimensionality rather than simply increasing nuisance strength.
Let $Z_t\in\mathbb{R}^{d_z}$ with $Z_t\sim\mathcal{N}(0,I)$ and define
\begin{equation}
\gamma=\frac{0.5}{\sqrt{d_z}}\mathbf{1},\qquad \beta=\frac{0.4}{\sqrt{d_z}}\mathbf{1}.
\end{equation}
Set
\begin{equation}
X_t = Z_t^\top \gamma + \varepsilon^x_t,\quad \varepsilon^x_t\sim\mathcal{N}(0,\sigma_x^2),\qquad
g(Z_t)=Z_t^\top \beta.
\end{equation}
Then
\begin{equation}
Y_t=
\begin{cases}
g(Z_t) + 0.3\,X_t + \varepsilon^y_t, & t\le \tau,\\
g(Z_t) + 0.9\,X_t + \varepsilon^y_t, & t> \tau,
\end{cases}
\qquad \varepsilon^y_t\sim\mathcal{N}(0,\sigma_y^2).
\end{equation}

\paragraph{\texttt{PEF04\_STRENGTH\_SHIFT\_DISTRACTORS}} 
\textbf{Summary:} Strength change with distractor $Z$ dimensions (only a subset truly confounds).
\textbf{Expected:} Change.
\textbf{Why:} Only a subset of the conditioning variables are causally relevant while the remainder are irrelevant distractors, reflecting practical over-conditioning.
This checks whether methods degrade when conditioning includes nuisance covariates that should ideally be ignored.
Split $Z_t=(Z_t^{(s)},Z_t^{(n)})$ with $d_s$ signal dims and $d_n$ distractors, and let only $Z^{(s)}$ enter the DGP:
\begin{equation}
X_t=(Z_t^{(s)})^\top \gamma + \varepsilon^x_t,\qquad
g(Z_t)=(Z_t^{(s)})^\top \beta,
\end{equation}
with $\gamma=\tfrac{0.5}{\sqrt{d_s}}\mathbf{1}$ and $\beta=\tfrac{0.4}{\sqrt{d_s}}\mathbf{1}$.
Then
\begin{equation}
Y_t=
\begin{cases}
g(Z_t) + 0.3\,X_t + \varepsilon^y_t, & t\le \tau,\\
g(Z_t) + 0.9\,X_t + \varepsilon^y_t, & t> \tau,
\end{cases}
\qquad \varepsilon^y_t\sim\mathcal{N}(0,\sigma_y^2),
\end{equation}
while presenting the full $Z_t$ (signal+distractors) as the observed conditioning set.

\paragraph{\texttt{PEF05\_STRENGTH\_SHIFT\_MULTIZ\_NOSCALE}} 
\textbf{Summary:} Strength change with multi-dim $Z$ (unscaled coefficients; confounding grows with $d_z$).
\textbf{Expected:} Change.
\textbf{Why:} This is a harsher version of \texttt{PEF03}: without normalization, the confounding contribution grows with dimension, making the conditional signal harder to isolate.
It probes robustness in the realistic regime where nuisance strength increases with the size of the adjustment set.
Same as \texttt{PEF03} but use unscaled coefficients
\begin{equation}
\gamma=0.5\,\mathbf{1},\qquad \beta=0.4\,\mathbf{1},
\end{equation}
and keep the same effect shift $0.3\to 0.9$ in the equations above.

\paragraph{\texttt{PEF06\_SIGN\_FLIP\_GAUSS}}
\textbf{Summary:} Sign flip in an explicitly linear-Gaussian setting: pre $+\alpha X$, post $-\alpha X$.
\textbf{Expected:} Change.
\textbf{Why:} This is the analytically clean counterpart to \texttt{PEF01}, constructed so that $(X,Z,Y)$ are jointly Gaussian under linear links and Gaussian noise.
It serves as a sanity-check that methods behave well in the regime where partial-correlation / Gaussian-copula reasoning is exact.
\begin{equation}
Z_t\sim\mathcal{N}(0,1),\qquad X_t = Z_t + \varepsilon^x_t,\qquad \varepsilon^x_t\sim\mathcal{N}(0,\sigma_x^2),
\end{equation}
and
\begin{equation}
Y_t=
\begin{cases}
0.5\,Z_t + 0.6\,X_t + \varepsilon^y_t, & t\le \tau,\\
0.5\,Z_t - 0.6\,X_t + \varepsilon^y_t, & t> \tau,
\end{cases}
\qquad \varepsilon^y_t\sim\mathcal{N}(0,\sigma_y^2).
\end{equation}

\paragraph{\texttt{PEF07\_SPARSE\_CHANGE\_CORR\_X}} 
\textbf{Summary:} Sparse driver set changes: which $X_j$ affects $Y$ shifts post-cp.
\textbf{Expected:} Change.
\textbf{Why:} The identity of the active causal driver changes post-split while covariates remain strongly correlated, so detection requires distinguishing true parent-set changes from correlation structure.
It checks whether methods can localize the changepoint to the correct driver feature under multicollinearity and sparse support.
Let $F_t,Z_t\sim\mathcal N(0,1)$ independent and for $j=1,\dots,m$ define correlated covariates
\begin{equation}
X^{(j)}_t = 0.8\,F_t + 0.6\,Z_t + \varepsilon^{x,j}_t,\qquad \varepsilon^{x,j}_t\sim\mathcal N(0,\sigma_x^2),
\end{equation}
and take $X^{(1)}_t$ as the designated driver for reporting.
Choose two indices $j_{\mathrm{pre}}\neq j_{\mathrm{post}}$ (e.g.\ $1$ and $2$). Then
\begin{equation}
Y_t=
\begin{cases}
0.5\,Z_t + 0.6\,X^{(j_{\mathrm{pre}})}_t + \varepsilon^y_t, & t\le \tau,\\
0.5\,Z_t + 0.6\,X^{(j_{\mathrm{post}})}_t + \varepsilon^y_t, & t> \tau,
\end{cases}
\qquad \varepsilon^y_t\sim\mathcal{N}(0,\sigma_y^2).
\end{equation}

\paragraph{\texttt{PEF08\_POISSON\_SLOPE}}
\textbf{Summary:} Rate shift in a discrete stochastic multiplicative coefficient.
\textbf{Expected:} Change.
\textbf{Why:} The causal mechanism changes from $Y = Z_1 X + \varepsilon$ to $Y = Z_2 X + \varepsilon$, where $Z$ is a discrete Poisson random variable.
In segment~1, the coupling strength is drawn from $\text{Poisson}(\lambda=0.5)$. In segment~2, the rate increases to $\lambda=5$.
This introduces a discrete, integer-valued slope coefficient that varies per sample. The change manifests as an increase in the average effect size accompanied by a specific form of heteroskedasticity where the variance of $Y$ scales with the discrete intensity of the coupling.

\subsubsection{PNL — Nonlinear mechanism changes}

\paragraph{\texttt{PNL01\_SHAPE\_CHANGE}} 
\textbf{Summary:} Functional form change: driver link switches from monotone/saturating to oscillatory/non-monotone.
\textbf{Expected:} Change.
\textbf{Why:} The dependence of $Y$ on $X$ changes in \emph{shape} rather than just a coefficient, so $p(Y\mid X,Z)$ is different across segments even if marginal scales are similar.
This checks whether a method detects genuine nonlinear mechanism changes beyond correlation or variance shifts.
\begin{equation}
Y_t=
\begin{cases}
0.5\,Z_t + 1.0\,\tanh(1.5\,X_t) + \varepsilon^y_t, & t\le \tau,\\
0.5\,Z_t + 1.0\,\cos(2.0\,X_t) + \varepsilon^y_t, & t> \tau,
\end{cases}
\qquad \varepsilon^y_t\sim\mathcal{N}(0,\sigma_y^2).
\end{equation}

\paragraph{\texttt{PNL02\_INTERACTION\_EDGE\_ON\_XZ}} 
\textbf{Summary:} Interaction appears: post adds an $X\times h(Z)$ term (state-dependent effect).
\textbf{Expected:} Change.
\textbf{Why:} This introduces a mechanism change in which the effect of the driver becomes \emph{heterogeneous} across confounder states, not just a global slope shift.
It tests whether a detector is sensitive to changes in nonlinear conditional dependence geometry that are missed by purely additive models.
Define
\begin{equation}
h(Z_t)=\tanh(1.0\cdot Z_t).
\end{equation}
Then
\begin{equation}
Y_t=
\begin{cases}
0.5\,Z_t + 0.6\,X_t + \varepsilon^y_t, & t\le \tau,\\
0.5\,Z_t + 0.6\,X_t + 0.6\,X_t\,\tanh(Z_t) + \varepsilon^y_t, & t> \tau,
\end{cases}
\qquad \varepsilon^y_t\sim\mathcal{N}(0,\sigma_y^2).
\end{equation}

\paragraph{\texttt{PNL03\_TAIL\_GATED\_EDGE\_ON}} 
\textbf{Summary:} Tail-activated coupling: dependence turns on mainly for large $|X|$.
\textbf{Expected:} Change.
\textbf{Why:} The driver effect is concentrated in the tails of $X$, so global linear summaries may show little change even though the conditional mechanism differs.
This tests sensitivity to changepoints in tail dependence / state-dependent coupling, which are common in financial and risk regimes.
Let $a_t=|X_t|$ and let $\theta_q = Q_q(a_t)$ be the empirical $q$-quantile over the post segment (default $q=0.6$).
Define the smooth gate with sharpness $k=10$:
\begin{equation}
s(a_t)=\frac{1}{1+\exp\{-10(a_t-\theta_q)\}}.
\end{equation}
Then
\begin{equation}
Y_t=
\begin{cases}
0.5\,Z_t + \varepsilon^y_t, & t\le \tau,\\
0.5\,Z_t + 0.6\,s(|X_t|)\,X_t + \varepsilon^y_t, & t> \tau,
\end{cases}
\qquad \varepsilon^y_t\sim\mathcal{N}(0,\sigma_y^2).
\end{equation}

\paragraph{\texttt{PNL04\_QUADRATIC\_FLIP}}
\textbf{Summary:} Sign inversion of the quadratic dependence (Convex $\to$ Concave).
\textbf{Expected:} Change.
\textbf{Why:} The conditional variance is invariant. The change is isolated to the functional form of the conditional mean $\mathbb{E}[Y|X]$.
In segment~1, $Y$ depends on $X$ via a convex quadratic function:
\begin{equation}
Y_t \;=\; \beta X_t^2 + f_Z(Z_t) + \varepsilon_t, \quad \beta = 3.0.
\end{equation}
In segment~2, the dependence flips to concave ($\beta = -3.0$).
This creates a drastic topological change in the conditional mean structure (from a U-shaped valley to an inverted U-shaped hill) while maintaining identical marginal distributions for the covariates and error terms.

\paragraph{\texttt{PNL05\_HIGH\_FREQ\_SINE}}
\textbf{Summary:} Disappearance of a high-frequency sinusoidal dependence.
\textbf{Expected:} Change.
\textbf{Why:} This scenario is designed to challenge distance-based correlation metrics (e.g., dCor, HSIC).
In segment~1, $Y$ depends on $X$ via a high-frequency sine wave:
\begin{equation}
Y_t \;=\; \sin(8\pi X_t) + f_Z(Z_t) + \varepsilon_t.
\end{equation}
In segment~2, the dependence vanishes ($Y_t = f_Z(Z_t) + \varepsilon_t$).
For finite sample sizes ($N=400$), the high-frequency oscillations of segment~1 effectively average out in the integral calculations of global distance metrics, often yielding a distance correlation near zero (falsely suggesting independence). A local method, however, should distinguish the structured manifold of segment~1 from the unstructured noise of segment~2.

\subsubsection{PNM — Conditional noise-shape changes}

\paragraph{\texttt{PNM01\_COND\_SKEW}} 
\textbf{Summary:} Change in conditional skewness with invariant conditional mean and variance.
\textbf{Expected:} Change.
\textbf{Why:} The conditional mean and variance are strictly matched to the baseline (Gaussian) for all $X$. The change is isolated to the dependence of the noise \textit{skewness} on $X$.
In segment~1, $\varepsilon_t \sim \mathcal{N}(0, \sigma_0^2)$.
In segment~2, $\varepsilon_t$ follows a standardized Lognormal distribution where the shape parameter $\sigma_t$ depends on $|X_t|$. Specifically, $\sigma_t = 0.1 + 2.0|X_t|$.
We generate $\varepsilon_t$ via pointwise standardization:
\begin{equation}
\varepsilon_t \;=\; \sigma_0 \frac{L_t - \mu(\sigma_t)}{\sqrt{v(\sigma_t)}}, 
\quad \text{where } L_t \sim \text{Lognormal}(0, \sigma_t^2),
\end{equation}
with theoretical moments $\mu(s)=e^{s^2/2}$ and $v(s)=(e^{s^2}-1)e^{s^2}$.
For $X_t \approx 0$, the distribution is nearly Gaussian. As $|X_t|$ increases, the distribution becomes highly right-skewed while maintaining exact mean zero and unit variance. In this scenario we override the default mean model by setting $\alpha=0$ and $f_Z\equiv 0$, i.e.\ $Y_t=\varepsilon_t$, so any detected change is attributable purely to the $X$-dependent \emph{skewness} of the noise.
\paragraph{\texttt{PNM02\_COND\_TAILS}} 
\textbf{Summary:} Change in conditional kurtosis (Symmetric Lognormal) with invariant conditional mean, variance, and skewness.
\textbf{Expected:} Change.
\textbf{Why:} This scenario is the structural mirror of \texttt{PNM01}. It uses the same $X$-dependent shape mechanism ($\sigma_t = 0.1 + 2.0|X_t|$), but applies a random sign flip to strictly enforce symmetry ($\text{Skew}=0$). The change is isolated purely to \textit{conditional kurtosis}.
In segment~1, $\varepsilon_t \sim \mathcal{N}(0, \sigma_0^2)$.
In segment~2, $\varepsilon_t$ follows a standardized \textit{Symmetric Lognormal} distribution:
\begin{equation}
\varepsilon_t \;=\; \sigma_0 \frac{S_t \cdot e^{\sigma_t Z_t}}{\sqrt{v(\sigma_t)}}, \quad \text{where } Z_t \sim \mathcal{N}(0, 1),\; S_t \sim \text{Rademacher}(\pm 1),
\end{equation}
and $v(\sigma) = e^{2\sigma^2}$.
As $|X_t|$ increases, $\sigma_t$ grows, causing the kurtosis to explode (creating a sharp central peak and extremely heavy tails) while the distribution remains perfectly symmetric and variance-matched. In this scenario we also set $\alpha=0$ and $f_Z\equiv 0$ (so $Y_t=\varepsilon_t$), ensuring the only mechanism change comes from the $X$-dependent \emph{kurtosis/heavy tails} of the symmetric noise law.

\paragraph{\texttt{PNM03\_COND\_MIXTURE}} 
\textbf{Summary:} Conditional bimodality shift with global moment matching.
\textbf{Expected:} Change.
\textbf{Why:} The unconditional mean and variance are matched to the baseline. The change manifests as an $X$-dependent transition from unimodality to bimodality.
In segment~1, $\varepsilon_t \sim \mathcal{N}(0, \sigma_0^2)$.
In segment~2, we define a logistic mixing weight $w_t = \sigma(\eta X_t)$ with slope $\eta=5.0$ and generate a raw mixture of two Gaussians separated by $6\sigma$:
\begin{equation}
U_t \sim w_t \mathcal{N}(-3, 0.5^2) + (1-w_t) \mathcal{N}(+3, 0.5^2).
\end{equation}
The entire segment $\{U_t\}$ is then globally standardized to have sample mean $0$ and standard deviation $\sigma_0$.
Unlike the previous two scenarios, this global standardization preserves the conditional mean structure $\mathbb{E}[Y|X]$, allowing the distribution to smoothly morph from a left-mode to a right-mode as $X$ varies, creating a strong topological signal.

\subsubsection{PVR — Volatility / variance regime changes}
\paragraph{\texttt{PVR01\_COND\_HETSKED}} 
\textbf{Summary:} Change from homoskedasticity to symmetric quadratic heteroskedasticity with exact global variance matching.
\textbf{Expected:} Change.
\textbf{Why:} The conditional mean is invariant. The change isolates the conditional variance structure: segment~1 has constant variance, while segment~2 has variance dependent on $X^2$.
In segment~1, $\varepsilon_t \sim \mathcal{N}(0, \sigma_0^2)$.
In segment~2, $\varepsilon_t$ follows a heteroskedastic Gaussian distribution driven by a quadratic variance function:
\begin{equation}
\varepsilon_t \;=\; \sigma_0 \cdot \xi_t \cdot \sqrt{(1-\theta) + \theta X_t^2}, \quad \text{where } \xi_t \sim \mathcal{N}(0,1).
\end{equation}
We set $\theta=0.99$. This specific parametrization ensures that the global variance remains invariant: since $\mathbb{E}[X_t^2]=1$ (standardized inputs), the expected variance is $\sigma_0^2 [(1-\theta) + \theta(1)] = \sigma_0^2$.
However, the \textit{conditional} variance varies drastically: at $X_t \approx 0$, the noise is silenced ($\text{Var} \approx 0.01\sigma_0^2$), whereas at $|X_t| \approx 2$, the noise explodes ($\text{Var} \approx 4.0\sigma_0^2$), creating a distinct U-shaped volatility profile.


\paragraph{\texttt{PVR02\_VOL\_CLUSTER}} 
\textbf{Summary:} Volatility clustering regime: post residuals exhibit ARCH(1)-type dynamics.
\textbf{Expected:} Change.
\textbf{Why:} The conditional mean mechanism is unchanged across segments, but the \emph{second-order} structure changes:
segment~2 violates i.i.d.\ residual assumptions by introducing time-varying conditional variance.
This checks whether detectors can identify regime changes driven by volatility clustering (common in financial time series)
without requiring a shift in the instantaneous conditional mean.

Let $\xi_t\sim\mathcal{N}(0,1)$ and define the residual process by
\begin{equation}
\varepsilon_t =
\begin{cases}
0.10\,\xi_t, & t\le \tau,\\
\sigma_t\,\xi_t, & t> \tau,
\end{cases}
\qquad
\sigma_t^2 = 0.2 + 0.6\,\varepsilon_{t-1}^2 \quad (t>\tau).
\end{equation}
Then for all $t$,
\begin{equation}
Y_t = 0.5\,Z_t + 0.6\,X_t + \varepsilon_t .
\end{equation}

\paragraph{\texttt{PVR03\_GLOBAL\_NOISE\_SCALE}} 
\textbf{Summary:} Global variance / SNR regime: noise scale increases post-split.
\textbf{Expected:} Change.
\textbf{Why:} Although the structural mean term is unchanged, a large increase in noise variance reduces the effective signal-to-noise ratio and weakens conditional dependence.
Under a copula/partial-correlation view, the conditional dependence parameter changes with $\sigma$, so this constitutes a genuine regime change in $p(Y\mid X,Z)$.
\begin{equation}
Y_t = 0.5\,Z_t + 0.6\,X_t + \varepsilon_t,
\quad
\varepsilon_t\sim
\begin{cases}
\mathcal{N}(0,0.10^2), & t\le \tau,\\
\mathcal{N}(0,(0.50)^2), & t>\tau.
\end{cases}
\end{equation}

\subsubsection{PSM — Smooth / gradual regime transitions}

\paragraph{\texttt{PSM01\_SMOOTH\_TRANSITION}} 
\textbf{Summary:} Gradual edge-on: the $X\to Y$ effect is mixed in smoothly around the changepoint index $\tau$ over a window of width $W$.
\textbf{Expected:} Change.
\textbf{Why:} Many real regime shifts are not instantaneous; dependence often ramps up over time rather than appearing at a single index.
This scenario checks whether a detector can still identify a mechanism change when the transition is smooth and the “changepoint” is diffuse.

Let $Z_t\sim\mathcal{N}(0,1)$ and generate $X_t$ from $Z_t$ via the standard confounded design in the suite.
Define the nonlinear baseline and effect (default choices in the generator)
\begin{equation}
f_Z(Z_t)=\sin(1.0\cdot Z_t),
\qquad
f_X(X_t)=0.6\,\tanh(1.0\cdot X_t).
\end{equation}
Let $W\ge 1$ be the transition width and set $\kappa = \max(W/6,1)$.
Define the smooth mixing weight (logistic ramp)
\begin{equation}
w_t \;=\; \frac{1}{1+\exp\left\{-(t-\tau)/\kappa\right\}}
\;=\;
\frac{1}{1+\exp\left\{-6(t-\tau)/W\right\}}
\quad (\text{when } W\ge 6).
\end{equation}
Then the observation model is
\begin{equation}
Y_t \;=\; f_Z(Z_t)\;+\; w_t\, f_X(X_t)\;+\;\varepsilon^y_t,
\qquad \varepsilon^y_t\sim\mathcal{N}(0,\sigma_y^2).
\end{equation}

\subsection{Negative Controls (Nulls)} \label{section: negative controls}
\subsubsection{NCL — Baseline nulls}

\paragraph{\texttt{NCL01\_BASE\_NULL}} 
\textbf{Summary:} Stationary null: both segments $Y\leftarrow f(Z)$ (no $X$ effect).
\textbf{Expected:} Null.
\textbf{Why:} This is the fundamental Type-I error control: the full data-generating process is stationary across the split and $Y\indep X\mid Z$ holds in both segments.
It checks that the end-to-end pipeline (including calibration) does not spuriously detect changepoints under complete stability.
\begin{equation}
Y_t \;=\; 0.5\,Z_t + \varepsilon^y_t,\qquad \varepsilon^y_t\sim\mathcal{N}(0,\sigma_y^2),
\quad \text{for both } t\le \tau \text{ and } t>\tau,
\end{equation}
with the default confounded driver $X_t = Z_t + \varepsilon^x_t$ generated but not entering $Y$.

\paragraph{\texttt{NCL02\_MULTIX\_ALIGN}} 
\textbf{Summary:} Multi-$X$ null aligned to PMB03 world: correlated candidates but no mechanism change.
\textbf{Expected:} Null.
\textbf{Why:} This mirrors the multivariate-candidate construction of \texttt{PMB03\_EDGE\_ON\_NEWDRV} while removing the changepoint, so any rejection indicates sensitivity to correlation structure rather than a true parent-set change.
It is a key robustness check for implementations that handle multi-feature $X$ and a designated driver index.
Let $m\ge 2$, $Z_t\sim\mathcal{N}(0,1)$ and for $j=1,\dots,m$,
\begin{equation}
X^{(j)}_t = Z_t + \varepsilon^{x,j}_t,\quad \varepsilon^{x,j}_t\sim\mathcal{N}(0,\sigma_x^2).
\end{equation}
Then, for all $t$,
\begin{equation}
Y_t = 0.6\,X^{(1)}_t + 0.5\,Z_t + \varepsilon^y_t,
\qquad \varepsilon^y_t\sim\mathcal{N}(0,\sigma_y^2),
\end{equation}
with no change across the split (and the designated driver set consistently with the PMB03-aligned world).

\subsubsection{NIV — Invariance checks}

\paragraph{\texttt{NIV01\_GLOBAL\_RESCALE}} 
\textbf{Summary:} Unit change: multiply $(X,Y,Z)$ post by a constant; copula invariant.
\textbf{Expected:} Null.
\textbf{Why:} This simulates a pure change of measurement units, which should not create a changepoint in the underlying conditional dependence structure.
It checks that methods are invariant to global positive rescaling and do not spuriously reject due to normalization or scale sensitivity.
Let $(X_t,Y_t,Z_t)$ be generated under the default stationary mechanism, and apply a post-split rescaling by $\lambda=10$:
\begin{equation}
(X_t,Y_t,Z_t)\leftarrow
\begin{cases}
(X_t,Y_t,Z_t), & t\le \tau,\\
(10\,X_t,\;10\,Y_t,\;10\,Z_t), & t>\tau.
\end{cases}
\end{equation}

\paragraph{\texttt{NIV02\_Y\_MONO\_TRANSFORM}}
\textbf{Summary:} Strictly monotone transform of $Y$ post-cp (rank-invariance sanity).
\textbf{Expected:} Null.
\textbf{Why:} Copula-based (rank-based) dependence is invariant to strictly increasing transformations of $Y$, so the conditional copula $C_{Y,X\mid Z}$ should be unchanged.
This checks whether a method is truly targeting copula/conditional dependence rather than being sensitive to Euclidean scaling or marginal shapes of $Y$.
Let $(X_t,Y_t,Z_t)$ be generated under a fixed stationary mechanism, and for $t>\tau$ apply a strictly increasing map $\phi$:
\begin{equation}
Y_t \leftarrow
\begin{cases}
Y_t, & t\le \tau,\\
\phi(Y_t), & t>\tau,
\end{cases}
\qquad \text{with $\phi$ strictly increasing (e.g.\ $\phi(y)=\mathrm{asinh}(y)$).}
\end{equation}

\paragraph{\texttt{NIV03\_X\_MONO\_GIVEN\_Z}}
\textbf{Summary:} Strictly monotone transform of $X$ conditional on $Z$ post-cp.
\textbf{Expected:} Null.
\textbf{Why:} A strictly increasing transform of $X$ (within each confounder state) preserves the conditional rank structure and hence should not change the conditional copula $C_{Y,X\mid Z}$.
This checks whether the method depends on the conditional order information of $X$ given $Z$, rather than the raw scale or marginal distribution of $X$.
Let $(X_t,Y_t,Z_t)$ be generated under a fixed stationary mechanism, and for $t>\tau$ apply a $Z$-dependent monotone map
\begin{equation}
X_t \leftarrow
\begin{cases}
X_t, & t\le \tau,\\
a(Z_t)\,\phi(X_t) + b(Z_t), & t>\tau,
\end{cases}
\qquad \text{where $a(Z_t)>0$ and $\phi$ is strictly increasing.}
\end{equation}

\subsubsection{NMD — Marginal drift controls}

\paragraph{\texttt{NMD01\_Z\_LOC\_SCALE}} 
\textbf{Summary:} Shift/scale $Z$ marginals; keep the conditional structure stable.
\textbf{Expected:} Null.
\textbf{Why:} The confounder distribution changes across segments, altering $p(Z)$ (and hence potentially $p(X,Z)$), but the structural equations for $X\mid Z$ and $Y\mid Z$ remain unchanged.
This tests whether methods incorrectly interpret marginal confounder drift as a mechanism change when the conditioning set includes $Z$.
\begin{equation}
Z_t\sim
\begin{cases}
\mathcal{N}(0,1), & t\le \tau,\\
\mathcal{N}(0.5,\,1.6^2), & t>\tau,
\end{cases}
\qquad
X_t = Z_t+\varepsilon^x_t,\;\; \varepsilon^x_t\sim\mathcal{N}(0,\sigma_x^2),
\end{equation}
and in both segments
\begin{equation}
Y_t = 0.5\,Z_t + \varepsilon^y_t,\qquad \varepsilon^y_t\sim\mathcal{N}(0,\sigma_y^2).
\end{equation}

\paragraph{\texttt{NMD02\_X\_MEAN\_SHIFT}} 
\textbf{Summary:} Shift $X$ marginal mean across segments (constructed null).
\textbf{Expected:} Null.
\textbf{Why:} The marginal distribution of the driver changes substantially, but the response mechanism is constructed so that $Y\indep X\mid Z$ holds in both segments.
This checks whether a method spuriously rejects under covariate shift when the target conditional mechanism is invariant.
Let $Z_t\sim\mathcal{N}(0,1)$ and define
\begin{equation}
X_t =
\begin{cases}
Z_t + \varepsilon^x_t, & t\le \tau,\\
Z_t + 2.0 + \varepsilon^x_t, & t>\tau,
\end{cases}
\qquad \varepsilon^x_t\sim\mathcal{N}(0,\sigma_x^2),
\end{equation}
and generate $Y$ with no driver effect in both segments:
\begin{equation}
Y_t = 0.5\,Z_t + \varepsilon^y_t,\qquad \varepsilon^y_t\sim\mathcal{N}(0,\sigma_y^2).
\end{equation}

\paragraph{\texttt{NMD03\_Y\_TREND}} 
\textbf{Summary:} Add a deterministic linear trend to $Y$ (mechanism constant).
\textbf{Expected:} Null.
\textbf{Why:} A deterministic time trend changes the marginal distribution of $Y$ across segments, but the stochastic mechanism linking $Y$ to $(X,Z)$ is unchanged.
This tests whether a detector confuses deterministic nonstationarity with a changepoint in the conditional mechanism.
Let $t=0,1,\dots,n-1$ and set
\begin{equation}
Y_t = 0.5\,X_t + 0.02\,t + \varepsilon^y_t,\qquad \varepsilon^y_t\sim\mathcal{N}(0,\sigma_y^2),
\end{equation}
with the same equation applying for both $t\le \tau$ and $t>\tau$ (no structural change).

\paragraph{\texttt{NMD04\_Y\_SEASON}} 
\textbf{Summary:} Add strong deterministic seasonality to $Y$ (mechanism constant).
\textbf{Expected:} Null.
\textbf{Why:} A periodic component makes segment-wise marginals differ depending on phase coverage, yet the underlying conditional mechanism is unchanged.
This checks robustness to deterministic seasonal structure that can induce apparent distribution shifts without any causal break.
With period $50$ and amplitude $2.0$, define
\begin{equation}
Y_t = 0.5\,X_t + 2.0\,\sin\!\left(\frac{2\pi t}{50}\right) + \varepsilon^y_t,
\qquad \varepsilon^y_t\sim\mathcal{N}(0,\sigma_y^2),
\end{equation}
with no change across the split.

\subsubsection{NNS — Noise-law / tail-shape controls}

\paragraph{\texttt{NNS01\_NOISE\_LAW\_SHIFT}} 
\textbf{Summary:} Gaussian $\to$ Laplace noise post (variance matched); mechanism unchanged.
\textbf{Expected:} Null.
\textbf{Why:} The marginal noise distribution changes shape while variance is matched, but the conditional mean mechanism is identical in both segments.
This tests whether a method overreacts to noise-law differences when the target is stability of $p(Y\mid X,Z)$ under variance control.
Let $X_t = Z_t+\varepsilon^x_t$ with $\varepsilon^x_t\sim\mathcal{N}(0,\sigma_x^2)$ and define
\begin{equation}
Y_t = 0.5\,Z_t + \varepsilon^y_t,
\end{equation}
where
\begin{equation}
\varepsilon^y_t\sim
\begin{cases}
\mathcal{N}(0,\sigma_y^2), & t\le \tau,\\
\mathrm{Laplace}(0,b), & t>\tau,
\end{cases}
\qquad b=\frac{\sigma_y}{\sqrt{2}}\;\;(\mathrm{Var}=\sigma_y^2).
\end{equation}

\paragraph{\texttt{NNS02\_TAILS\_SHIFT}} 
\textbf{Summary:} Gaussian $\to$ Student-$t$ noise post (variance matched); mechanism unchanged.
\textbf{Expected:} Null.
\textbf{Why:} The residual distribution becomes heavy-tailed post-split while its variance is matched, creating a tail-risk nuisance without altering the driver mechanism.
This tests whether methods confuse tail-shape differences with true mechanism changes when first/second moments are controlled.
Define a stationary linear mechanism
\begin{equation}
Y_t = 0.5\,X_t + \varepsilon_t,
\end{equation}
with
\begin{equation}
\varepsilon_t\sim
\begin{cases}
\mathcal{N}(0,\sigma_y^2), & t\le \tau,\\
c\,t_{3}, & t>\tau,
\end{cases}
\qquad c=\frac{\sigma_y}{\sqrt{\mathrm{Var}(t_3)}}=\sigma_y\sqrt{\frac{1}{3}},
\end{equation}
since $\mathrm{Var}(t_\nu)=\nu/(\nu-2)$ for $\nu>2$ (so $\mathrm{Var}(t_3)=3$), yielding $\mathrm{sd}(c\,t_3)=\sigma_y$.

\subsubsection{NCF — Covariate/confounding drift controls}


\paragraph{\texttt{NCF01\_XGZ\_DRIFT}} 
\textbf{Summary:} Severe location and scale drift in $X\mid Z$ while $Y\mid Z$ is invariant.
\textbf{Expected:} Null.
\textbf{Why:} This induces a strong covariate shift in the driver distribution $p(X\mid Z)$ without altering the response mechanism, testing whether a detector spuriously flags a changepoint due to feature drift.
Let $Z_t\sim\mathcal{N}(0,1)$ and $\varepsilon^x_t\sim\mathcal{N}(0,\sigma_x^2)$. We define the confounded driver as:
\begin{equation}
X_t =
\begin{cases}
\sin(Z_t) + \varepsilon^x_t, & t\le \tau,\\
5.0 \cdot \big(\sin(Z_t) + \varepsilon^x_t\big) + 10.0, & t>\tau.
\end{cases}
\end{equation}
In both segments, the outcome is generated independently of $X$ as $Y_t = 0.5\,Z_t + \varepsilon^y_t$, ensuring that $Y\indep X\mid Z$ holds throughout. 

\paragraph{\texttt{NCF02\_Z\_COV\_SHIFT}} 
\textbf{Summary:} Change $\mathrm{Cov}(Z)$ across segments with $Y\mid X,Z$ unchanged.
\textbf{Expected:} Null.
\textbf{Why:} The dependence structure within the confounder vector drifts, altering $p(Z)$ and $p(X,Z)$, but the structural equation for $Y$ given $(X,Z)$ remains the same.
This tests whether a method mistakes shifts in feature dependence (within $Z$) for a mechanism change when conditioning on the full observed $Z$.
Let $Z_t\in\mathbb{R}^{d_z}$ be Gaussian with equicorrelation covariance
\begin{equation}
Z_t \sim \mathcal{N}(0,\Sigma_\rho),\qquad (\Sigma_\rho)_{ij}=
\begin{cases}
1, & i=j,\\
\rho, & i\neq j,
\end{cases}
\end{equation}
with
\begin{equation}
\rho =
\begin{cases}
\rho_{\mathrm{pre}}, & t\le \tau,\\
\rho_{\mathrm{post}}, & t>\tau.
\end{cases}
\end{equation}
Generate the driver from the (unchanged) confounding structure
\begin{equation}
X_t = \frac{1}{\sqrt{d_z}}\mathbf{1}^\top Z_t + \varepsilon^x_t,\qquad \varepsilon^x_t\sim\mathcal{N}(0,\sigma_x^2),
\end{equation}
and keep the response mechanism fixed for all $t$:
\begin{equation}
Y_t = \frac{0.5}{\sqrt{d_z}}\mathbf{1}^\top Z_t + \varepsilon^y_t,\qquad \varepsilon^y_t\sim\mathcal{N}(0,\sigma_y^2).
\end{equation}

\paragraph{\texttt{NCF03\_XGZ\_DRIFT\_COPULA}} 
\textbf{Summary:} Copula drift in $X\mid Z$ with perfectly matched marginals; $\alpha=0$ (hard null).
\textbf{Expected:} Null.
\textbf{Why:} The conditional dependence (copula) between $X$ and $Z$ is significantly altered post-split, but the marginal distributions remain strictly invariant. $Y$ is generated exclusively from $Z$, ensuring $p(Y\mid X,Z)$ is unchanged and $Y \indep X \mid Z$ holds globally.
Let $Z_t\sim\mathcal{N}(0,1)$ and $\varepsilon^x_t\sim\mathcal{N}(0,1)$. We generate $X_t$ using a shifting Gaussian copula structure:
\begin{equation}
X_t =
\begin{cases}
\rho_1 Z_t + \sqrt{1-\rho_1^2} \,\varepsilon^x_t, & t\le \tau,\\
\rho_2 Z_t + \sqrt{1-\rho_2^2} \,\varepsilon^x_t, & t>\tau.
\end{cases}
\end{equation}
Setting $\rho_1 = 0.2$ and $\rho_2 = 0.8$ induces a severe drift in the joint distribution $p(X, Z)$, while mathematically guaranteeing that the marginal distribution $p(X)$ remains standard normal $\mathcal{N}(0,1)$ throughout. The outcome is generated independently of $X$ as $Y_t = 0.5\,Z_t + \varepsilon^y_t$. 
This constitutes a stringent null test: it isolates the algorithm's sensitivity to structural shifts in the covariate space, penalizing methods that suffer from intra-neighborhood leakage unless properly calibrated.

\paragraph{\texttt{NCF04\_FZ\_DRIFT\_ONLY}} 
\textbf{Summary:} Drift in $f(Z)$ only (constructed null; no $X\to Y$ change).
\textbf{Expected:} Null.
\textbf{Why:} The deterministic $Z$-only component in $Y$ changes across segments, but the $X\to Y$ mechanism is unchanged (and can be set to zero), so $p(Y\mid X,Z)$ is invariant with respect to the driver.
This checks whether a method properly conditions out $Z$ rather than reacting to changes in the marginal behavior of $Y$ induced by $Z$.
Let $Z_t\sim\mathcal{N}(0,1)$ and define
\begin{equation}
Y_t =
\begin{cases}
\sin(Z_t)+\varepsilon^y_t, & t\le \tau,\\
\tanh(Z_t)+\varepsilon^y_t, & t>\tau,
\end{cases}
\qquad \varepsilon^y_t\sim\mathcal{N}(0,\sigma_y^2),
\end{equation}
with the driver generated as usual (e.g.\ $X_t=Z_t+\varepsilon^x_t$) but not entering $Y$.

\paragraph{\texttt{NCF05\_DISCRETE\_Z}} 
\textbf{Summary:} Switch $Z$ to a discrete/mixture-type confounder with drifting mixture weights.
\textbf{Expected:} Null.
\textbf{Why:} The distribution of the confounder changes in a discrete/mixture manner, altering $p(Z)$ and thus the marginal distributions of $(X,Y)$, but the conditional mechanism given $Z$ is unchanged.
This checks whether methods remain calibrated when $Z$ is non-Gaussian and the confounder composition drifts across segments.
Let $Z_t\in\{0,1\}$ with
\begin{equation}
\mathbb{P}(Z_t=1)=
\begin{cases}
p_{\mathrm{pre}}, & t\le \tau,\\
p_{\mathrm{post}}, & t>\tau,
\end{cases}
\end{equation}
and generate (for all $t$) a confounded driver and response with fixed structural equations
\begin{equation}
X_t = 1.0\,Z_t + \varepsilon^x_t,\qquad
Y_t = 0.5\,Z_t + \varepsilon^y_t,\qquad
\varepsilon^x_t\sim\mathcal{N}(0,\sigma_x^2),\;\varepsilon^y_t\sim\mathcal{N}(0,\sigma_y^2).
\end{equation}

\subsubsection{NDR — Implementation robustness}

\paragraph{\texttt{NDR01\_FEATURE\_PERMUTE}} 
\textbf{Summary:} Permute non-driver feature columns at $cp$ (indexing/pipeline robustness).
\textbf{Expected:} Null.
\textbf{Why:} The underlying data-generating mechanism is unchanged up to column order, so any detection indicates an implementation bug or an unintended dependence on feature ordering.
This is a simple but important check that the pipeline uses the designated driver index correctly and is robust to irrelevant feature permutations.
Let $X_t\in\mathbb{R}^m$ be a multi-feature covariate vector with designated driver in column $0$.
Draw $(X_t,Y_t,Z_t)$ from a stationary mechanism for all $t$, but for $t>\tau$ apply a permutation $\pi$ to the non-driver columns:
\begin{equation}
X^{(0)}_t \leftarrow X^{(0)}_t,\qquad
(X^{(1)}_t,\dots,X^{(m-1)}_t)\leftarrow (X^{(\pi(1))}_t,\dots,X^{(\pi(m-1))}_t),
\end{equation}
leaving $Y_t$ and $Z_t$ unchanged.

\subsubsection{NPO — Partial observability / latent confounding}

\paragraph{\texttt{NPO01\_LATENT\_Z\_STABLE}} 
\textbf{Summary:} Only $Z_{\mathrm{obs}}$ is provided; $Z_{\mathrm{hid}}$ is unobserved; regime stable.
\textbf{Expected:} Null.
\textbf{Why:} This tests robustness to partial observability: the true confounder is multidimensional but the method only receives a subset, creating residual confounding that can inflate dependence without any changepoint.
Because the regime is stable across segments, any detection reflects sensitivity to unobserved confounding rather than a true mechanism change.
Let the true confounder be $Z_t=(Z^{\mathrm{obs}}_t, Z^{\mathrm{hid}}_t)$ with independent components
\begin{equation}
Z^{\mathrm{obs}}_t\sim\mathcal{N}(0,1),\qquad Z^{\mathrm{hid}}_t\sim\mathcal{N}(0,1),
\end{equation}
and generate
\begin{equation}
X_t = 1.0\,Z^{\mathrm{obs}}_t + 1.0\,Z^{\mathrm{hid}}_t + \varepsilon^x_t,\qquad \varepsilon^x_t\sim\mathcal{N}(0,\sigma_x^2),
\end{equation}
\begin{equation}
Y_t = 0.5\,Z^{\mathrm{obs}}_t + 0.5\,Z^{\mathrm{hid}}_t + \varepsilon^y_t,\qquad \varepsilon^y_t\sim\mathcal{N}(0,\sigma_y^2),
\end{equation}
for all $t$ (no change at $\tau$). The method is only given $Z^{\mathrm{obs}}_t$ as the observed conditioning variable.

\section{Computational resources}
\label{section:compute-resources}

All experiments are run on CPU only; no GPU acceleration is used. The final
reproduction runs were launched on machines with Intel(R)
Xeon(R) Gold 6252N CPUs at 2.30GHz, with 96 CPU cores and 252 GB of memory available. Wall-clock times are measured by the code around each experiment runner or major compute stage and are summarised automatically. 

For the synthetic benchmark, the logs include both the total benchmark time and per-stage timings for permutation null sampling and AUROC Monte Carlo computation. For the real-data experiments, the logs include the data-pull time and changepoint-scan time, including the permutation p-value computations at selected peaks. For appendix sensitivity
experiments, each script logs its full wall-clock runtime. Parallel execution is
utilised at scan-level, or permutation-level parallelism where applicable.

The main synthetic benchmark uses 49 scenarios, seven methods, 50 independent
replicates, 499 permutation samples per p-value, and 500 Monte Carlo replicates
for AUROC estimation at sample size $n=400$. The real-data analysis uses the
proposed 499 permutations and the window sizes reported in the corresponding experiment descriptions.

\begin{table*}[!t]
\centering
\small
\setlength{\tabcolsep}{4pt}
\renewcommand{\arraystretch}{1.12}
\begin{tabularx}{\textwidth}{@{} p{0.25\textwidth} X p{0.13\textwidth} p{0.16\textwidth} @{}}
\toprule
\textbf{Experiment} & \textbf{Workload} & \textbf{Wall-clock} & \textbf{Parallelism} \\
\midrule
Synthetic benchmark & 49 scenarios, 7 methods, $R=50$ independent replicates, $B=499$ permutations, $M=500$ AUROC Monte Carlo replicates, and $n=400$ observations & 4.18 h & 96 workers \\
Real-world changepoint scans & 6 scans using $\hat{Q}_{X,Y,Z}^{(\Lambda)}$, $B=499$ permutations, and p-values computed for up to 20 selected peaks per scan & 4.3 min & 96 scan workers; 1 permutation worker \\
\bottomrule
\end{tabularx}
\caption{Compute resources used for the reported experiments. All experiments were run on CPU only on Intel(R) Xeon(R) Gold 6252N CPUs at 2.30GHz, with 96 CPU cores available and 251.5 GB RAM; no GPU acceleration was used. The reported wall-clock timings cover the corresponding experiment runner and include data generation or loading, statistic computation, and permutation or Monte Carlo loops where applicable.}
\label{tab:compute_resources}
\end{table*}

\end{document}